\documentclass[12pt]{article}
\usepackage[cp1250]{inputenc}
\usepackage[verbose,a4paper,tmargin=0.5in,lmargin=1in,rmargin=1in]{geometry}
\usepackage[pdftex]{graphicx}
\usepackage{amsfonts}
\usepackage{indentfirst}
\usepackage{latexsym}
\usepackage{amsmath}
\usepackage{amsthm}
\usepackage{amssymb}
\usepackage{psfrag}
\usepackage[mathscr]{eucal} 
\usepackage{color}
\usepackage{cite}
\usepackage{longtable}
\usepackage{verbatim}
\usepackage{enumitem}
\usepackage[font=normalsize]{caption}
\usepackage{bigints}

\usepackage{graphicx}
\usepackage[all]{xy}
\usepackage{float}
\usepackage{afterpage}

\usepackage[table, dvipsnames]{xcolor}

\newtheorem{Twierdzenie}{Theorem}[section]
\newtheorem{Wniosek}{Corollary}[section]
\newtheorem{Definicja}{Definition}[section]

\newfloat{Scheme}{tbp}{muj}

\renewcommand{\theequation}{\arabic{section}.\arabic{equation}}

\definecolor{light_gray_1}{gray}{0.85}
\definecolor{light_gray_2}{gray}{0.95}

\title{Hyperheavenly spaces and their application in Walker and para-Kähler geometries: part I}

\author{$\textrm{Adam Chudecki}^{*}$}

\begin{document}

\maketitle

$*$ Center of Mathematics and Physics, Lodz University of Technology, 
\newline
$\ \ \ \ \ $ Al. Politechniki 11, 90-924 Łódź, Poland, adam.chudecki@p.lodz.pl
\newline
\newline
\newline
\textbf{Abstract}. 
\newline
Spaces equipped with congruences of null strings are considered. A special attention is paid to the spaces which belong to the two-sided Walker class and para-Kähler class. Properties of an intersection of self-dual and anti-self-dual congruences of null strings are used as an additional criterion for a classification of such spaces. Finally, a few examples of para-Kähler and para-Kähler-Einstein spaces are presented.
%\newline
%\newline
%\textbf{PACS numbers:} 04.20.Cv, 04.20.Jb, 04.20.Gz

%#####################################################################################

\renewcommand{\arraystretch}{1.5}
\setlength\arraycolsep{4pt}
\setcounter{equation}{0}

\section{Introduction}

\subsection{Background}

This article is thought of as a first part of more extensive work devoted to the para-Hermite and para-Kähler spaces (abbreviated by pH-spaces and pK-spaces, respectively; also pHE stands for \textsl{para-Hermite-Einstein} and pKE stands for \textsl{para-Kähler-Einstein}). PH structures were defined for the first time in \cite{Libermann}. Since then pH structures and their modifications appeared in many geometrical problems. Recently, such spaces have been considered in the papers \cite{An, Bor, Bor_Makhmali_Nurowski}. In \cite{Bor} a relation between homogeneous pKE metric called \textsl{the dancing metric} and (2,3,5)-distributions with maximal algebra of infinitesimal symmetries has been found. In \cite{An} it is explained how to construct a rank 2 distribution on the 5-dimensional circle bundles of null self-dual and anti-self-dual planes of 4-dimensional conformal structures of neutral signature. Such a distribution is called \textsl{twistor distribution}. A fundamental invariant of twistor distributions, so-called \textsl{Cartan quartic}, depends on the components of the Weyl curvature of the conformal structure. The question asked in \cite{Bor_Makhmali_Nurowski} was whether the root type of the Cartan quartic of twistor distribution agrees with the root type of the quartic representation of the Weyl tensor? In \cite{Bor_Makhmali_Nurowski} it is shown that the answer to this question is positive. 

In \cite{Bor_Makhmali_Nurowski} the authors emphasize the importance of pKE metrics. To obtain examples of pKE metrics the authors directly integrated Cartan structure equations. They presented examples of pKE-spaces for which the anti-self-dual Weyl spinor is of Petrov-Penrose (real) types [II], [III] and [N]. They also found a general form of the metric of  pKE-space of the type [D] (to be more precise, it is the metric of the type $[\textrm{D}_{r}]^{nn} \otimes [\textrm{D}_{r}]^{ee}$, see Sections \ref{nomenclature} and \ref{sekcja_klasyfikacja_petrova} for explanation of the symbols). This metric depends on 5 constants\footnote{It is well-known that a general solution of the Einstein vacuum equations of the type [D] in Lorentzian spaces depends on 7 constants (this solution is known as Plebański-Demiański solution \cite{Plebanski_Demianski}). Spaces equipped with a neutral signature metric which are of the (real) type [D] splits into three classes. The first class is the type $[\textrm{D}_{r}]^{nn} \otimes [\textrm{D}_{r}]^{nn}$. It depends only on cosmological constant. The second class is the type $[\textrm{D}_{r}]^{nn} \otimes [\textrm{D}_{r}]^{ee}$ (it depends on 5 constants, \cite{Bor_Makhmali_Nurowski}). The last class is the type $[\textrm{D}_{r}]^{ee} \otimes [\textrm{D}_{r}]^{ee}$. Such a space is not pKE anymore although it belongs to pHE class. The general metric of this last class remains unknown.}. 

Examples of pKE and pHE metrics of various Petrov-Penrose types have been presented also in our previous work \cite{Chudecki_przyklady}. These examples have been obtained as a real neutral slice of the 4-dimensional, complex space called \textsl{the hyperheavenly space ($\mathcal{HH}$-space)}. Note, that there are examples of the metrics of the type $[\textrm{D}]^{nn} \otimes [\textrm{D}]^{ee}$ in \cite{Chudecki_przyklady}. However, these examples depend on 3 complex constants only and they are not general. Thus, a progress made on this field by Bor, Makhmali and Nurowski is remarkable.

The results published in \cite{An, Bor, Bor_Makhmali_Nurowski,Chudecki_przyklady} suggest that pKE-spaces play an important role in mathematical physics. Hence, the following research program seems to be reasonable:
\begin{enumerate}[label=(\roman*)]
\item Classification of pKE metrics for which the self-dual (SD) or the anti-self-dual (ASD) Weyl spinor is algebraically degenerate,
\item Solving vacuum Einstein equations with cosmological constant for different classes of pKE-spaces or - if these equations are too complicated to be solved completely - to reduce them as much as possible.
\end{enumerate}
In this paper we have completed part (i). In fact, the classification we propose holds true also in pK, pH and pHE-spaces. We also carried out a part of (ii). To explain which part of (ii) has been fulfilled we need to focus on the properties of geometric structures referred to as \textsl{congruences of null strings} (also called \textsl{foliations of null strings}). 

Congruences of null strings are families of totally null and totally geodesic 2-dimensional surfaces (\textsl{the null strings}). In real 4-dimensional spaces they appear as the integral manifolds of the 2-dimensional, totally null distributions and they have been investigated since the fifties \cite{Walker} (\textsl{Walker spaces}). Note that congruences of null strings can only exist in conformal structures of neutral signature. They cannot exist in Lorentzian spaces and in Riemannian spaces\footnote{4-dimensional real manifolds equipped with a metric can be of three different types for which we use the following terminology. Spaces equipped with a metric of the signature $(+---)$ are called \textsl{Lorentzian} (or \textsl{hyperbolic}). \textsl{Neutral spaces} (also called \textsl{split} or \textsl{ultrahyperbolic}) are endowed with a metric of the signature $(++--)$. Finally, \textsl{Riemannian spaces} (also called \textsl{proper-Riemannian} or \textsl{Euclidean}) are equipped with a metric of the signature $(++++)$.}. However, congruences of null strings are also admitted by complex 4-dimensional spaces called \textsl{the hyperheavenly spaces} ($\mathcal{HH}$-spaces) \cite{Plebanski_Robinson_1,Plebanski_Robinson_2}. This fact is crucial for our further work.

Hyperheavenly spaces evolved from \textsl{heavenly spaces} ($\mathcal{H}$-spaces) in the seventies and along with spinors \cite{Penrose,Plebanski_Spinors} and \textsl{twistors} they are a powerful tool in complex analysis of a spacetime. $\mathcal{HH}$-spaces are defined as a complex 4-dimensional manifolds equipped with a holomorphic metric which satisfies the vacuum Einstein equations with cosmological constant and for which SD (or ASD) part of the Weyl tensor is algebraically degenerate. Null strings in $\mathcal{HH}$-spaces are 2-dimensional holomorphic surfaces (in neutral spaces they are 2-dimensional real surfaces).

Properties of congruences of null strings have been investigated in \cite{Plebanski_Rozga}. For our purposes it is necessary to explain their most important property which is called \textsl{expansion of the congruence}. In general, congruences of null strings are not parallely propagated (for details see Section \ref{congruences_of_null_strings}). Following Boyer, Finley and Plebański we call such congruences \textsl{expanding} \cite{Boyer_Finley_Plebanski}. The congruences which are parallely propagated we call \textsl{nonexpanding}. Additionally, congruences of null strings are SD or ASD\footnote{In the sense that at each point $p$ of a null string a bivector tangent to the null string is SD or ASD.} which depends on the orientation. After changing the orientation SD congruences become ASD congruences and vice-versa.

Real 4-dimensional Einstein spaces can be obtained as \textsl{real slices} \cite{Rozga} of $\mathcal{HH}$-spaces. It is quite hard to obtain Riemannian and Lorentzian spaces from $\mathcal{HH}$-spaces but real neutral slices of $\mathcal{HH}$-spaces can be obtained easily. Thus, $\mathcal{HH}$-spaces are useful tool in investigations of pHE and pKE-spaces \cite{Przanowski_Formanski_Chudecki, Chudecki_przyklady}. Also, pH and pK-spaces can be obtained as a real slice of \textsl{weak $\mathcal{HH}$-spaces} which are generalizations of $\mathcal{HH}$-spaces (see Section \ref{subsekcja_degn_x_any}).

PH-spaces are usually defined as real neutral spaces equipped with an integrable \textsl{almost para-complex} structure\footnote{\textsl{An almost para-complex structure} is an endomorphism $K:T\mathcal{M} \rightarrow T \mathcal{M}$ such that $K^{2} = \textrm{id}_{T\mathcal{M}}$ whose $\pm 1$-eigenspaces have rank 2.} and a metric of split signature which satisfies certain compatibility condition (for a brief treatment of this topic see \cite{Bor_Makhmali_Nurowski}). Equivalently, pH-spaces are equipped with two distinct expanding congruences of null strings of the same duality. PK-spaces are pH-spaces for which so called \textsl{para-Kähler} 2-form is closed. It is equivalent to the fact that both congruences of null strings are nonexpanding (see, e.g. \cite{Przanowski_Formanski_Chudecki}). Let the orientation be chosen in such a manner that both nonexpanding congruences of null strings are ASD. Hence, the ASD Weyl spinor of pK-spaces must be of the type [D] or [O] Consequently, pK-spaces can be of the types $[\textrm{any}] \otimes [\textrm{D}]^{nn}$ (if curvature scalar $R \ne 0$) or $[\textrm{any}] \otimes [\textrm{O}]^{n}$ (if $R = 0$, in this case a space is SD). 

Consider now pK-spaces equipped with an additional congruence of SD null strings. This SD congruence can be expanding or nonexpanding. The present paper is devoted to the pK-spaces equipped with a nonexpanding congruence of SD null strings. It will be shown that these are spaces of the types $[\textrm{II,D}]^{n} \otimes [\textrm{D}]^{nn}$ or $[\textrm{III,N}]^{n} \otimes [\textrm{O}]^{n}$. Also, if more SD congruences of null strings exist in a space then this space is of one of the types $[\textrm{II,D}]^{ne} \otimes [\textrm{D}]^{nn}$, $[\textrm{D}]^{nn} \otimes [\textrm{D}]^{nn}$, $[\textrm{III}]^{ne} \otimes [\textrm{O}]^{n}$ and finally $[\textrm{II}]^{nee} \otimes [\textrm{D}]^{nn}$. However, this last type is not analyzed in this article\footnote{Frankly, spaces equipped with more then 2 congruences of null strings of the same duality are so interesting that they deserve a separate paper.}. Note, that among pKE-spaces only types $[\textrm{II}]^{n} \otimes [\textrm{D}]^{nn}$, $[\textrm{D}]^{nn} \otimes [\textrm{D}]^{nn}$ and $[\textrm{III,N}]^{n} \otimes [\textrm{O}]^{n}$ are possible and we found all metrics of these types.

Our considerations are local and in general complex. We consider complex manifolds of dimension four equipped with a holomorphic metric. The results can be easily carried over to the case of real manifolds with a neutral signature metric. We do not consider real Lorentzian slices and real Riemannian slices of the metrics presented in this paper.

\subsection{Summary of main results}

There are three main aims of this paper. The first aim is a detailed analysis of the spaces equipped with three distinct nonexpanding congruences of null strings (one SD and two ASD). Such spaces are \textsl{two-sided Walker} and para-Kähler. The results are presented in Sections \ref{subsekcja_deg_nx_D_nn} and \ref{sekcja_rozwiazania_Einsteinowskie}. Especially interesting is the metric (\ref{twierdzenie_metryka_Walker_2}) which is a general metric of the space of the type $[\textrm{II}]^{n} \otimes [\textrm{D}]^{nn}$. This metric specialized to the Einstein case gives the metric (\ref{twierdzenie_metryka_typ_IIIn_x_nic_Einstein}). 

Note, that a first step to obtain the metric (\ref{twierdzenie_metryka_Walker_2}) is to equip a weak $\mathcal{HH}$-space with an additional expanding congruence of ASD null strings (Section \ref{subsekcja_degn_x_dege}). Law and Matsushita called such spaces \textsl{sesquiWalker spaces} \cite{Law}. However, the Authors of \cite{Law} found only one class of metrics of such spaces (the metric (\ref{twierdzenie_metryka_degn_x_anye_minus_minus}) in our paper). The second (more generic) class remained undiscovered in \cite{Law}. We filled this gap (the metric (\ref{twierdzenie_metryka_degn_x_anye_plus_plus})). 

The second aim is an analysis of the SD spaces which can be obtained from the metric (\ref{twierdzenie_metryka_Walker_2}). The results are gathered is Sections \ref{subsekcja_III_NXnic} and \ref{subsekcja_III_NXnic_Einstein}. An advantage of our approach is the fact that we obtained SD metrics of the types $[\textrm{III,N}]^{n} \otimes [\textrm{O}]^{n}$ with maximally reduced number of arbitrary functions. A special attention is paid to the metric (\ref{ogolna_metryka_vacuum_IIIxO}) which is the general metric of a space of the type $[\textrm{N}]^{n} \otimes [\textrm{O}]^{n}$. Such a space has a rare property of being two-sided conformally recurrent \cite{Plebanski_Przanowski_rec}. However, in \cite{Plebanski_Przanowski_rec} the authors listed three classes of such metrics with five arbitrary functions of two variables each. We proved that these three classes can be reduced to a single class which depends on two functions of two variables and one constant. 

The third aim is a development of an approach to the problem of subclassification of the spaces equipped with congruences of null strings of a different duality. SD and ASD congruences of null strings intersect and this intersection constitutes the congruence of null (complex) geodesics. Properties of this intersection are used as a subcriterion in the classification. A foundation of our approach is presented in Section \ref{section_preliminaries} (earlier in \cite{Chudecki_typy_N}) and full classification is listed in Appendix \ref{Dodatek_klasyfikacja}.

The paper contains also examples of the spaces equipped with expanding and nonexpanding congruences of null strings of the same duality. These are the metrics (\ref{twierdzenie_metrykadegn_x_degne_minisminus_plusplus}), (\ref{twierdzenie_metrykadegn_x_degne_minisminus_minisminus}), (\ref{IIxD_mmmmmmpp}), (\ref{metryka_DxD}) and (\ref{ogolna_metryka_vacuum_IIIxO}) with constraints (\ref{jawne_przyklady_rozwiazan}). According to our best knowledge these are the first explicit examples of such metrics.

Our paper is organized, as follows. In Section \ref{section_preliminaries} the algebraic preliminaries of the subject are presented (formalism, Petrov-Penrose classification, congruences of null strings and their intersections, weak $\mathcal{HH}$-spaces). Section \ref{Sekcja_przestrzenie_Walkera} is devoted to the sesquiWalker spaces, two-sided Walker spaces and pK-spaces. In Section \ref{sekcja_rozwiazania_Einsteinowskie} two-sided Walker-Einstein spaces and pKE-spaces are considered. Finally, in Appendix \ref{Dodatek_klasyfikacja} the detailed classification of the spaces equipped with at most two congruences of SD and ASD spaces is presented.

\subsection{Some basic abbreviations}
\label{nomenclature}

A remark about abbreviations used in our paper is needed. We use abbreviations:
\begin{eqnarray}
\nonumber
\mathcal{C} &-& \textrm{congruence of null strings}
\\ \nonumber
\mathcal{C}s &-& \textrm{congruences of null strings (if there is no need to establish}
\\ \nonumber
&& \textrm{the number of congruences or the properties of congruences)}
\\ \nonumber
\mathcal{C}_{m^{A}} &-& \textrm{SD congruence of null strings generated by a spinor } m^{A}
\\ \nonumber
\mathcal{C}_{m^{\dot{A}}} &-& \textrm{ASD congruence of null strings generated by a spinor } m^{\dot{A}}
\\ \nonumber
\mathcal{C}^{n} &-& \textrm{nonexpanding congruence of null strings}
\\ \nonumber
\mathcal{C}^{e} &-& \textrm{expanding congruence of null strings}
\\ \nonumber
\mathcal{C}^{nn} &-& \textrm{two nonexpanding congruences of null strings of the same duality}
\\ \nonumber
\mathcal{C}^{ne} &-& \textrm{two congruences of null strings of the same duality,} 
\\ \nonumber
&& \textrm{one nonexpanding, one expanding}
\\ \nonumber
\mathcal{C}^{ee} &-& \textrm{two expanding congruences of null strings of the same duality}
\end{eqnarray}
There are also possibilities $\mathcal{C}^{nee}$, $\mathcal{C}^{eee}$ and $\mathcal{C}^{eeee}$, but they do not appear in our paper. We also use "mixed" symbols with obvious meaning, like $\mathcal{C}^{n}_{m^{A}}$ (nonexpanding SD congruence of null strings generated by a spinor $m^{A}$), etc.

Intersection of SD and ASD $\mathcal{C}s$ constitutes a congruence of null geodesics which properties are described by the optical scalars (shear, expansion and twist). For such structures we use the following abbreviations
\begin{eqnarray}
\nonumber
\mathcal{I} &-& \textrm{congruence of null geodesics } (\mathcal{I} \textrm{ like } \mathcal{I}\textrm{ntersection})
\\ \nonumber
\mathcal{I}s &-& \textrm{congruences of null geodesics (if there is no need to establish}
\\ \nonumber
&& \textrm{the number of intersections or the properties of intersections)}
\\ \nonumber
\mathcal{I}(\mathcal{C}_{m^{A}},\mathcal{C}_{m^{\dot{A}}}) &-& \textrm{congruence of null geodesics which is an intersection of } \mathcal{C}_{m^{A}} \textrm{ and } \mathcal{C}_{m^{\dot{A}}}
\\ \nonumber
\mathcal{I}(\mathcal{C}^{n},\mathcal{C}^{n}) &-& \textrm{congruence of null geodesics which is an intersection of SD } \mathcal{C}^{n} \textrm{ and ASD } \mathcal{C}^{n}
\\ \nonumber
\mathcal{I}(\mathcal{C}^{n},\mathcal{C}^{e}) &-& \textrm{congruence of null geodesics which is an intersection of SD } \mathcal{C}^{n} \textrm{ and ASD } \mathcal{C}^{e}
\\ \nonumber
\mathcal{I}(\mathcal{C}^{e},\mathcal{C}^{e}) &-& \textrm{congruence of null geodesics which is an intersection of SD } \mathcal{C}^{e} \textrm{ and ASD } \mathcal{C}^{e}
\\ \nonumber
\mathcal{I}^{--} &-& \textrm{congruence of null geodesics which is nonexpanding and nontwisting} 
\\ \nonumber
\mathcal{I}^{+-} &-& \textrm{congruence of null geodesics which is expanding but nontwisting} 
\\ \nonumber
\mathcal{I}^{-+} &-& \textrm{congruence of null geodesics which is nonexpanding but twisting} 
\\ \nonumber
\mathcal{I}^{++} &-& \textrm{congruence of null geodesics which is expanding and twisting} 
\end{eqnarray}

We also use equalities like $\mathcal{I}(\mathcal{C}_{m^{A}},\mathcal{C}_{m^{\dot{A}}})=\mathcal{I}^{--}$ or $\mathcal{I}(\mathcal{C}^{n},\mathcal{C}^{e})=\mathcal{I}^{++}$. For example, $\mathcal{I}(\mathcal{C}_{m^{A}},\mathcal{C}_{m^{\dot{A}}})=\mathcal{I}^{--}$ means that the intersection of  congruences $\mathcal{C}_{m^{A}}$ and $\mathcal{C}_{m^{\dot{A}}}$ is nonexpanding and nontwisting.

A space which is equipped with SD (or ASD) $\mathcal{C}^{n}$ is called \textsl{Walker space} \cite{Walker}. If a space is equipped with $\mathcal{C}^{nn}$ we deal with \textsl{double Walker spaces} \cite{Law_3}. If there are one SD $\mathcal{C}^{n}$ and one ASD $\mathcal{C}^{n}$, we call such a space \textsl{two-sided Walker} \cite{Chudecki_Przanowski_Walkery}. Consider a space equipped with SD $\mathcal{C}^{e}$ and ASD $\mathcal{C}^{e}$. Law and Matsushita called such spaces \textsl{real AlphaBeta-geometries} \cite{Law}. They also called spaces equipped with SD $\mathcal{C}^{n}$ and ASD $\mathcal{C}^{e}$ (or vice-versa) \textsl{sesquiWalker spaces}\footnote{All these names of  spaces refer to 4-dimensional neutral spaces.}.

At this point a set with reasonable names of the spaces has been exhausted. How a space equipped with SD $\mathcal{C}^{nn}$ and ASD $\mathcal{C}^{nn}$ should be called? The natural answer is \textsl{two-sided double Walker space} and it is somehow acceptable. But what about a space equipped with SD $\mathcal{C}^{n}$ and ASD $\mathcal{C}^{nn}$? The name \textsl{one-sided Walker one-sided double Walker} seems to be a little bit sloppy. If we admit the existence of $\mathcal{C}s$ of a different duality and different properties more problems with names of the spaces arise.

Realizing this, we propose uniform abbreviations. Note, that in Lorentzian spaces algebraic types of the SD and ASD parts of the Weyl spinor are the same (see Section \ref{sekcja_klasyfikacja_petrova}). Thus, to determine Petrov-Penrose type of the conformal curvature it is sufficient to use a single symbol (for example $[\textrm{I}]$ or $[\textrm{D}]$). 

However, the SD and ASD Weyl spinors can be of a different Petrov-Penrose type in Riemannian, neutral and complex spaces. Thus, a single symbol of a type is not sufficient. Usually full data of a type of the conformal curvature is given in the following symbol
\begin{equation}
\nonumber
[\textrm{SD}_{\textrm{type}}] \otimes [\textrm{ASD}_{\textrm{type}}]
\end{equation}
where
\begin{eqnarray}
\nonumber
&& \textrm{SD}_{\textrm{type}}, \textrm{ASD}_{\textrm{type}} = \{ \textrm{I}, \textrm{II}, \textrm{D}, \textrm{III}, \textrm{N}, \textrm{O}  \} \textrm{ in complex spaces}
\\ \nonumber
&& \textrm{SD}_{\textrm{type}}, \textrm{ASD}_{\textrm{type}} = \{ \textrm{I}, \textrm{D}, \textrm{O}  \} \textrm{ in Riemannian spaces}
\\ \nonumber
&& \textrm{SD}_{\textrm{type}}, \textrm{ASD}_{\textrm{type}} = \{ \textrm{I}_{r}, \textrm{I}_{rc}, \textrm{I}_{c}, \textrm{II}_{r}, \textrm{II}_{rc}, \textrm{D}_{r}, \textrm{D}_{c}, \textrm{III}_{r}, \textrm{N}_{r}, \textrm{O}_{r}  \} \textrm{ in neutral spaces}
\end{eqnarray}
For example, the symbol $[\textrm{D}] \otimes [\textrm{N}]$ means that the SD Weyl spinor is of the type [D] and the ASD Weyl spinor is of the type [N]. The types $\textrm{I}$, $\textrm{I}_{r}$, $\textrm{I}_{rc}$ and $\textrm{I}_{c}$ are called \textsl{algebraically general}. All other types are called \textsl{algebraically special} or \textsl{algebraically degenerate}.

If, additionally, a space is equipped with one or more SD or ASD $\mathcal{C}s$, we add superscripts, i.e., we use a symbol 
\begin{equation}
\label{symbol_typu}
[\textrm{SD}_{\textrm{type}}]^{i_{1}i_{2}...} \otimes [\textrm{ASD}_{\textrm{type}}]^{j_{1}j_{2}...}
\end{equation}
where the number of superscripts $i$ $(j)$ carries information about the number of SD (ASD) $\mathcal{C}s$. $i_{1}, i_{2},..., j_{1}, j_{2},... = \{ n,e \}$ where $n$ stands for $\mathcal{C}^{n}$, while $e$ stands for $\mathcal{C}^{e}$. Hence, the generic complex spaces are:
\begin{eqnarray}
\nonumber
[\textrm{deg} ]^{n} \otimes [\textrm{any}] &-& \textrm{weak nonexpanding } \mathcal{HH}\textrm{-spaces}
\\ \nonumber
[\textrm{deg} ]^{e} \otimes [\textrm{any} ] &-& \textrm{weak expanding } \mathcal{HH}\textrm{-spaces}
\\ \nonumber
[\textrm{deg} ]^{n} \otimes [\textrm{any} ], C_{ab}=0 &-& \textrm{nonexpanding } \mathcal{HH}\textrm{-spaces}
\\ \nonumber
[\textrm{deg} ]^{e} \otimes [\textrm{any} ], C_{ab}=0 &-& \textrm{expanding } \mathcal{HH}\textrm{-spaces}
\end{eqnarray}
where \textsl{deg} means that Petrov-Penrose type is algebraically special, \textsl{any} means that Petrov-Penrose type is arbitrary, $C_{ab}$ is the traceless Ricci tensor. Similarly, for neutral spaces we get:
\begin{eqnarray}
\nonumber
[\textrm{deg} ]^{n} \otimes [\textrm{any}] &-& \textrm{Walker spaces}, \textrm{real weak nonexpanding } \mathcal{HH}\textrm{-spaces}
\\ \nonumber
[\textrm{D} ]^{nn} \otimes [\textrm{any}] &-& \textrm{double Walker spaces}, \textrm{para-Kähler spaces}
\\ \nonumber
[\textrm{D} ]^{ee} \otimes [\textrm{any}] &-& \textrm{para-Hermite spaces}
\\ \nonumber
[\textrm{deg} ]^{n} \otimes [\textrm{deg} ]^{n} &-& \textrm{two-sided Walker spaces}
\\ \nonumber
[\textrm{deg} ]^{n} \otimes [\textrm{any} ]^{e} &-& \textrm{sesquiWalker spaces}
\\ \nonumber
[\textrm{any} ]^{e} \otimes [\textrm{any} ]^{e} &-& \textrm{AlphaBeta-geometries}
\\ \nonumber
[\textrm{deg} ]^{e} \otimes [\textrm{any} ] &-& \textrm{real weak expanding } \mathcal{HH}\textrm{-spaces}
\\ \nonumber
[\textrm{deg} ]^{n} \otimes [\textrm{any} ], C_{ab}=0 &-& \textrm{Walker-Einstein spaces}, \textrm{real nonexpanding } \mathcal{HH}\textrm{-spaces}
\\ \nonumber
[\textrm{deg} ]^{e} \otimes [\textrm{any} ], C_{ab}=0 &-& \textrm{real expanding } \mathcal{HH}\textrm{-spaces}
\end{eqnarray}
 Of course, our abbreviations hold true also if the orientation is changed. For example, $[\textrm{D} ]^{nn} \otimes [\textrm{any}]$ and $[\textrm{any}] \otimes [\textrm{D} ]^{nn}$ are essentially the same spaces. 

If the SD (ASD) Weyl spinor vanishes, then the corresponding space is called ASD (SD) and it is equipped with infinitely many SD (ASD) $\mathcal{C}s$. If additionally $R \ne 0$ then all these $\mathcal{C}s$ are expanding and in such a case we use the symbol $[\textrm{O}]^{e}$. If $R =0$ then there are both expanding and nonexpanding $\mathcal{C}s$. In such a case we use the symbol $[\textrm{O}]^{n}$. Hence 
\begin{eqnarray}
\nonumber
[\textrm{any} ] \otimes [\textrm{O}]^{n} &-& \textrm{SD spaces with } R = 0 \textrm{ (if } C_{ab}=0: \textrm{ heavenly spaces with } \Lambda=0)
\\ \nonumber
[\textrm{any} ] \otimes [\textrm{O}]^{e} &-& \textrm{SD spaces with } R \ne 0 \textrm{ (if } C_{ab}=0: \textrm{ heavenly spaces with } \Lambda \ne 0)
\\ \nonumber
[\textrm{O}]^{n} \otimes [\textrm{any} ]   &-& \textrm{ASD spaces with } R = 0 \textrm{ (if } C_{ab}=0: \textrm{ hellish spaces with } \Lambda=0)
\\ \nonumber
[\textrm{O}]^{e} \otimes [\textrm{any} ]   &-& \textrm{ASD spaces with } R \ne 0 \textrm{ (if } C_{ab}=0: \textrm{ hellish spaces with } \Lambda \ne 0)
\end{eqnarray}
We believe, that these abbreviations simplify considerably the rest of the article.

\setcounter{equation}{0}
\section{Preliminaries}
\label{section_preliminaries}

\subsection{The formalism}

\subsubsection{Generic complex case}

In this section we present the foundations of the null tetrad and spinorial formalisms. The spinorial formalism used in this paper is Infeld - Van der Waerden - Plebański notation. For more details see \cite{Plebanski_Spinors,Pleban_formalism_2} or \cite{Chudecki_struny,Plebanski_Przanowski_rec} for a brief summary. 

Let $(\mathcal{M}, ds^{2})$ be a 4-dimensional complex analytic differential manifold equipped with a holomorphic metric. The metric $ds^{2}$ can be written in the form 
\begin{equation}
ds^{2} = 2e^{1}e^{2} + 2e^{3}e^{4} = -\frac{1}{2} g_{A\dot{B}} g^{A \dot{B}}
\end{equation}
where 1-forms $(e^{1},e^{2},e^{3},e^{4})$ are members of a \textsl{complex null tetrad} (\textsl{a complex coframe}) and they form a basis of 1-forms. The relation between $e^{a}$ and $g^{A \dot{B}}$ reads
\begin{equation}
\label{definicccja_gAB}
(g^{A\dot{B}}) := \sqrt{2}
\left[\begin{array}{cc}
e^4 & e^2 \\
e^1 & -e^3
\end{array}\right] 
 , \  A=1,2,  \ \dot{B}=\dot{1},\dot{2}
\end{equation}
A dual basis is denoted by $(\partial_{1}, \partial_{2}, \partial_{3}, \partial_{4})$ and it is also called \textsl{a complex null tetrad} (\textsl{a complex frame}). A dual basis in spinorial formalism takes the form
\begin{equation}
(\partial_{A\dot{B}}) := -\sqrt{2}
\left[\begin{array}{cc}
\partial_{4} & \partial_{2} \\
\partial_{1} & -\partial_{3}
\end{array}\right]
\end{equation}
Spinorial indices are manipulated according to the following rules 
\begin{equation}
\label{spinorial_indices_lowering_rule}
m_{A} = \ \epsilon_{A B} m^{B}
\ , \ \ \ 
m^{A} = m_{B} \epsilon^{BA}
\ , \ \ \
m_{\dot{A}} = \ \epsilon_{\dot{A} \dot{B}} m^{\dot{B}}
\ , \ \ \ 
m^{\dot{A}} = m_{\dot{B}} \epsilon^{\dot{B} \dot{A}}
\end{equation}
where $\epsilon_{AB}$ and $\epsilon_{\dot{A}\dot{B}}$ are the spinor Levi-Civita symbols
\begin{eqnarray}
&& (\epsilon_{AB})  := \left[ \begin{array}{cc}
                            0 & 1   \\
                           -1 & 0  
                            \end{array} \right]  =:  (\epsilon^{AB} )
\ \ \ , \ \ \ 
 (\epsilon_{\dot{A}\dot{B}})  := \left[ \begin{array}{cc}
                            0 & 1   \\
                           -1 & 0  
                            \end{array} \right]  =:  (\epsilon^{\dot{A}\dot{B}} ) 
\\ \nonumber
&& \epsilon_{AC} \epsilon^{AB} = \delta^{B}_{C} \ \ \ , \ \ \ \epsilon_{\dot{A}\dot{C}} \epsilon^{\dot{A}\dot{B}} = \delta^{\dot{B}}_{\dot{C}} \ \ \ , \ \ \ 
(\delta^{A}_{C})= (\delta^{\dot{B}}_{\dot{C}})= \left[ \begin{array}{cc}
                            1 & 0   \\
                            0 & 1  
                            \end{array} \right]
\end{eqnarray}
Rules (\ref{spinorial_indices_lowering_rule}) imply the following rules for the objects from a tangent space
\begin{eqnarray}
&& \partial^{A} = \partial_{B} \epsilon^{A B}, \ \partial_{A} = \epsilon_{B A} \partial^{B}, \ \textrm{where} \ \partial^{A} := \frac{\partial}{\partial x_{A}}, \ \partial_{A} := \frac{\partial}{\partial x^{A}}
\\ \nonumber
&& \partial^{\dot{A}} = \partial_{\dot{B}} \epsilon^{\dot{A} \dot{B}}, \ \partial_{\dot{A}} = \epsilon_{\dot{B} \dot{A}} \partial^{\dot{B}}, \ \textrm{where} \ \partial^{\dot{A}} := \frac{\partial}{\partial x_{\dot{A}}}, \ \partial_{\dot{A}} := \frac{\partial}{\partial x^{\dot{A}}}
\end{eqnarray}

Define 2-forms $S^{AB}$ and $S^{\dot{A}\dot{B}}$
\begin{subequations}
\begin{eqnarray}
&& (S^{AB}) := 
\left[\begin{array}{cc}
2 e^4 \wedge e^2 & e^1 \wedge e^2 + e^3 \wedge e^4 \\
e^1 \wedge e^2 + e^3 \wedge e^4 & 2 e^3 \wedge e^1
\end{array}\right] 
\\ 
&&(S^{\dot{A} \dot{B}}) := 
\left[\begin{array}{cc}
2 e^4 \wedge e^1 & -e^1 \wedge e^2 + e^3 \wedge e^4 \\
-e^1 \wedge e^2 + e^3 \wedge e^4 & 2 e^3 \wedge e^2
\end{array}\right] 
\end{eqnarray}
\end{subequations}
$S^{AB}$ are SD and $S^{\dot{A}\dot{B}}$ are ASD. They form a basis of 2-forms.

The first Cartan structure equations in the spinorial formalism read
\begin{equation}
 d g^{A \dot{B}} + \mathbf{\Gamma}^{A}_{\ \; C} \wedge g^{C \dot{B}} + 
                                  \mathbf{\Gamma}^{\dot{B}}_{\ \; \dot{C}} \wedge g^{A \dot{C}} = 0
\end{equation}
where $\mathbf{\Gamma}_{AB}$ and $\mathbf{\Gamma}_{\dot{A}\dot{B}}$ are SD and ASD spinorial connection forms. Their relation with the spinorial forms $\Gamma_{ab}$ in the null tetrad formalism is well-known
      \begin{eqnarray}
\label{zwiazek_koneksji_spinorowej_i_klasycznej}
&&(\mathbf{\Gamma}_{AB}) = -\frac{1}{2}
                    \left[ \begin{array}{cc}
                            2\, \Gamma_{42} &  \Gamma_{12} + \Gamma_{34}  \\
                            \Gamma_{12} + \Gamma_{34} & 2\, \Gamma_{31}  
                            \end{array} \right]
\\ \nonumber
&&(\mathbf{\Gamma}_{\dot{A} \dot{B}}) = -\frac{1}{2}
                    \left[ \begin{array}{cc}
                            2\, \Gamma_{41} &  -\Gamma_{12} + \Gamma_{34}   \\
                             -\Gamma_{12} + \Gamma_{34} & 2\, \Gamma_{32}  
                            \end{array} \right]
\end{eqnarray}
If we use a decomposition
\begin{equation}
\mathbf{\Gamma}_{AB} =: - \frac{1}{2} \mathbf{\Gamma}_{AB C \dot{D}} g^{C \dot{D}}, \ \mathbf{\Gamma}_{\dot{A} \dot{B}} =: - \frac{1}{2} \mathbf{\Gamma}_{\dot{A} \dot{B} C \dot{D}} g^{C \dot{D}}, \ \Gamma_{ab} =: \Gamma_{abc} e^{c}
\end{equation}
the explicit relations between $\mathbf{\Gamma}_{AB M \dot{N}}$, $\mathbf{\Gamma}_{\dot{A}\dot{B} M \dot{N}}$ and  $\Gamma_{abc}$ read
\begin{eqnarray}
\label{jawne_relacje_miedzy_koneksjami_obu_rodzajow}
&& \mathbf{\Gamma}_{11A\dot{B}} = \sqrt{2}
\left[\begin{array}{cc}
\Gamma_{424}  & \Gamma_{422} \\
\Gamma_{421}  & -\Gamma_{423}
\end{array}\right]  , \ \ \ 
\mathbf{\Gamma}_{22A\dot{B}} = \sqrt{2}
\left[\begin{array}{cc}
\Gamma_{314} & \Gamma_{312} \\
\Gamma_{311} & -\Gamma_{313}
\end{array}\right]
\\ \nonumber
&& \mathbf{\Gamma}_{12A\dot{B}} = \frac{1}{\sqrt{2}}
\left[\begin{array}{cc}
\Gamma_{124} + \Gamma_{344} \ & \Gamma_{122} + \Gamma_{342} \\
\Gamma_{121} + \Gamma_{341} \ & -\Gamma_{123} - \Gamma_{343}
\end{array}\right]
\\ \nonumber
&& \mathbf{\Gamma}_{\dot{1}\dot{1} A\dot{B}} = \sqrt{2}
\left[\begin{array}{cc}
\Gamma_{414} & \Gamma_{412} \\
\Gamma_{411} & -\Gamma_{413}
\end{array}\right] , \ \ \ 
\mathbf{\Gamma}_{\dot{2}\dot{2} A\dot{B}} = \sqrt{2}
\left[\begin{array}{cc}
\Gamma_{324} & \Gamma_{322} \\
\Gamma_{321} & -\Gamma_{323}
\end{array}\right]
\\ \nonumber
&& \mathbf{\Gamma}_{\dot{1}\dot{2} A\dot{B}} = \frac{1}{\sqrt{2}}
\left[\begin{array}{cc}
-\Gamma_{124} + \Gamma_{344} & -\Gamma_{122} + \Gamma_{342} \\
-\Gamma_{121} + \Gamma_{341} & \Gamma_{123} - \Gamma_{343}
\end{array}\right]
\end{eqnarray}
The formula for the covariant derivative of an arbitrary spinor field in the spinorial formalism reads
\begin{eqnarray}
\label{covariant_spinorial_derivative}
\nabla_{M \dot{N}} \Psi^{A \dot{B}}_{C \dot{D}} &=& \partial_{M \dot{N}} \Psi^{A \dot{B}}_{C \dot{D}} 
+ \mathbf{\Gamma}^{A}_{\ SM \dot{N}} \, \Psi^{S \dot{B}}_{C \dot{D}} 
- \mathbf{\Gamma}^{S}_{\ CM \dot{N}} \, \Psi^{A \dot{B}}_{S \dot{D}}
\\ \nonumber
&& \ \ \ \ \ \ \ \ \ \ \ \ \ 
+ \mathbf{\Gamma}^{\dot{B}}_{\ \dot{S}M \dot{N}} \, \Psi^{A \dot{S}}_{C \dot{D}}
- \mathbf{\Gamma}^{\dot{S}}_{\ \dot{D}M \dot{N}} \, \Psi^{A \dot{B}}_{C \dot{S}}
\end{eqnarray}   
where           
\begin{equation}
\nonumber
\nabla_{A\dot{B}} := g_{aA\dot{B}} \nabla^{a} \ , \ \ \ \partial_{A\dot{B}} := g_{aA\dot{B}} \partial^{a} 
\end{equation}
and the matrices $g_{aA\dot{B}}$ are defined by the relation $g^{A\dot{B}} = g_{a}^{\ A\dot{B}}e^{a}$.

The second Cartan structure equations read
\begin{equation}
\label{drugie_rownania_struktury}
\mathbf{R}^{A}_{\ \; B} = d \mathbf{\Gamma}^{A}_{\ \; B} + \mathbf{\Gamma}^{A}_{\ \; C} \wedge \mathbf{\Gamma}^{C}_{\ \; B}, \ 
\mathbf{R}^{\dot{A}}_{\ \; \dot{B}} = d \mathbf{\Gamma}^{\dot{A}}_{\ \; \dot{B}} + \mathbf{\Gamma}^{\dot{A}}_{\ \; \dot{C}} \wedge \mathbf{\Gamma}^{\dot{C}}_{\ \; \dot{B}}
\end{equation}
$\mathbf{R}^{A}_{\ \; B}$ and $\mathbf{R}^{\dot{A}}_{\ \; \dot{B}}$ are the curvature 2-forms of the connection $\mathbf{\Gamma}^{A}_{\ \; B}$ or $\mathbf{\Gamma}^{\dot{A}}_{\ \; \dot{B}}$, respectively. Decomposition of $\mathbf{R}_{AB}= \mathbf{R}_{(AB)}$ and $\mathbf{R}_{\dot{A}\dot{B}} = \mathbf{R}_{(\dot{A}\dot{B})}$ reads
\begin{eqnarray}
\label{definicja_curvatury}
\mathbf{R}_{AB} &=& - \frac{1}{2} \, C_{ABCD} \, S^{CD} + \frac{R}{24} \, S_{AB} 
           + \frac{1}{2} \, C_{AB \dot{C}\dot{D}} \, S^{\dot{C} \dot{D}}
\\ \nonumber
\mathbf{R}_{\dot{A}\dot{B}} &=& - \frac{1}{2} \, C_{\dot{A}\dot{B}\dot{C}\dot{D}} \, S^{\dot{C} \dot{D}}
                  + \frac{R}{24} \, S_{\dot{A} \dot{B}} 
           + \frac{1}{2} \, C_{CD \dot{A}\dot{B}} \, S^{CD}
\end{eqnarray}
$C_{ABCD}=C_{(ABCD)}$ ($C_{\dot{A}\dot{B}\dot{C}\dot{D}}=C_{(\dot{A}\dot{B}\dot{C}\dot{D})}$) is the spinorial image of the SD (ASD) part of the Weyl tensor; $C_{AB \dot{C}\dot{D}} = C_{(AB) \dot{C}\dot{D}} =C_{AB (\dot{C}\dot{D})}$ is the spinorial image of the traceless Ricci tensor and, finally, $R$ is the curvature scalar. 

Dotted and undotted spinors transform as follows
\begin{equation}
m'^{A_{1}...A_{n}} = L^{A_{1}}_{\ \; R_{1}} \, ... \, L^{A_{n}}_{\ \; R_{n}} \, m^{R_{1}...R_{n}}, \ 
m'^{\dot{A}_{1}...\dot{A}_{n}} = M^{\dot{A}_{1}}_{\ \; \dot{R}_{1}} \, ... \, M^{\dot{A}_{n}}_{\ \; \dot{R}_{n}} \, m^{\dot{R}_{1}...\dot{R}_{n}}
\end{equation}
where $L^{A}_{\ \; R}, \, M^{\dot{A}}_{\ \; \dot{R}} \in SL(2, \mathbb{C})$.

\subsubsection{Real neutral case}

Any 4-dimensional real space can be obtained from a generic 4-dimensional complex space by the procedure of a real slice of a complex space \cite{Rozga}. To obtain such a slice one has to use \textsl{reality conditions}. We skip a discussion about the reality conditions in Lorentzian and Riemannian spaces, we focus only on neutral spaces. There are two different ways of obtaining neutral spaces from complex spaces. Following \cite{Rod_Hill_Nurowski} we call them $S_{r}$ ("split real") and $S_{c}$ ("split complex")
\begin{eqnarray}
\nonumber
\textrm{Reality conditions } S_{r}: && \overline{g^{A\dot{B}}}  = g^{A\dot{B}} \ \Longrightarrow \ \overline{e^{a}} = e^{a} , \ a=1,2,3,4
\\ \nonumber
\textrm{Reality conditions } S_{c}: && \overline{g^{A\dot{B}}} = \eta^{AC}\eta^{\dot{B}\dot{D}} \, \overline{g_{C\dot{D}}} \ \Longrightarrow \ \overline{e^{2}} = e^{1}, \ \overline{e^{4}}=- e^{3}  
\\ \nonumber
&& \textrm{where }(\eta^{AB}) =(\eta^{\dot{A}\dot{B}}) := \left[ \begin{array}{cc}
                            1 & 0   \\
                            0 & -1  
                            \end{array} \right] 
\end{eqnarray}
where "bar" means the complex conjugation. Application of $S_{r}$ is easy: it is enough to replace all complex coordinates by real ones and all holomorphic functions by real smooth ones. 

In the rest of the text we always use the reality conditions $S_{r}$. The reason why we prefer $S_{r}$ to $S_{c}$ is as follows. We are interested in neutral slices of complex spaces which inherit $\mathcal{C}s$ in the sense that complex $\mathcal{C}s$ in a complex space become real $\mathcal{C}s$ in neutral slice of the complex space. In other words, $\mathcal{C}s$ must be characterized by \textsl{the real index} equal 2 \cite{Kopczynski, Rod_Hill_Nurowski}. We refer the Reader to Section 7.2 of \cite{Rod_Hill_Nurowski} where advantages of $S_{r}$ over $S_{c}$ in such a case have been clearly explained.

\subsection{Petrov-Penrose classification}
\label{sekcja_klasyfikacja_petrova}

Algebraic classification of totally symmetric 4-index spinors has been presented in \cite{Penrose}. This classification applied to the SD (ASD) Weyl spinor leads to the \textsl{Petrov-Penrose} classification of the conformal curvature. A contraction of $C_{ABCD}$ with an arbitrary 1-index spinor $\xi^{A}$ such that $\xi^{2} \ne 0$ yields
\begin{equation}
C_{ABCD} \xi^{A} \xi^{B} \xi^{C} \xi^{D} =  (\xi^{2})^{4} \, \mathcal{P} (z)
\end{equation}
where $\mathcal{P} (z)$ is a 4-th order polynomial in $z := \xi^{1} / \xi^{2}$. Due to the fundamental theorem of algebra $\mathcal{P}$ can be always brought to a factorized form. Hence 
\begin{equation}
(\xi^{2})^{4} \, \mathcal{P} (z) = (a_{A} \xi^{A}) (b_{B} \xi^{B})(c_{C} \xi^{C})(d_{D} \xi^{D})
\end{equation}
Because of the arbitrariness of $\xi^{A}$ we find
\begin{equation}
C_{ABCD} = a_{(A} b_{B} c_{C} d_{D)}
\end{equation}
Complex, 1-index, undotted spinors $a_{A}$, $b_{A}$, $c_{A}$ and $d_{A}$ are called \textsl{Penrose spinors}. Penrose spinors are mutually linearly independent in general. In such a case the SD Weyl spinor is \textsl{algebraically general}. It corresponds to the case when the polynomial $\mathcal{P}(z)$ has four different roots. If at least two Penrose spinors are proportional to each other then the SD Weyl spinor is \textsl{algebraically special}. There are different patterns of the roots of polynomial $\mathcal{P}(z)$. Each of them corresponds to different  Petrov-Penrose types of $C_{ABCD}$. 

If $C_{ABCD}$ is complex then there are 6 different Petrov-Penrose types. In neutral spaces $C_{ABCD}$ is real. Hence, the scheme of the roots of $\mathcal{P}(z)$ is more complicated. There appear 10 different Petrov-Penrose types. The symbols which are usually used as abbreviations of the corresponding Petrov-Penrose types of spinor $C_{ABCD}$ and the scheme of the roots of the polynomial $\mathcal{P}(z)$ are gathered in the Table \ref{typy_C_ABCD}. Note, that in neutral spaces types $[\textrm{I}_{r}]$, $[\textrm{II}_{r}]$, $[\textrm{III}_{r}]$, $[\textrm{D}_{r}]$ and $[\textrm{N}_{r}]$ decompose into product of real 1-index spinors\footnote{The subscript $r$ in the symbols $[\textrm{III}_{r}]$, $[\textrm{N}_{r}]$ and $[\textrm{O}_{r}]$ is basically redundant. However, we keep this subscript to distinguish the types in complex spaces from the types in neutral spaces.}. In \cite{Bor_Makhmali_Nurowski} such types have been called \textsl{the special real Petrov types}.
\begin{table}[!ht]
\begin{center}
\begin{tabular}{|c|c|c|c|c|c|}   \hline
\multicolumn{3}{|c|}{Complex case}  & \multicolumn{3}{|c|}{Real case}  \\  \hline
Type & $C_{ABCD}=$  &  Roots of $\mathcal{P} (z)$ & Type & $C_{ABCD}=$  & Roots of $\mathcal{P} (z)$ \\ \hline
$[\textrm{I}]$      & $a_{(A} b_{B} c_{C} d_{D)}$  & $Z_{1}Z_{2}Z_{3}Z_{4}$ &
$[\textrm{I}_{r}]$  & $m_{(A} n_{B} r_{C} s_{D)}$  &  $R_{1}R_{2}R_{3}R_{4}$ \\ \cline{4-6}
& & &
$[\textrm{I}_{rc}]$  & $m_{(A} n_{B} a_{C} \bar{a}_{D)}$  &  $R_{1}R_{2}Z\bar{Z}$ \\ \cline{4-6}
& & &
$[\textrm{I}_{c}]$   & $a_{(A} \bar{a}_{B} b_{C} \bar{b}_{D)}$ &  $Z_{1}\bar{Z}_{1}Z_{2}\bar{Z}_{2}$ \\ \hline
$[\textrm{II}]$      & $a_{(A} a_{B} b_{C} c_{D)}$  & $Z_{1}^{2} Z_{2}Z_{3}$ &
$[\textrm{II}_{r}]$  & $m_{(A} m_{B} n_{C} r_{D)}$  &  $R_{1}^{2}R_{2}R_{3}$ \\ \cline{4-6}
& & &
$[\textrm{II}_{rc}]$ & $m_{(A} m_{B} a_{C} \bar{a}_{D)}$ &  $R^{2}Z\bar{Z}$ \\ \hline
$[\textrm{D}]$       & $a_{(A} a_{B} b_{C} b_{D)}$   & $Z_{1}^{2}Z_{2}^{2}$ &
$[\textrm{D}_{r}]$   & $m_{(A} m_{B} n_{C} n_{D)}$   &  $R_{1}^{2} R_{2}^{2}$ \\ \cline{4-6}
& & &
$[\textrm{D}_{c}]$   & $a_{(A} a_{B} \bar{a}_{C} \bar{a}_{D)}$ &  $Z^{2} \bar{Z}^{2}$ \\ \hline
$[\textrm{III}]$     & $a_{(A} a_{B} a_{C} b_{D)}$    & $Z_{1}^{3} Z_{2}$ &
$[\textrm{III}_{r}]$ & $m_{(A} m_{B} m_{C} n_{D)}$    &  $R_{1}^{3}R_{2}$  \\ \hline
$[\textrm{N}]$       & $a_{A} a_{B} a_{C} a_{D}$   & $Z^{4}  $ &
$[\textrm{N}_{r}]$   & $m_{A} m_{B} m_{C} m_{D}$   &  $R^{4}  $ \\ \hline
$[\textrm{O}]$       & $0$                                             &  $-$ &
$[\textrm{O}_{r}]$   & $0$                                             & $-$ \\ \hline
\end{tabular}
\caption{Petrov-Penrose types of complex and real totally symmetric 4-index spinor. $Z$ means that a root is complex while $R$ stands for a real root; the power denotes the multiplicity of the corresponding root; spinors $a_{A}$, $b_{A}$, $c_{A}$ and $d_{A}$ are complex, spinors $m_{A}$, $n_{A}$, $r_{A}$ and $s_{A}$ are real; bar stands for the complex conjugation.}
\label{typy_C_ABCD}
\end{center}
\end{table}

Consider now a pair $\nu_{A}$ and $\mu_{A}$ of normalized spinors, $\nu^{A} \mu_{A}=1$. Such spinors form a basis of 1-index undotted spinors. The SD Weyl spinor can be presented in the form
\begin{eqnarray}
\label{krzywizna_rozlozona_na_wspolll}
2C_{ABCD} &=:& C^{(1)} \, \nu_{A}\nu_{B}\nu_{C}\nu_{D} + 4 C^{(2)} \, \mu_{(A}\nu_{B}\nu_{C}\nu_{D)} + 6 C^{(3)} \, \mu_{(A}\mu_{B}\nu_{C}\nu_{D)}  \ \ \ \ \ \ 
\\ \nonumber
&& + 4 C^{(4)} \, \mu_{(A}\mu_{B}\mu_{C}\nu_{D)} +  C^{(5)} \, \mu_{A}\mu_{B}\mu_{C}\mu_{D} 
\end{eqnarray}
where the scalars $C^{(i)}$, $i = 1,2,3,4,5$ are called \textsl{the SD conformal curvature coefficients}. 

For our further purposes it is convenient to mention how a spinorial basis $(\nu_{A}, \mu_{B})$ can be adapted to the structure of the SD Weyl spinor. Let the SD Weyl spinor be algebraically degenerate, $C_{ABCD} = a_{(A} a_{B} b_{C} c_{D)}$ with $a_{A}$ being a multiple Penrose spinor. If we choose a spinorial basis of undotted spinors in such a manner that $\nu_{A} \sim a_{A}$, then $C^{(5)} = C^{(4)}=0$. In such a case the conditions for the algebraic types of the SD Weyl spinor read
\begin{eqnarray}
\label{warunki_na_typy}
\textrm{type [II]}: && C^{(3)} \ne 0, 2C^{(2)} C^{(2)} - 3 C^{(1)} C^{(3)} \ne 0
\\ \nonumber
\textrm{type [D]}: && C^{(3)} \ne 0, 2C^{(2)} C^{(2)} - 3 C^{(1)} C^{(3)} = 0
\\ \nonumber
\textrm{type [III]}: && C^{(3)} = 0, C^{(2)} \ne 0
\\ \nonumber
\textrm{type [N]}: && C^{(3)} =  C^{(2)} = 0, C^{(1)} \ne 0
\\ \nonumber
\textrm{type [O]} : && C^{(3)} =  C^{(2)} =  C^{(1)} = 0
\end{eqnarray}
Now consider the case of neutral signature conformal structures. Let the SD Weyl spinor be algebraically degenerate with $m_{A}$ being a real multiple Penrose spinor. Then $\nu_{A} \sim m_{A}$ yields $C^{(5)} = C^{(4)}=0$ and 
\begin{eqnarray}
\label{warunki_na_typy_real}
\textrm{type } [\textrm{II}_r]: && C^{(3)} \ne 0, 2C^{(2)} C^{(2)} - 3 C^{(1)} C^{(3)} > 0
\\ \nonumber
\textrm{type } [\textrm{II}_{rc}]: && C^{(3)} \ne 0, 2C^{(2)} C^{(2)} - 3 C^{(1)} C^{(3)} < 0
\\ \nonumber
\textrm{type } [\textrm{D}_r]: && C^{(3)} \ne 0, 2C^{(2)} C^{(2)} - 3 C^{(1)} C^{(3)} = 0
\\ \nonumber
\textrm{type } [\textrm{III}_r]: && C^{(3)} = 0, C^{(2)} \ne 0
\\ \nonumber
\textrm{type } [\textrm{N}_r]: && C^{(3)} =  C^{(2)} = 0, C^{(1)} \ne 0
\\ \nonumber
\textrm{type } [\textrm{O}_r]: && C^{(3)} =  C^{(2)} =  C^{(1)} = 0
\end{eqnarray}

The classification presented in this Section can be applied, \textsl{mutatis mutandis}, to the ASD Weyl spinor $C_{\dot{A}\dot{B}\dot{C}\dot{D}}$.

\subsection{Congruences of null strings}
\label{congruences_of_null_strings}

A structure which play a fundamental role in the further considerations is \textsl{a congruence of null strings} (abbreviated by $\mathcal{C}$, see Section \ref{nomenclature} for explanation of all the abbreviations). We recall only basic properties of $\mathcal{C}s$, for deeper analysis of the topic, see, e.g., \cite{Plebanski_Rozga,Chudecki_struny}). 

Consider first a 2-dimensional SD holomorphic distribution  $\mathcal{D}_{m^{A}} = \{ m_{A} a_{\dot{A}}, m_{A} b_{\dot{A}} \}$, $a_{\dot{A}} b^{\dot{B}} \ne 0$. Such a distribution is defined by the Pfaff system
\begin{equation}
m_{A} g^{A \dot{B}} = 0
\end{equation}
$\mathcal{D}_{m^{A}}$ is integrable in the Frobenius sense if and only if the spinor field $m_{A}$ satisfies the equations
\begin{equation}
\label{rownania_strun_SD}
m^{A} m^{B} \nabla_{A \dot{C}} m_{B} = 0
\end{equation}
Eqs. (\ref{rownania_strun_SD}) are called \textsl{SD null string equations}. If Eqs. (\ref{rownania_strun_SD}) hold true one says that the spinor $m_{A}$ generates the \textsl{congruence of SD null strings}. The integral manifolds of the distribution $\mathcal{D}_{m^{A}}$ are totally null and geodesic, 2-dimensional SD holomorphic surfaces (\textsl{SD null strings}, also called \textsl{totally null planes}). The family of such surfaces constitute \textsl{the congruence of SD null strings}. Eqs. (\ref{rownania_strun_SD}) can be rewritten in the equivalent form
\begin{equation}
\label{rozwiniete_rownania_strun}
\nabla_{A \dot{C}} m_{B} = Z_{A \dot{C}} m_{B} + \epsilon_{AB} M_{\dot{C}}
\end{equation}
where $Z_{A \dot{C}}$ is \textsl{the Sommers vector} and $M_{\dot{C}}$ is \textsl{the expansion of the congruence}\footnote{Note, that expansion of the congruence of null strings is a different concept than the expansion of the congruence of null geodesics.}. To understand better the geometrical meaning of $M_{\dot{C}}$ let $\mathcal{X}=X^{A\dot{B}} \partial_{A\dot{B}}$ be an arbitrary vector field and $\mathcal{V}=V^{A\dot{B}} \partial_{A\dot{B}} \in \mathcal{D}_{m^{A}} \Longleftrightarrow V_{A\dot{B}}=m_{A} V_{\dot{B}} $. Then
\begin{equation}
\nabla_{\mathcal{X}} V_{M \dot{N}} = m_{M} Y_{\dot{M}} + V_{\dot{M}} X_{M \dot{B}} M^{\dot{B}}
\end{equation}
Hence, $\nabla_{\mathcal{X}} \mathcal{V} \in \mathcal{D}_{m^{A}}$ for any vector field $\mathcal{V} \in \mathcal{D}_{m^{A}}$ and for an arbitrary vector field $\mathcal{X}$ if and only if $M_{\dot{A}} = 0$ holds. In such a case $\mathcal{D}_{m^{A}}$ is parallely propagated and the family of the null strings is a set of \textsl{totally null parallel planes} \cite{Walker}. According to Plebański - Robinson - Rózga terminology $\mathcal{C}s$ with vanishing expansion are called \textsl{nonexpanding} (or \textsl{plane}) while $\mathcal{C}s$ with $M_{\dot{A}} \ne 0$ are called \textsl{expanding} (or \textsl{deviating}). 

ASD $\mathcal{C}s$ are similarly defined but they are generated by dotted spinors. If a spinor $m_{\dot{A}}$ generates an ASD $\mathcal{C}$ then it satisfies \textsl{equations of ASD null strings} $m^{\dot{C}} m^{\dot{B}} \nabla_{A \dot{C}} m_{\dot{B}}=0$ or, equivalently
\begin{equation}
\label{ASD_null_string_equations}
\nabla_{A \dot{C}} m_{\dot{B}} = \dot{Z}_{A \dot{C}} m_{\dot{B}} + \epsilon_{\dot{C}\dot{B}} M_{A}
\end{equation}
where $\dot{Z}_{A \dot{C}}$ is the Sommers vector of $\mathcal{C}_{m^{\dot{A}}}$ and $M_{A}$ is the expansion of $\mathcal{C}_{m^{\dot{A}}}$. Note, that expansion of SD $\mathcal{C}_{m^{A}}$ (ASD $\mathcal{C}_{m^{\dot{A}}}$) is given by the dotted $M_{\dot{A}}$ (undotted $M_{A}$) spinor field.

If $C_{ABCD} \ne 0$ ($C_{\dot{A}\dot{B}\dot{C}\dot{D}} \ne 0$) then the following facts hold true (see \cite{Plebanski_Rozga,Chudecki_struny} for proofs)
\begin{itemize}
\item if a spinor $m_{A}$ ($m_{\dot{A}}$) generates a congruence of SD (ASD) null strings, then it is a undotted (dotted) Penrose spinor
\item if a spinor $m_{A}$ ($m_{\dot{A}}$) generates a nonexpanding congruence of SD (ASD) null strings, then it is a multiple undotted (dotted) Penrose spinor. 
\end{itemize}
From these facts it follows that if $C_{ABCD} \ne 0$ ($C_{\dot{A}\dot{B}\dot{C}\dot{D}} \ne 0$) then the maximal number of distinct SD (ASD) $\mathcal{C}s$ is 4 and such a case is possible only if the SD (ASD) Weyl spinor is of the type $[\textrm{I}]$. All the possibilities are presented in the Scheme \ref{Degeneration_scheme_of_complex_case}. 
\begin{Scheme}[!ht]
\begin{displaymath}
%\resizebox{1\textwidth}{!}{
\xymatrixcolsep{0.0cm}
\xymatrixrowsep{0.8cm}
\xymatrix{
0 \ \mathcal{C}s: &  [\textrm{I}]  \ar[d] & [\textrm{II}]  \ar[d]  \ar[dr]  & & [\textrm{III}]  \ar[d] \ar[dr] &    &
[\textrm{D}]  \ar[d]  \ar[dr]  &  &  &  [\textrm{N}]  \ar[d]  \ar[dr] & &  \\ 
1 \ \mathcal{C}:  & [\textrm{I}]^{e} \ar[d]  &  [\textrm{II}]^{e} \ar[d] \ar[dr] & [\textrm{II}]^{n} \ar[d]  &  [\textrm{III}]^{e} \ar[d]  \ar[dr] & [\textrm{III}]^{n} \ar[d] & 
[\textrm{D}]^{e} \ar[d]  \ar[dr] & [\textrm{D}]^{n} \ar[d]  \ar[dr] &  &  [\textrm{N}]^{e}  & [\textrm{N}]^{n}  &  \\
2 \ \mathcal{C}s:  & [\textrm{I}]^{ee}  \ar[d]  &  [\textrm{II}]^{ee} \ar[d] \ar[dr] & [\textrm{II}]^{en} \ar[d] &  [\textrm{III}]^{ee}  & [\textrm{III}]^{en}  & 
[\textrm{D}]^{ee}  & [\textrm{D}]^{en}  & [\textrm{D}]^{nn} &    &   &  \\
3 \ \mathcal{C}s:  & [\textrm{I}]^{eee} \ar[d]  &  [\textrm{II}]^{eee}  & [\textrm{II}]^{een}  &  &  & &   &  &  &   &  \\
4 \ \mathcal{C}s:  & [\textrm{I}]^{eeee}   &    & &  &  & &   &  &  &   &  \\
}
%}
\end{displaymath} 
\caption{Types of the Weyl spinors in spaces equipped with different numbers of congruences of null strings in complex case.}
\label{Degeneration_scheme_of_complex_case}
\end{Scheme}

A little more complicated scheme holds true for a neutral case. Note, that real $\mathcal{C}$ is generated by a real Penrose spinor. It implies that if the SD (ASD) Weyl spinor is of the types $[\textrm{I}_{c}]$ or $[\textrm{D}_{c}]$ then a space does not admit any SD (ASD) $\mathcal{C}s$, see Scheme \ref{Degeneration_scheme_of_real_case}.
\begin{Scheme}[!ht]
\begin{displaymath}
\resizebox{1\textwidth}{!}{
\xymatrixcolsep{0.0cm}
\xymatrixrowsep{0.8cm}
\xymatrix{
0 \ \mathcal{C}s: &  [\textrm{I}_{r}]  \ar[d] & [\textrm{I}_{rc}]  \ar[d] & [\textrm{I}_{c}] & [\textrm{II}_{r}]  \ar[d]  \ar[dr]  & & [\textrm{II}_{rc}] \ar[d]  \ar[dr] & & [\textrm{III}_{r}]  \ar[d] \ar[dr] &    &
[\textrm{D}_{r}]  \ar[d]  \ar[dr]  &  &  & [\textrm{D}_{c}] & [\textrm{N}_{r}]  \ar[d]  \ar[dr] & &  \\ 
1 \ \mathcal{C}:  & [\textrm{I}_{r}]^{e} \ar[d]  & [\textrm{I}_{rc}]^{e} \ar[d] &  & [\textrm{II}_{r}]^{e} \ar[d] \ar[dr] & [\textrm{II}_{r}]^{n} \ar[d]  & [\textrm{II}_{rc}]^{e} & [\textrm{II}_{rc}]^{n} & [\textrm{III}_{r}]^{e} \ar[d]  \ar[dr] & [\textrm{III}_{r}]^{n} \ar[d] & 
[\textrm{D}_{r}]^{e} \ar[d]  \ar[dr] & [\textrm{D}_{r}]^{n} \ar[d]  \ar[dr] &  & & [\textrm{N}_{r}]^{e}  & [\textrm{N}_{r}]^{n}  &  \\
2 \ \mathcal{C}s:  & [\textrm{I}_{r}]^{ee}  \ar[d] & [\textrm{I}_{rc}]^{ee}  & &  [\textrm{II}_{r}]^{ee} \ar[d] \ar[dr] & [\textrm{II}_{r}]^{en} \ar[d] & & &  [\textrm{III}_{r}]^{ee}  & [\textrm{III}_{r}]^{en}  & 
[\textrm{D}_{r}]^{ee}  & [\textrm{D}_{r}]^{en}  & [\textrm{D}_{r}]^{nn} &  &  &   &  \\
3 \ \mathcal{C}s:  & [\textrm{I}_{r}]^{eee} \ar[d]  & & &  [\textrm{II}_{r}]^{eee}  & [\textrm{II}_{r}]^{een} & & &  &  & &   &  &  &  & &  \\
4 \ \mathcal{C}s:  & [\textrm{I}_{r}]^{eeee} & &  & & &  & &  &  & &  & &  &  &   &  \\
}
}
\end{displaymath} 
\caption{Types of the Weyl spinors in spaces equipped with different numbers of congruences of null strings in neutral case.}
\label{Degeneration_scheme_of_real_case}
\end{Scheme}

\subsection{Properties of intersections of SD and ASD congruences of null strings}

Consider a space which is equipped with $\mathcal{C}_{m^{A}}$ (with the expansion $M_{\dot{A}}$) and $\mathcal{C}_{m^{\dot{A}}}$ (with the expansion $M_{A}$). These $\mathcal{C}s$ intersect and this intersection constitutes a congruence of (complex) null geodesics. Properties of such $\mathcal{I}$  have been investigated in \cite{Plebanski_Rozga} and then in \cite{Chudecki_typy_N}. Let $K_{a}$ be a null vector field along the $\mathcal{I}(\mathcal{C}_{m^{A}}, \mathcal{C}_{m^{\dot{A}}})$. Hence, $K_{a} \sim m_{A}m_{\dot{A}}$. Then one defines \textsl{the expansion} $\theta$ and \textsl{the twist} $\varrho$ of $\mathcal{I}(\mathcal{C}_{m^{A}}, \mathcal{C}_{m^{\dot{A}}})$ as follows\footnote{Note, that these definitions are the same like in the Lorentzian case. Nevertheless, the geometrical interpretation of $\theta$ and $\varrho$ in the complex case is not clear yet.}
\begin{subequations}
\begin{eqnarray}
\theta &:=& \frac{1}{2} \nabla^{a}K_{a} 
\\
\varrho^{2} &:=& \frac{1}{2} \nabla_{[a}K_{b]} \, \nabla^{a}K^{b} 
\end{eqnarray}
\end{subequations} 
It has been proven in \cite{Chudecki_typy_N} that if $\mathcal{I}(\mathcal{C}_{m^{A}}, \mathcal{C}_{m^{\dot{A}}})$ is in an affine parametrization then $\theta$ and $\varrho$ are proportional to the following scalars
\begin{eqnarray}
\label{properties_of_the_null_geodesics}
\theta & \sim &  m_{A} M^{A} + m_{\dot{A}} M^{\dot{A}}
\\ \nonumber
\varrho & \sim &  m_{A} M^{A} - m_{\dot{A}} M^{\dot{A}}
\end{eqnarray}
Hence, $\theta$ and $\varrho$ of $\mathcal{I}(\mathcal{C}_{m^{A}}, \mathcal{C}_{m^{\dot{A}}})$ depend on expansions $M_{\dot{A}}$ and $M_{A}$. There are four possibilities for which we propose the following symbols
\begin{eqnarray}
\label{definicja_wlasnosci_kongruencji_zerowych_geodezyjnych}
[++]: \, \theta \ne 0, \varrho \ne 0
\\ \nonumber
[+-]: \, \theta \ne 0, \varrho = 0
\\ \nonumber
[-+]: \, \theta = 0, \varrho \ne 0
\\ \nonumber
[--]: \, \theta = 0, \varrho = 0
\end{eqnarray}
Consequently, one arrives at the Table \ref{Mozliwe_przeciecia_kongruencji}. Note, that $\mathcal{I}^{-+}$ and $\mathcal{I}^{+-}$ appear only as an intersection of expanding $\mathcal{C}s$. There are two different types of $\mathcal{I}^{++}$. Namely, $\mathcal{I}^{++} = \mathcal{I}(\mathcal{C}^{n},\mathcal{C}^{e})$ or $\mathcal{I}^{++} = \mathcal{I}(\mathcal{C}^{e},\mathcal{C}^{e})$. There are also three different types of $\mathcal{I}^{--}$: $\mathcal{I}^{--} = \mathcal{I}(\mathcal{C}^{n},\mathcal{C}^{n})$ or $\mathcal{I}^{--} = \mathcal{I}(\mathcal{C}^{n},\mathcal{C}^{e}) = \mathcal{I}(\mathcal{C}^{e},\mathcal{C}^{n})$ or $\mathcal{I}^{--} = \mathcal{I}(\mathcal{C}^{e},\mathcal{C}^{e})$.
\begin{table}[h]
\begin{center}
\begin{tabular}{|c|c|c|}   \hline
Expansions &  $M^{A}=0$  &  $M^{A} \ne 0$    \\  \hline
$M^{\dot{A}}=0$ & $[--]$ & $[--]$, \, $[++]$ \\ \hline
$M^{\dot{A}} \ne 0$ & $[--]$, \, $[++]$ & $[--]$, \, $[-+]$, \, $[+-]$, \, $[++]$ \\ \hline
\end{tabular}
\caption{Types of congruences of null geodesics via properties of congruences of null strings.}
\label{Mozliwe_przeciecia_kongruencji}
\end{center}
\end{table}

\subsection{Classification of spaces equipped with at most 2 different congruences of SD and ASD null strings}

Consider now four distinct $\mathcal{C}s$ (two SD and two ASD): $\mathcal{C}_{m^{A}}$, $\mathcal{C}_{n^{A}}$,  $\mathcal{C}_{m^{\dot{A}}}$ and $\mathcal{C}_{n^{\dot{A}}}$; $m^{A} n_{A} \ne 0$, $m^{\dot{A}} n_{\dot{A}} \ne 0$. Their properties are gathered in the Table \ref{cztery_kongruencje}. Properties of the intersections of these $\mathcal{C}s$ are listed in the Table \ref{cztery_kongruencje_przeciecia}.
\begin{table}[ht]
\begin{center}
\begin{tabular}{|c|c|c|c|c|}   \hline
 Congruence  &   Duality   & Spinor  & Expansion & Tangent vector    \\ \hline \hline
 $\mathcal{C}_{m^{A}}$ & SD & $m_{A}$  & $M_{\dot{A}}$  & $ \{ m_{A} x_{\dot{B}} ,  m_{A} y_{\dot{B}} \}$, $x_{\dot{B}}y^{\dot{B}} \ne 0$ \\ \hline
 $\mathcal{C}_{n^{A}}$ & SD & $n_{A}$  & $N_{\dot{A}}$  & $ \{ n_{A} x_{\dot{B}} ,  n_{A} y_{\dot{B}} \}$, $x_{\dot{B}}y^{\dot{B}} \ne 0$ \\ \hline
  $\mathcal{C}_{m^{\dot{A}}}$ & ASD & $m_{\dot{A}}$  & $M_{A}$  & $ \{ x_{B} m_{\dot{A}}  ,  y_{B} m_{\dot{A}} \}$, $x_{B}y^{B} \ne 0$ \\ \hline
  $\mathcal{C}_{n^{\dot{A}}}$ & ASD & $n_{\dot{A}}$  & $N_{A}$  & $ \{ x_{B} n_{\dot{A}}  ,  y_{B} n_{\dot{A}} \}$, $x_{B}y^{B} \ne 0$ \\ \hline
\end{tabular}
\caption{Four distinct congruences of null strings.}
\label{cztery_kongruencje}
\end{center}
\end{table}

\begin{table}[ht]
\begin{center}
\begin{tabular}{|c|c|c|c|}   \hline
 Intersection  &   Tangent vector   & Expansion   & Twist    \\ \hline \hline
 $\mathcal{I} (\mathcal{C}_{m^{A}}, \mathcal{C}_{m^{\dot{A}}})$ & $m_{A}m_{\dot{A}}$ & $\theta \sim m_{A}M^{A} + m_{\dot{A}} M^{\dot{A}}$  & $\varrho \sim m_{A}M^{A} - m_{\dot{A}} M^{\dot{A}}$   \\ \hline
 $\mathcal{I} (\mathcal{C}_{m^{A}}, \mathcal{C}_{n^{\dot{A}}})$ & $m_{A}n_{\dot{A}}$ & $\theta \sim m_{A}N^{A} + n_{\dot{A}} M^{\dot{A}}$  & $\varrho \sim m_{A}N^{A} - n_{\dot{A}} M^{\dot{A}}$   \\ \hline
$\mathcal{I} (\mathcal{C}_{n^{A}}, \mathcal{C}_{m^{\dot{A}}})$  & $n_{A}m_{\dot{A}}$ & $\theta \sim n_{A}M^{A} + m_{\dot{A}} N^{\dot{A}}$  & $\varrho \sim n_{A}M^{A} - m_{\dot{A}} N^{\dot{A}}$   \\ \hline
 $\mathcal{I} (\mathcal{C}_{n^{A}}, \mathcal{C}_{n^{\dot{A}}})$ & $n_{A}n_{\dot{A}}$ & $\theta \sim n_{A}N^{A} + n_{\dot{A}} N^{\dot{A}}$  & $\varrho \sim n_{A}N^{A} - n_{\dot{A}} N^{\dot{A}}$   \\ \hline
\end{tabular}
\caption{Four intersections of congruences of null strings.}
\label{cztery_kongruencje_przeciecia}
\end{center}
\end{table}

We propose a classification of spaces equipped with such structures in terms of the properties of $\mathcal{C}s$ and $\mathcal{I}s$. In the symbol (\ref{symbol_typu}) Petrov-Penrose types of both SD and ASD Weyl spinors as well as a number and properties of $\mathcal{C}s$ are gathered. However, (\ref{symbol_typu}) does not cover the properties of $\mathcal{I}s$. Thus, we extend the symbol (\ref{symbol_typu}) to the form 
\begin{subequations}
\label{symbol_na_CCCC}
\begin{eqnarray}
\label{symbol_na_C_C}
\textrm{if } \mathcal{C}_{m^{A}} \textrm{ and }  \mathcal{C}_{m^{\dot{A}}} \ \textrm{exist:} \ \  \ \ \ \ \ \ \ \ \ \ \ \;  && \{ [ \textrm{SD}_{\textrm{type}} ]^{i_{1}} \otimes [ \textrm{ASD}_{\textrm{type}} ]^{j_{1}}, [k_{1}] \}
\\ 
\label{symbol_na_C_CC}
\textrm{if } \mathcal{C}_{m^{A}}, \mathcal{C}_{m^{\dot{A}}} \textrm{ and } \mathcal{C}_{n^{\dot{A}}}\ \textrm{exist:} \ \ \ \ \ \ \ \; && \{ [ \textrm{SD}_{\textrm{type}} ]^{i_{1}} \otimes [ \textrm{ASD}_{\textrm{type}} ]^{j_{1}j_{2}}, [k_{1},k_{2}] \}
\\ 
\label{symbol_na_CC_CC}
\textrm{if } \mathcal{C}_{m^{A}},\mathcal{C}_{n^{A}}, \mathcal{C}_{m^{\dot{A}}} \textrm{ and } \mathcal{C}_{n^{\dot{A}}}\ \textrm{exist:} \ \ && \{ [ \textrm{SD}_{\textrm{type}} ]^{i_{1}i_{2}} \otimes [ \textrm{ASD}_{\textrm{type}} ]^{j_{1}j_{2}}, [k_{1},k_{2},k_{3},k_{4}] \} \ \ \ \ \ \ 
\end{eqnarray}
\end{subequations}
where the superscripts $i_{1}, i_{2}, j_{1}, j_{2} = \{ n , e \}$ carry information about the expansion of  $\mathcal{C}_{m^{A}}$, $\mathcal{C}_{n^{A}}$, $\mathcal{C}_{m^{\dot{A}}}$ and $\mathcal{C}_{n^{\dot{A}}}$, respectively. Then, $k_{1}, k_{2}, k_{3}, k_{4} = \{ --,-+,+-,++ \} $ are the properties of $\mathcal{I} (\mathcal{C}_{m^{A}}, \mathcal{C}_{m^{\dot{A}}})$, $\mathcal{I} (\mathcal{C}_{m^{A}}, \mathcal{C}_{n^{\dot{A}}})$, $\mathcal{I} (\mathcal{C}_{n^{A}}, \mathcal{C}_{m^{\dot{A}}})$ and $\mathcal{I} (\mathcal{C}_{n^{A}}, \mathcal{C}_{n^{\dot{A}}})$, respectively. Note, that the properties of $\mathcal{I}s$ are listed in order which is very important. Hence, $[--,++]$ and $[++,--]$ are not, in general, the same geometries. Subtleties hidden in the symbols (\ref{symbol_na_CCCC}) will be explained with help of the three examples. 

\textbf{Example 1.} Consider the type 
\begin{equation}
\nonumber
\{ [ \textrm{II} ]^{e} \otimes [ \textrm{D} ]^{nn}, [--,++] \}
\end{equation}
In this case the SD Weyl spinor is of the type [II] and the SD Weyl spinor is of the type [D]. The space is equipped with $\mathcal{C}^{e}_{m^{A}}$, $\mathcal{C}^{n}_{m^{\dot{A}}}$ and $\mathcal{C}^{n}_{n^{\dot{A}}}$ with the following properties of $\mathcal{I}s$: $\mathcal{I} (\mathcal{C}_{m^{A}}, \mathcal{C}_{m^{\dot{A}}})=\mathcal{I}^{--}$,  $\mathcal{I} (\mathcal{C}_{m^{A}}, \mathcal{C}_{n^{\dot{A}}}) = \mathcal{I}^{++}$. Because both ASD $\mathcal{C}s$ are nonexpanding, the types $\{ [ \textrm{II} ]^{e} \otimes [ \textrm{D} ]^{nn}, [--,++] \}$ and $\{ [ \textrm{II} ]^{e} \otimes [ \textrm{D} ]^{nn}, [++,--] \} $ are, in fact, the same geometries. Hence, in this case one can replace $[--,++]$ with $[++,--]$ without changing the type. In general such a replacement leads to different types (see the next example).

\textbf{Example 2.} Consider now two types 
\begin{equation}
\nonumber
\{ [ \: \cdot \: ]^{e} \otimes [ \: \cdot \: ]^{ne}, [--,++] \}, \  \{ [ \: \cdot \: ]^{e} \otimes [ \: \cdot \: ]^{ne}, [++,--] \}
\end{equation}
Both these types are equipped with three $\mathcal{C}s$: $\mathcal{C}^{e}_{m^{A}}$, $\mathcal{C}^{n}_{m^{\dot{A}}}$  and $\mathcal{C}^{e}_{n^{\dot{A}}}$ and in both the cases $\mathcal{I}s$ are of the type $[--]$ and $[++]$. At the first glance, both types are the same geometries. However, closer investigation shows that in the first case $\mathcal{I} (\mathcal{C}^{e}_{m^{A}}, \mathcal{C}^{n}_{m^{\dot{A}}})=\mathcal{I}^{--}$ and $\mathcal{I} (\mathcal{C}^{e}_{m^{A}}, \mathcal{C}^{e}_{n^{\dot{A}}})=\mathcal{I}^{++}$ hold while in the second case $\mathcal{I} (\mathcal{C}^{e}_{m^{A}}, \mathcal{C}^{n}_{m^{\dot{A}}})=\mathcal{I}^{++}$ and $\mathcal{I} (\mathcal{C}^{e}_{m^{A}}, \mathcal{C}^{e}_{n^{\dot{A}}})=\mathcal{I}^{--}$ hold. Hence, these types represent different geometries. This example shows how important is the order in which properties of $\mathcal{I}s$ are listed.

\textbf{Example 3.} The last example is the most subtle. Consider
\begin{equation}
\nonumber
\{ [ \: \cdot \: ]^{nn} \otimes [ \: \cdot \: ]^{ee}, [--,--,++,++] \}, \  \{ [ \: \cdot \: ]^{nn} \otimes [ \: \cdot \: ]^{ee}, [--,++,++,--] \}
\end{equation}
In these cases all $\mathcal{I}s$ are $\mathcal{I} (\mathcal{C}^{n}, \mathcal{C}^{e})$. Also, both types are equipped with two $\mathcal{I}^{--}$ and two $\mathcal{I}^{++}$. However, the type $\{ [ \: \cdot \: ]^{nn} \otimes [ \: \cdot \: ]^{ee}, [--,--,++,++] \}$ is characterized by the condition $\mathcal{I} (\mathcal{C}^{n}_{m^{A}}, \mathcal{C}^{e}_{m^{\dot{A}}})=\mathcal{I} (\mathcal{C}^{n}_{m^{A}}, \mathcal{C}^{e}_{n^{\dot{A}}})=\mathcal{I}^{--}$. Type $\{ [ \: \cdot \: ]^{nn} \otimes [ \: \cdot \: ]^{ee}, [--,++,++,--] \}$ is slightly different because in this type $\mathcal{I} (\mathcal{C}^{n}_{m^{A}},\mathcal{C}^{e}_{m^{\dot{A}}}) = \mathcal{I}^{--}$ and $\mathcal{I} (\mathcal{C}^{n}_{m^{A}},\mathcal{C}^{e}_{n^{\dot{A}}})= \mathcal{I}^{++}$ hold.

In Appendix \ref{Dodatek_klasyfikacja} a detailed classification of spaces equipped with at most two SD and two ASD $\mathcal{C}s$ is given. Also it is shown how a null tetrad can be adapted to $\mathcal{C}s$.

\textbf{Remark}. Note, that $\mathcal{I} (\mathcal{C}^{n}, \mathcal{C}^{n})$ is always $\mathcal{I}^{--}$. In such a case the symbols (\ref{symbol_na_CCCC}) can be simplified by omitting $[k_{1}]$, $[k_{1}, k_{2}]$ and $[k_{1}, k_{2}, k_{3}, k_{4}]$ parts. Thus, in what follows the type $\{ [ \: \cdot \: ]^{n} \otimes [ \: \cdot \: ]^{n}, [--] \}$ will be replaced by the simpler symbol $[ \: \cdot \: ]^{n} \otimes [ \: \cdot \: ]^{n}$; $\{ [ \: \cdot \: ]^{n} \otimes [ \: \cdot \: ]^{nn}, [--,--] \}$ by $[ \: \cdot \: ]^{n} \otimes [ \: \cdot \: ]^{nn}$ and finally $\{ [ \: \cdot \: ]^{nn} \otimes [ \: \cdot \: ]^{nn}, [--,--,--,--] \}$ by $ [ \: \cdot \: ]^{nn} \otimes [ \: \cdot \: ]^{nn}$.

\subsection{Weak hyperheavenly spaces}
\label{subsekcja_degn_x_any}

\subsubsection{The metric, connection and curvature}

In this section we introduce basic definitions and notations of the \textsl{weak hyperheavenly spaces}. Such spaces have been defined in \cite{Chudecki_Przanowski_Walkery}. 
\begin{Definicja}
\label{definicja_weak_HH_spaces}
\textsl{Weak hyperheavenly space (weak $\mathcal{HH}$-space)} is a pair $(\mathcal{M}, ds^{2})$ where $\mathcal{M}$ is a 4-dimensional complex analytic differential manifold and $ds^{2}$ is a holomorphic metric, satisfying the following conditions:
\begin{enumerate}[label=(\roman*)]
\item there exists a 2-dimensional holomorphic totally null self-dual integrable distribution given by the Pfaff system
\begin{equation}
\nonumber
%\label{condition_1}
m_{A} \, g^{A \dot{B}} = 0  , \ m_{A} \ne 0
\end{equation}
\item the SD Weyl spinor $C_{ABCD}$ is algebraically degenerate and $m_{A}$ is a multiple Penrose spinor i.e.
\begin{equation}
\nonumber
%\label{condition_2}
C_{ABCD} \, m^{A} m^{B} m^{C} =0 
\end{equation}
\end{enumerate}
\end{Definicja}
Hence, weak $\mathcal{HH}$-space is a space of type $[\textrm{deg}]^{e} \otimes [\textrm{any}]$. In what follows we specialize the SD $\mathcal{C}$ to be nonexpanding. Note, that if $\mathcal{C}$ is nonexpanding then it is generated by a spinor which is a multiple Penrose spinor. Thus - for \textsl{weak nonexpanding $\mathcal{HH}$-spaces} (i.e., types $[\textrm{deg}]^{n} \otimes [\textrm{any}]$) - the condition $(ii)$ in Definition \ref{definicja_weak_HH_spaces} is implied by the condition $(i)$.

The metric of weak nonexpanding $\mathcal{HH}$-space can be brought to the form \cite{Chudecki_Przanowski_Walkery}
\begin{equation}
\label{metryka_slaba_HH}
\frac{1}{2} \, ds^{2}=  -dp^{\dot{A}}  dq_{\dot{A}} + Q^{\dot{A}\dot{B}} \, dq_{\dot{A}}  dq_{\dot{B}}= e^{1}e^{2} + e^{3} e^{4}
\end{equation}
where $(q_{\dot{A}}, p_{\dot{B}})$ are local complex coordinates. The coordinate system $(q_{\dot{A}}, p_{\dot{B}})$ is chosen in such a manner that $p_{\dot{A}}$ are coordinates on null strings while coordinates $q_{\dot{A}}$ label the null strings. $Q^{\dot{A}\dot{B}} = Q^{(\dot{A}\dot{B})}$ are holomorphic functions of variables $(q^{\dot{A}}, p^{\dot{B}})$. After imposing the reality conditions $S_{r}$ the metric (\ref{metryka_slaba_HH}) becomes the metric of the Walker space. 

An appropriate choice of a null tetrad is essential for our further purposes. It is convenient to work with the null tetrad $(e^{1}, e^{2}, e^{3}, e^{4})$ defined as follows
\begin{eqnarray}
\label{tetrada_spinorowa}
[ e^{3} , e^{1}  ] &=& -\frac{1}{\sqrt{2}} \, g^{2}_{\ \, \dot{A}} =  dq_{\dot{A}}
\\ \nonumber
[  e^{4} , e^{2}  ]   &=& \frac{1}{\sqrt{2}} \, g^{1 \dot{A}} = -dp^{\dot{A}} + Q^{\dot{A} \dot{B}} \, dq_{\dot{B}}
\end{eqnarray}
The operators
\begin{eqnarray}
&& \partial_{\dot{A}} := \frac{\partial}{\partial p^{\dot{A}}}, \ \eth_{\dot{A}} :=  \frac{\partial}{\partial q^{\dot{A}}} - Q_{\dot{A}}^{\ \ \dot{B}} \partial_{\dot{B}} 
\\ \nonumber
&& \partial^{\dot{A}} := \frac{\partial}{\partial p_{\dot{A}}}, \ \eth^{\dot{A}} :=   \frac{\partial}{\partial q_{\dot{A}}} + Q^{\dot{A} \dot{B}} \partial_{\dot{B}} 
\end{eqnarray}
form the dual basis
\begin{equation}
-\partial_{\dot{A}} = [\partial_{4}, \partial_{2} ] , \ 
\eth^{\dot{A}} = [ \partial_{3} , \partial_{1}], \ \Longrightarrow \ \partial_{A\dot{B}}  = \sqrt{2} \, [ \partial_{\dot{B}}, \eth_{\dot{B}}]
\end{equation}
The null tetrad defined by (\ref{tetrada_spinorowa}) is called \textsl{Plebański tetrad}. It is easy to see that null strings (leaves of SD $\mathcal{C}$) are spanned by the vectors $(\partial_{2}, \partial_{4})$, i.e., $\mathcal{C}^{n}_{m^{A}}$ is generated by the spinor $m_{A} = [0,m]$, $m \ne 0$. In this sense Plebański tetrad is \textsl{adapted} to the SD $\mathcal{C}$.

From the first structure equations one finds that the only nonzero spinorial connection coefficients read
\begin{equation}
\label{rozpisane_wspolczynniki_koneksji_spinorowej}
\mathbf{\Gamma}_{122 \dot{D}} = - \frac{1}{\sqrt{2}} \,  \partial^{\dot{A}}  Q_{\dot{A}\dot{D}}, \
\mathbf{\Gamma}_{222 \dot{D}} = -\sqrt{2}  \, \eth^{\dot{A}} Q_{\dot{A}\dot{D}}, \
\mathbf{\Gamma}_{\dot{A} \dot{B} 2 \dot{D}} =  \sqrt{2} \,  
 \partial_{(\dot{A}} Q_{\dot{B})\dot{D}}  
\end{equation}
Nonzero SD curvature coefficients $C^{(i)}$, the curvature scalar $R$ and the ASD Weyl spinor take the form
\begin{eqnarray}
\label{krzywizna}
&&  C^{(3)} = \frac{R}{6} = -\frac{1}{3} \, \partial_{\dot{A}} \partial_{\dot{B}} Q^{\dot{A} \dot{B}}, \
 C^{(2)}  = - \,  \partial^{\dot{A}}  \eth^{\dot{B}} Q_{\dot{A} \dot{B}}
\\ \nonumber
&& \frac{1}{2} \, C^{(1)}  = 
- \eth^{\dot{A}}  \eth^{\dot{B}} Q_{\dot{A} \dot{B}}  +
 (\eth^{\dot{A}} Q_{\dot{A} \dot{B}})
(\partial_{\dot{C}} Q^{\dot{B} \dot{C}}) , \ C_{\dot{A}\dot{B}\dot{C}\dot{D}} = -  
\partial_{(\dot{A}}\partial_{\dot{B}} Q_{\dot{C} \dot{D})}
\end{eqnarray}
The traceless Ricci tensor is given by the formulas
\begin{equation}
\label{Traceless_Ricci}
C_{11 \dot{A} \dot{B}} =0, \ C_{12 \dot{A} \dot{B}} = - \frac{1}{2}
 \partial_{(\dot{A}} \partial^{\dot{C}} Q_{\dot{B}) \dot{C}} , \
C_{22 \dot{A} \dot{B}} = -  \partial_{(\dot{A}}  \eth^{\dot{C}} Q_{\dot{B}) \dot{C}}
\end{equation}

The metric (\ref{metryka_slaba_HH}) remains invariant under the following transformations of the coordinates
\begin{equation}
\label{gauge}
q'^{\dot{A}} = q'^{\dot{A}} (q^{\dot{M}}), \ p'^{\dot{A}} = D^{-1 \ \ \dot{A}}_{\ \ \; \dot{B}} \, p^{\dot{B}} + \sigma^{\dot{A}}
\end{equation}
where $\sigma^{\dot{A}}=\sigma^{\dot{A}}(q_{\dot{B}})$ are arbitrary functions and 
\begin{eqnarray}
\nonumber
D_{\dot{A}}^{\ \ \dot{B}} 
&:=& \frac{\partial q'_{\dot{A}}}{\partial q_{\dot{B}}} 
= \Delta \, \frac{\partial q^{\dot{B}}}{\partial q'^{\dot{A}}}
= \Delta \, \frac{\partial p'_{\dot{A}}}{\partial p_{\dot{B}}}
= \frac{\partial p^{\dot{B}}}{\partial p'^{\dot{A}}} , \ \Delta := \det 
\left( \frac{\partial q'_{\dot{A}}}{\partial q_{\dot{B}}} \right)
 = \frac{1}{2} \, D_{\dot{A}\dot{B}} D^{\dot{A}\dot{B}}
\\ 
D^{-1 \ \ \dot{B}}_{\ \ \; \dot{A}} 
&=& \frac{\partial q_{\dot{A}}}{\partial q'_{\dot{B}}}
= \Delta^{-1} \, \frac{\partial q'^{\dot{B}}}{\partial q^{\dot{A}}}
= \frac{\partial p'^{\dot{B}}}{\partial p^{\dot{A}}}
= \Delta^{-1} \, \frac{\partial p_{\dot{A}}}{\partial p'_{\dot{B}}}
\end{eqnarray}
Hence
\begin{eqnarray}
 D^{\, \dot{A}}_{\ \ \dot{B}} &:=& D_{\dot{M}}^{\ \ \dot{N}} \epsilon^{\dot{M} \dot{A}} \epsilon_{\dot{B} \dot{N}} = - \Delta \, D^{-1 \ \ \dot{A}}_{\ \ \; \dot{B}}
\\ \nonumber
D^{-1 \; \dot{A}}_{\ \ \ \ \ \dot{B}} &:=& D^{-1 \ \ \dot{N}}_{\ \ \; \dot{M}} \epsilon^{\dot{M} \dot{A}} \epsilon_{\dot{B} \dot{N}} = - \frac{1}{\Delta} \, D_{\dot{B}}^{\ \ \dot{A}}
\end{eqnarray}
Functions $Q^{\dot{A}\dot{B}}$ transform under (\ref{gauge}) as follows
\begin{equation}
\label{transformacja_Q}
Q'^{\dot{A} \dot{B}} = D^{-1 \ \ \dot{A}}_{\ \ \; \dot{R}} \, D^{-1 \ \ \dot{B}}_{\ \ \; \dot{S}}
Q^{\dot{R} \dot{S}} + D^{-1 \ \ ( \dot{A}}_{\ \ \; \dot{R}} \, 
\frac{\partial p'^{\dot{B})}}{\partial q_{\dot{R}}} 
\end{equation}
Transformations (\ref{gauge}) are equivalent to the spinorial transformations 
\begin{eqnarray}
\label{transformacja_spinorowa}
L^{A}_{\ \; B} &=& \left[ \begin{array}{cc}
       \Delta^{-\frac{1}{2}} & h  \Delta^{\frac{1}{2}}   \\
       0 &  \Delta^{\frac{1}{2}}  
       \end{array} \right], \ 2 h := \frac{\partial \sigma^{\dot{R}}}{\partial q'^{\dot{R}}}
\\ \nonumber
M^{\dot{A}}_{\ \; \dot{B}} &=& \Delta^{\frac{1}{2}} \, D^{-1 \ \ \dot{A}}_{\ \ \; \dot{B}}
\end{eqnarray}
Note, that the metric (\ref{metryka_slaba_HH}) can be rewritten if the form
\begin{equation}
\label{metryka_slaba_HH_alternative_form}
\frac{1}{2} \, ds^{2}=  dqdy - dpdx + \mathcal{A} \, dp^{2} -2 \mathcal{Q} \, dpdq + \mathcal{B} \, dq^{2}
\end{equation}
where we denoted
\begin{equation}
\label{skroty_wspolrzednych}
q^{\dot{1}} =: q, \ q^{\dot{2}} =: p, \ p^{\dot{1}} =: x, \ p^{\dot{2}} =: y, \ Q^{\dot{A} \dot{B}} =: \left[ \begin{array}{cc}
       \mathcal{A} \ & \mathcal{Q}   \\
       \mathcal{Q} \ & \mathcal{B}  
       \end{array} \right]
\end{equation}

Summarizing, (\ref{metryka_slaba_HH}) ((\ref{metryka_slaba_HH_alternative_form}), alternatively) is the general metric of the spaces equipped with a single SD $\mathcal{C}^{n}$. In the next Sections we consider spaces equipped with more $\mathcal{C}s$ (see Schemes \ref{Structure} and \ref{Structure2}).

\begin{Scheme}[!ht]
\begin{displaymath}
%\resizebox{1\textwidth}{!}{
\xymatrixcolsep{0.0cm}
\xymatrixrowsep{0.0cm}
\xymatrix{
    &  [\textrm{deg}]^{n} \otimes [\textrm{any}]   \ar[ddd]    \\
    &  \\
   \textrm{let } \mathcal{C}^{e}_{m^{\dot{A}}} \textrm{ exists}  \ar[r]^-{\textrm{Sections } \ref{subsekcja_degn_x_dege_plus_plus}, \ \ref{subsekcja_degn_x_dege_minus_minus}} & \\
   &  \{ [\textrm{deg}]^{n} \otimes [\textrm{any}]^{e},[++]  \} \ \ \ (\textrm{Theorem } \ref{theorem_degn_x_anye_plus_plus}) \\
   &  \{ [\textrm{deg}]^{n} \otimes [\textrm{any}]^{e},[--]  \} \ \ \ (\textrm{Theorem } \ref{theorem_degn_x_anye_minus_minus})  \ar[ddd] \\
   &  \\
   \textrm{let } \mathcal{C}^{n}_{m^{\dot{A}}} \textrm{ exists} \ar[r]^-{\textrm{Section } \ref{subsekcja_degn_x_degn}} & \\
   &   [\textrm{deg}]^{n} \otimes [\textrm{deg}]^{n} \ \ \ (\textrm{Theorem } \ref{main_theorem_degn_x_degn})    \ar[ddd] \\
   &  \\
   \textrm{let } \mathcal{C}^{n}_{m^{\dot{A}}} \textrm{ and } \mathcal{C}^{e}_{n^{\dot{A}}} \textrm{ exist} \ar[r]^-{\textrm{Section } \ref{subsekcja_degn_x_degne_obaprzypadki}} & \\
  & \{ [\textrm{deg}]^{n} \otimes [\textrm{deg}]^{ne},[--,++] \} \ \ \ (\textrm{Theorem } \ref{theorem_degn_x_degne_minisminus_plusplus}) \\
  & \{ [\textrm{deg}]^{n} \otimes [\textrm{deg}]^{ne},[--,--] \} \ \ \ (\textrm{Theorem } \ref{theorem_degn_x_degne_minisminus_minusminus})   \ar[ddd] \\
  &  \\
   \textrm{let } \mathcal{C}^{n}_{m^{\dot{A}}} \textrm{ and } \mathcal{C}^{n}_{n^{\dot{A}}} \textrm{ exist}  \ar[r]^-{\textrm{Section } \ref{subsekcja_deg_nx_D_nn_subsekcja}} & \\
   &  [\textrm{deg}]^{n} \otimes [\textrm{D}]^{nn} \ \ \ (\textrm{Theorem } \ref{main_theorem_degn_x_degnn}) \ar[ddd]    \\
   &  \\
   \textrm{let } \mathcal{C}^{n}_{m^{\dot{A}}}, \mathcal{C}^{n}_{n^{\dot{A}}} \textrm{ and } \mathcal{C}^{e}_{n^{A}} \textrm{ exist} \ar[r]^-{\textrm{Sections } \ref{subsection_II_D_mmmmmmpp}, \ \ref{subsection_D_D_mmmmmmpp}} & \\
  & \{ [\textrm{II,D}]^{ne} \otimes [\textrm{D}]^{nn},[--,--,++,++]  \} \qquad \qquad \ \ \ \ \ \ \; \\
  & \{ [\textrm{II}]^{ne} \otimes [\textrm{D}]^{nn},[--,--,--,++] \} \ \ \ (\textrm{Theorem } \ref{theorem_II_D_mmmmmmpp})  \\
  & \{ [\textrm{D}]^{ne} \otimes [\textrm{D}]^{nn},[--,--,--,++] \} \ \ \ (\textrm{Theorem } \ref{theorem_D_D_mmmmmmpp}) \ar[ddd] \\
  &  \\
   \textrm{let } \mathcal{C}^{n}_{m^{\dot{A}}}, \mathcal{C}^{n}_{n^{\dot{A}}} \textrm{ and } \mathcal{C}^{n}_{n^{A}} \textrm{ exist}  \ar[r]^-{\textrm{Section } \ref{subsection_D_D_mmmmmmmm}} & \\
  &  [\textrm{D}]^{nn} \otimes [\textrm{D}]^{nn} \ \ \ (\textrm{Theorem } \ref{theorem_D_D_mmmmmmmm})
}
%}
\end{displaymath} 
\caption{Spaces considered in Sections \ref{subsekcja_degn_x_dege_plus_plus} - \ref{subsection_D_D_mmmmmmmm}.}
\label{Structure}
\end{Scheme}

\begin{Scheme}[!ht]
\begin{displaymath}
%\resizebox{1\textwidth}{!}{
\xymatrixcolsep{0.0cm}
\xymatrixrowsep{0.0cm}
\xymatrix{
   &   [\textrm{deg}]^{n} \otimes [\textrm{D}]^{nn}   \ar@{-}[dddd]    &    \\
     & & \\
     & & \\
      & & \\
  \ar@{-}[r] \ar@{-}[dd] & \ar@{-}[r]  &  \ar@{-}[dd] \\
     \ar[ddddd]_{(self-duality)}^{\textrm{Section } \ref{subsekcja_III_NXnic_alesubsekcja}} &  & \ar[ddddd]_{(Einstein)}^{\textrm{Section } \ref{typy_IInxDnn_alesubsekcja}} \\
      &  &  \\
      & & \\
       & & \\
       & & \\
       & & \\
   [\textrm{III}]^{n} \otimes [\textrm{O}]^{n} \ \ \ (\textrm{Theorem } \ref{theorem_vacuum_IIIxO}) &   &   [\textrm{II}]^{n} \otimes [\textrm{D}]^{nn} \ \ \ (\textrm{Theorem } \ref{theorem_vacuum_IIxD}) \\
   [\textrm{N}]^{n} \otimes [\textrm{O}]^{n} \ \ \ (\textrm{Theorem } \ref{theorem_vacuum_NxO})  \ar[dddddd]^-{\textrm{Section } \ref{subsekcja_III_NXnic_Einstein_alesubsekcja}}_{(self-duality \ and \ Einstein)} &   &  [\textrm{D}]^{nn} \otimes [\textrm{D}]^{nn}  \ \ \ (\textrm{Theorem } \ref{theorem_vacuum_IIxD})  \\
   &  & \\
   &  &  \\
    &  &  \\
     &  &  \\
      &  &  \\
   [\textrm{III}]^{n} \otimes [\textrm{O}]^{n} \ \ \ (\textrm{Theorem } \ref{ttwierdzenie_metryka_typ_III_Nn_x_nic_Einstein}) &   &     \\
   [\textrm{N}]^{n} \otimes [\textrm{O}]^{n} \ \ \ (\textrm{Theorem } \ref{ttwierdzenie_metryka_typ_III_Nn_x_nic_Einstein})  &   &     \\ 
}
\end{displaymath} 
\caption{Spaces considered in Sections \ref{subsekcja_III_NXnic_alesubsekcja} - \ref{subsekcja_III_NXnic_Einstein_alesubsekcja}.}
\label{Structure2}
\end{Scheme}

\subsubsection{ASD congruence $\mathcal{C}_{m^{\dot{A}}}$}

Before we proceed further according to Schemes \ref{Structure} and \ref{Structure2} it is helpful to write explicitly the ASD null string equations in weak $\mathcal{HH}$-spaces. Let ASD $\mathcal{C}^{e}_{m^{\dot{A}}}$ be generated by a spinor $m_{\dot{A}}$ with an expansion $M_{A}$. The ASD null string equations (\ref{ASD_null_string_equations}) contracted with $m^{\dot{B}}$, with help of (\ref{covariant_spinorial_derivative}) and (\ref{rozpisane_wspolczynniki_koneksji_spinorowej}) yield 
\begin{subequations}
\label{rownania_ASD_struny_typuCmm_dotA}
\begin{eqnarray}
\label{rownania_ASD_struny_typuCmm_dotA_1}
\frac{1}{\sqrt{2}} \, m_{\dot{A}} M_{1} &=& m^{\dot{B}} \frac{\partial m_{\dot{B}}}{\partial p^{\dot{A}}}
\\ 
\label{rownania_ASD_struny_typuCmm_dotA_2}
\frac{1}{\sqrt{2}} \, m_{\dot{A}} M_{2} &=& m^{\dot{B}} \frac{\partial m_{\dot{B}}}{\partial q^{\dot{A}}} + m^{\dot{S}} \frac{\partial}{\partial p^{\dot{S}}} ( m^{\dot{B}} Q_{\dot{A}\dot{B}} ) - m^{\dot{B}} Q_{\dot{A}\dot{B}} \, \frac{\partial m^{\dot{S}}}{\partial p^{\dot{S}}}
\end{eqnarray}
\end{subequations}
The spinor $m_{\dot{A}}$ can be always written in the form $m_{\dot{A}} = [z,1]$ where $z=z(q,p,x,y)$ (a spinor which generates $\mathcal{C}$ must be nonzero; one assumes $m_{\dot{2}} \ne 0$ and re-scale $m_{\dot{A}}$ to $m_{\dot{2}} = 1$). Under (\ref{gauge}) $z$ transforms as follows
\begin{equation}
\label{transformacja_na_z}
z' = \Delta^{-\frac{1}{2}} ( z p'_{p} - p'_{q})
\end{equation}
Contraction of (\ref{rownania_ASD_struny_typuCmm_dotA}) with $m^{\dot{A}}$ gives
\begin{subequations}
\begin{eqnarray}
\label{rownania_ASD_struny_typuCmm_z_1}
&& zz_{y} - z_{x} =0
\\
\label{rownania_ASD_struny_typuCmm_z_2}
&& z_{q}-zz_{p} - z_{y} \mathcal{Y} + z \frac{\partial \mathcal{Y}}{\partial y} - \frac{\partial \mathcal{Y}}{\partial x} = 0, \ \mathcal{Y} := \mathcal{B} + 2z \mathcal{Q} + z^{2} \mathcal{A} 
\end{eqnarray}
\end{subequations}
where $z_{y} \equiv \dfrac{\partial z}{\partial y}$, $z_{x} \equiv \dfrac{\partial z}{\partial x}$, $z_{q} \equiv \dfrac{\partial z}{\partial q}$, $z_{p} \equiv \dfrac{\partial z}{\partial p}$. The expansion $M^{A}$ reads
\begin{subequations}
\begin{eqnarray}
\label{M_1_dla_degxany}
&& \frac{1}{\sqrt{2}} \, M_{1} = - z_{y}
\\ 
\label{M_2_dla_degxany}
&& \frac{1}{\sqrt{2}} \, M_{2} = -z_{p} - \frac{\partial}{\partial x} (\mathcal{Q} + z \mathcal{A}) + z \frac{\partial}{\partial y} (\mathcal{Q} + z \mathcal{A})  - z_{y} (\mathcal{Q} + z \mathcal{A}) 
\end{eqnarray}
\end{subequations}
Finally, according to (\ref{properties_of_the_null_geodesics}) one finds
\begin{equation}
\label{properties_of_I_CmA_CmdotA}
\textrm{properties of  } \mathcal{I} (\mathcal{C}^{n}_{m^{A}}, \mathcal{C}^{e}_{m^{\dot{A}}}): \ \ \theta, \varrho \sim z_{y}
\end{equation}

\subsubsection{ASD congruence $\mathcal{C}_{n^{\dot{A}}}$}
\label{the_seond_ASD_congruence_section}

Consider now the second ASD $\mathcal{C}^{e}_{n^{\dot{A}}}$ generated by a spinor $n^{\dot{A}}$ such that $n^{\dot{A}} m_{\dot{A}} \ne 0 \ \longrightarrow \ n_{\dot{1}} \ne 0$. Hence, the spinor $n_{\dot{A}}$ can be brought to the form $n_{\dot{A}} = [1, w]$ where $w=w(q,p,x,y)$ with the transformation formula
\begin{equation}
w' = \Delta^{-\frac{1}{2}} (q'_{q} w - q'_{p})
\end{equation}
ASD null string equations for $\mathcal{C}^{e}_{n^{\dot{A}}}$ have the same form as (\ref{rownania_ASD_struny_typuCmm_dotA}) but with $m^{\dot{A}} \rightarrow n^{\dot{A}}$ and $M^{A} \rightarrow N^{A}$. Explicitly
\begin{subequations}
\label{rownania_ASD_struny_typuCmn_dotA}
\begin{eqnarray}
\label{rownania_ASD_struny_typuCmn_dotA_1}
 && w_{y}-ww_{x}= 0
\\ 
\label{rownania_ASD_struny_typuCmn_dotA_2}
&& w_{p}-ww_{q} + \frac{\partial \mathcal{Z}}{\partial y} - w \frac{\partial \mathcal{Z}}{\partial x} + w_{x} \mathcal{Z}= 0, \ \mathcal{Z} := \mathcal{A} + 2w \mathcal{Q} + w^{2} \mathcal{B}
\end{eqnarray}
\end{subequations}
and 
\begin{subequations}
\begin{eqnarray}
\label{N_1_dla_degxdeg}
&& \frac{1}{\sqrt{2}} \, N_{1} = w_{x}
\\ 
\label{N_2_dla_degxdeg}
&& \frac{1}{\sqrt{2}} \, N_{2} = w_{q} + w \frac{\partial}{\partial x} (\mathcal{Q} + w \mathcal{B}) - \frac{\partial}{\partial y} (\mathcal{Q} + w \mathcal{B}) - w_{x} (\mathcal{Q} + w \mathcal{B})
\end{eqnarray}
\end{subequations}
Moreover
\begin{equation}
\label{properties_of_I_CmA_CndotA}
\textrm{properties of  } \mathcal{I} (\mathcal{C}^{n}_{m^{A}}, \mathcal{C}^{e}_{n^{\dot{A}}}): \ \ \theta, \varrho \sim w_{x}
\end{equation}

%$$$$$$$$$$$$$$$$$$$$$$$$$$$$$$$$$$$$$$$$$$$$$$$$$$$$$$$$$$$$$$$$$$$$$

\section{Two-sided (complex) Walker spaces}
\label{Sekcja_przestrzenie_Walkera}
\setcounter{equation}{0}

\subsection{Spaces of the types $[\textrm{deg}]^{n} \otimes [\textrm{any}]^{e}$ and $[\textrm{deg}]^{n} \otimes [\textrm{deg}]^{n}$}
\label{subsekcja_degn_x_dege}

Assume that weak nonexpanding $\mathcal{HH}$-space is equipped with $\mathcal{C}^{e}_{m^{\dot{A}}}$. From (\ref{properties_of_I_CmA_CmdotA}) it follows that further steps depend on $z_{y}$. If $z_{y} \ne 0$ the space is of the type $\{ [\textrm{deg}]^{n} \otimes [\textrm{any}]^{e}, [++] \}$ while for $z_{y} =0$ we arrive at the space of the type $\{ [\textrm{deg}]^{n} \otimes [\textrm{any}]^{e}, [--] \}$.

\subsubsection{Spaces of the types $\{ [\textrm{deg}]^{n} \otimes [\textrm{any}]^{e}, [++] \}$ }
\label{subsekcja_degn_x_dege_plus_plus}

\begin{Twierdzenie} 
\label{theorem_degn_x_anye_plus_plus}
Let $(\mathcal{M}, ds^{2})$ be a complex (neutral) space of the type $\{ [\textrm{deg}]^{n} \otimes [\textrm{any}]^{e}, [++] \}$. Then there exists a local coordinate system $(q,p,x,z)$ such that the metric takes the form
\begin{eqnarray}
\label{twierdzenie_metryka_degn_x_anye_plus_plus}
\frac{1}{2} ds^{2} &=& - dpdx -z \, dq dx - (x - \Sigma_{z} ) \, dqdz + \mathcal{A} \, dp^{2} 
\\ \nonumber
&& +(\Sigma_{p} - 2 \mathcal{Q} ) \, dpdq + ((x-\Sigma_{z}) \, \Omega  + z \Sigma_{p} -2z  \mathcal{Q} - z^{2} \mathcal{A}) \, dq^{2}
\end{eqnarray}
where $\mathcal{A}=\mathcal{A} (q,p,x,z)$, $\mathcal {Q} = \mathcal {Q}(q,p,x,z) $, $\Sigma=\Sigma (q,p,z)$ and $\Omega=\Omega (q,p,z)$ are arbitrary holomorphic (real smooth) functions.
\end{Twierdzenie}
\begin{proof}
$\theta \ne 0$ and $\varrho \ne 0$ imply $z_{y} \ne 0$ (compare (\ref{properties_of_I_CmA_CmdotA})). A general solution of (\ref{rownania_ASD_struny_typuCmm_z_1}) is obtained by multiplying (\ref{rownania_ASD_struny_typuCmm_z_1}) by $dx \wedge dy \wedge dp \wedge dq$, treating $z$ as an independent variable and $y$ as a function of $(q,p,x,z)$. Finally
\begin{equation}
\label{general_solution_of_z_equation}
y = -xz + \Sigma (q,p,z)
\end{equation}
where $\Sigma$ is an arbitrary function of its variables. Hence
\begin{equation}
z_{x} = \frac{z}{\Sigma_{z}-x}, \ z_{y}  = \frac{1}{\Sigma_{z}-x}, \ z_{q} = -\frac{\Sigma_{q}}{\Sigma_{z}-x}, \ 
z_{p}  = -\frac{\Sigma_{p}}{\Sigma_{z}-x}
\end{equation}
Formula (\ref{general_solution_of_z_equation}) suggests a coordinate transformation $(q,p,x,y) \rightarrow (q,p,x,z)$. Denote
\begin{equation}
\tilde{\mathcal{Y}} = \tilde{\mathcal{Y}} (q,p,x,z) = \tilde{\mathcal{Y}} (q,p,x,y(x,z,q,p)) = \mathcal{Y} (q,p,x,y)
\end{equation}
Transformation of Eq. (\ref{rownania_ASD_struny_typuCmm_z_2}) to the coordinate system $(q,p,x,z)$ yields
\begin{equation}
\tilde{\mathcal{Y}} - (x-\Sigma_{z}) \tilde{\mathcal{Y}}_{x} - z \Sigma_{p} + \Sigma_{q} = 0
\end{equation}
with a general solution
\begin{equation}
\tilde{\mathcal{Y}} = (x-\Sigma_{z}) \Omega + z \Sigma_{p} - \Sigma_{q} 
\end{equation}
where $\Omega=\Omega (q,p,z)$ is an arbitrary function. The definition of $\mathcal{Y}$ (\ref{rownania_ASD_struny_typuCmm_z_2}) implies
\begin{equation}
\label{general_solution_of_Y_equation}
 \mathcal{B} + 2z \mathcal{Q} + z^{2} \mathcal{A} = (x-\Sigma_{z}) \Omega + z \Sigma_{p} - \Sigma_{q} 
\end{equation}
Thus, the solution for $\mathcal{B}$ is given in terms of $\mathcal{A}$, $\mathcal{Q}$, $\Sigma$ and $\Omega$. Inserting (\ref{general_solution_of_z_equation}) and (\ref{general_solution_of_Y_equation}) into (\ref{metryka_slaba_HH_alternative_form}) one arrives at (\ref{twierdzenie_metryka_degn_x_anye_plus_plus}).
\end{proof}
The metric (\ref{twierdzenie_metryka_degn_x_anye_plus_plus}) can be written in a simpler form but we do not enter this problem here. 

\textbf{Remark}. In \cite{Law} the authors found a general metric for a space equipped with SD $\mathcal{C}_{m^{A}}^{n}$ and ASD $\mathcal{C}_{m^{\dot{A}}}^{e}$ with an additional assumption of an integrability of a 3-dimensional distribution given by the Pfaff system $m_{A}m_{\dot{A}} g^{A\dot{A}}=0$. They called such a space \textsl{an integrable (sesquiWalker) $\alpha\beta$-geometry}. The adjective "integrable" in the definition given by Law and Matsushita is equivalent to the vanishing of the twist of $\mathcal{I} (\mathcal{C}^{n}_{m^{A}}, \mathcal{C}^{e}_{m^{\dot{A}}})$, i.e., $z_{y}=0$. Hence, the metric found in \cite{Law} is, in fact, the metric of a space of the type $\{ [\textrm{deg}]^{n} \otimes [\textrm{any}]^{e}, [--] \}$. We consider such spaces in the next Section.

\subsubsection{Spaces of the types $\{ [\textrm{deg}]^{n} \otimes [\textrm{deg}]^{e}, [--] \}$ }
\label{subsekcja_degn_x_dege_minus_minus}

\begin{Twierdzenie} 
\label{theorem_degn_x_anye_minus_minus}
Let $(\mathcal{M}, ds^{2})$ be a complex (neutral) space of the type $\{ [\textrm{deg}]^{n} \otimes [\textrm{any}]^{e}, [--] \}$. Then there exists a local coordinate system $(q,p,x,y)$ such that the metric takes the form
\begin{equation}
\label{twierdzenie_metryka_degn_x_anye_minus_minus}
\frac{1}{2} ds^{2} =  dqdy-dpdx + \mathcal{A} \, dp^{2} - 2 \mathcal{Q} \, dqdp + \mathcal{B} \, dq^{2}
\end{equation}
where $\mathcal{A}=\mathcal{A} (q,p,x,y)$, $\mathcal {Q} = \mathcal {Q}(q,p,x,y) $ and $\mathcal {B} = \mathcal {B}(q,p,y) $ are arbitrary holomorphic (real smooth) functions such that $\mathcal{Q}_{x} \ne 0$.
\end{Twierdzenie}
\begin{proof}
If $\theta=\varrho=0$ then $z_{y}=0$. From (\ref{rownania_ASD_struny_typuCmm_z_1}) one finds $z_{x}=0$ and consequently, $z=z(q,p)$.
Hence, $z$ can be gauged away without any loss of generality (compare (\ref{transformacja_na_z})). From (\ref{rownania_ASD_struny_typuCmm_z_2}) one finds $\mathcal{B}=\mathcal{B} (q,p,y)$. From (\ref{M_2_dla_degxany}) it follows that $M_{2} = - \sqrt{2} \mathcal{Q}_{x}$ so $\mathcal{Q}_{x} \ne 0$, otherwise $\mathcal{C}^{e}_{m^{\dot{A}}}$ becomes $\mathcal{C}^{n}_{m^{\dot{A}}}$.
\end{proof}
\textbf{Remark}. The metric (\ref{twierdzenie_metryka_degn_x_anye_minus_minus}) is exactly the metric found in \cite{Law}.

\subsubsection{Spaces of the types $ [\textrm{deg}]^{n} \otimes [\textrm{deg}]^{n}$ }
\label{subsekcja_degn_x_degn}

\begin{Twierdzenie} 
\label{main_theorem_degn_x_degn}
Let $(\mathcal{M}, ds^{2})$ be a complex (neutral) space of the type $[\textrm{deg}]^{n} \otimes [\textrm{deg}]^{n}$. Then there exists a local coordinate system $(q,p,x,y)$ such that the metric takes the form
\begin{equation}
\label{twierdzenie_metryka_Walker_1}
\frac{1}{2} ds^{2} =  dqdy-dpdx + \mathcal{A}  \, dp^{2} - 2 \mathcal{Q}  \, dqdp + \mathcal{B}  \, dq^{2}
\end{equation}
where $\mathcal{A}=\mathcal{A}(q,p,x,y)$, $\mathcal {Q} =\mathcal {Q}(q,p,y)$ and $\mathcal {B}=\mathcal {B}(q,p,y)$ are arbitrary holomorphic (real smooth) functions.
\end{Twierdzenie}
\begin{proof}
Consider the metric (\ref{twierdzenie_metryka_degn_x_anye_minus_minus}). Vanishing of the expansion of $\mathcal{C}^{e}_{m^{\dot{A}}}$ implies $M_{2} = - \sqrt{2} \mathcal{Q}_{x} =0$. Hence, $\mathcal {Q} =\mathcal {Q}(q,p,y)$.
\end{proof}
\textbf{Remark}. Theorem \ref{main_theorem_degn_x_degn} has been proven earlier in \cite{Chudecki_Przanowski_Walkery} but the proof presented here is much more concise.

In the next Section the metric (\ref{twierdzenie_metryka_Walker_1}) will be treated as a "starting point" for further considerations. Thus we point out that from now on the gauge (\ref{gauge}) is restricted to the transformations such that
\begin{equation}
\label{gauge_restricted}
q'=q'(q,p), \ p'=p'(p)
\end{equation}

\subsection{Spaces of the types $ [\textrm{deg}]^{n} \otimes [\textrm{deg}]^{ne}$ and $[\textrm{deg}]^{n} \otimes [\textrm{D}]^{nn}$}
\label{subsekcja_deg_nx_D_nn}

\subsubsection{Spaces of the types $\{ [\textrm{deg}]^{n} \otimes [\textrm{deg}]^{ne}, [--,++] \}$ and \\ $\{ [\textrm{deg}]^{n} \otimes [\textrm{deg}]^{ne}, [--,--] \}$ }
\label{subsekcja_degn_x_degne_obaprzypadki}

We consider now a two-sided Walker space with the general metric (\ref{twierdzenie_metryka_Walker_1}) and we equip this space with the second ASD $\mathcal{C}$, namely $\mathcal{C}^{e}_{n^{\dot{A}}}$ (see Section \ref{the_seond_ASD_congruence_section}). The next steps are analogous like in Section \ref{subsekcja_degn_x_dege}. First we consider a case with $w_{x} \ne 0$. This case corresponds to spaces of the types $\{ [\textrm{deg}]^{n} \otimes [\textrm{deg}]^{ne}, [--,++] \}$. Then a case with $w_{x}=0$ (what implies $w=0$) but with $N_{2} \ne 0$ leads to spaces of the types $\{ [\textrm{deg}]^{n} \otimes [\textrm{deg}]^{ne}, [--,--] \}$. 
\sloppy
\begin{Twierdzenie} 
\label{theorem_degn_x_degne_minisminus_plusplus}
Let $(\mathcal{M}, ds^{2})$ be a complex (neutral) space of the type $\{ [\textrm{deg}]^{n} \otimes [\textrm{deg}]^{ne}, [--,++] \}$. Then there exists a local coordinate system $(q,p,w,y)$ such that the metric takes the form
\begin{eqnarray}
\label{twierdzenie_metrykadegn_x_degne_minisminus_plusplus}
\frac{1}{2} ds^{2} &=&  dqdy+ w \, dpdy + (y-\Sigma_{w}) \, dpdw + \mathcal{B} \, dq^{2}
\\ \nonumber
&& - (2 \mathcal{Q} + \Sigma_{q}) \, dqdp + ( (y-\Sigma_{w}) \Omega - w \Sigma_{q} -2w \mathcal{Q} - w^{2} \mathcal{B}   ) \, dp^{2}
\end{eqnarray}
where $\mathcal{Q} = \mathcal{Q} (q,p,y)$, $\mathcal{B} = \mathcal{B} (q,p,y)$, $\Omega=\Omega(q,p,w)$ and $\Sigma=\Sigma(q,p,w)$ are arbitrary holomorphic (real smooth) functions.
\begin{proof}
We skip the proof due to its similarity to that of Theorem \ref{theorem_degn_x_anye_plus_plus}.
\end{proof}
\end{Twierdzenie}
\begin{Twierdzenie} 
\label{theorem_degn_x_degne_minisminus_minusminus}
Let $(\mathcal{M}, ds^{2})$ be a complex (neutral) space of the type $\{ [\textrm{deg}]^{n} \otimes [\textrm{deg}]^{ne}, [--,--] \}$. Then there exists a local coordinate system $(q,p,x,y)$ such that the metric takes the form
\begin{equation}
\label{twierdzenie_metrykadegn_x_degne_minisminus_minisminus}
\frac{1}{2} ds^{2} =  dqdy-dpdx  + \mathcal{A}  \, dp^{2} - 2 \mathcal{Q}  \, dqdp + \mathcal{B}  \, dq^{2}
\end{equation}
where $\mathcal{A}=\mathcal{A} (q,p,x)$, $\mathcal{Q}=\mathcal{Q} (q,p,y)$ and $\mathcal{B}=\mathcal{B} (q,p,y)$ are arbitrary holomorphic (real smooth) functions such that $\mathcal{Q}_{y} \ne 0$.
\begin{proof}
We skip the proof due to its similarity to that of Theorem \ref{theorem_degn_x_anye_minus_minus}.
\end{proof}
\end{Twierdzenie}

\textbf{Remark}. Spaces with metrics (\ref{twierdzenie_metrykadegn_x_degne_minisminus_plusplus}) and (\ref{twierdzenie_metrykadegn_x_degne_minisminus_minisminus}) are equipped with $\mathcal{C}^{ne}$. It was proven in \cite{Przanowski_Formanski_Chudecki} that if an Einstein space is equipped with SD (ASD) $\mathcal{C}^{ne}$ then its SD (ASD) Weyl spinor must vanish. Thus, the traceless Ricci tensor of not conformally flat spaces equipped with SD (or ASD) $\mathcal{C}^{ne}$ must be nonzero. We believe, that metrics (\ref{twierdzenie_metrykadegn_x_degne_minisminus_plusplus}) and (\ref{twierdzenie_metrykadegn_x_degne_minisminus_minisminus}) could be the first explicit examples of not conformally flat spaces equipped with $\mathcal{C}s$ with "mixed" properties and of the same duality.

Note, that the metric (\ref{twierdzenie_metrykadegn_x_degne_minisminus_minisminus}) remains invariant under the transformation (\ref{gauge}) restricted to
\begin{equation}
\label{gauge_restricted_2}
q'=q'(q), \ p'=p'(p)
\end{equation}

\subsubsection{Spaces of the types $ [\textrm{deg}]^{n} \otimes [\textrm{D}]^{nn}$}
\label{subsekcja_deg_nx_D_nn_subsekcja}

In this Section we finally arrive at the metric of the para-Kähler space of the type $[\textrm{deg}]^{n} \otimes [\textrm{D}]^{nn}$. 
\begin{Twierdzenie} 
\label{main_theorem_degn_x_degnn}
Let $(\mathcal{M}, ds^{2})$ be a complex (neutral) space of the type $[\textrm{deg}]^{n} \otimes [\textrm{D}]^{nn}$ ($[\textrm{deg}]^{n} \otimes [\textrm{D}_{r}]^{nn}$). Then there exists a local coordinate system $(q,p,x,y)$ such that the metric takes the form
\begin{equation}
\label{twierdzenie_metryka_Walker_2}
\frac{1}{2} ds^{2} =  dqdy-dpdx + \mathcal{A}  \, dp^{2}  + \mathcal{B}   \, dq^{2}
\end{equation}
where $\mathcal{A}=\mathcal{A} (q,p,x)$ and $\mathcal {B}=\mathcal{B}  (q,p,y)$ are arbitrary holomorphic (real smooth) functions.
\end{Twierdzenie}
\begin{proof}
Consider the metric (\ref{twierdzenie_metrykadegn_x_degne_minisminus_minisminus}). Vanishing of the expansion of the $\mathcal{C}^{e}_{n^{\dot{A}}}$ implies $N_{2} =- \sqrt{2} \mathcal{Q}_{y} =0$ (compare (\ref{N_2_dla_degxdeg})). Hence, $\mathcal{Q}=\mathcal{Q} (q,p)$. The transformation law for $\mathcal{Q}$ follows from (\ref{transformacja_Q}) and (\ref{gauge_restricted_2}) and it yields
\begin{equation}
\Delta \mathcal{Q}' = \mathcal{Q} + \frac{1}{2} \frac{dq'}{dq} \, \frac{\partial \sigma^{\dot{2}}}{\partial p} - \frac{1}{2} \frac{dp'}{dp} \, \frac{\partial \sigma^{\dot{1}}}{\partial q}
\end{equation}
Hence, $\mathcal{Q}$ can be gauged away without any loss of generality.
\end{proof}

Nonzero conformal curvature coefficients, the traceless Ricci tensor and the curvature scalar of the metric (\ref{twierdzenie_metryka_Walker_2}) in Plebański tetrad read
\begin{eqnarray}
\label{deg_n_x_D_nn_allformulas}
&& C^{(3)} = \frac{R}{6} = 2 C_{\dot{1}\dot{1}\dot{2}\dot{2}} = -\frac{1}{3} (\mathcal{A}_{xx} + \mathcal{B}_{yy}), \ 
C^{(2)} = -\mathcal{A}_{qx} - \mathcal{B}_{py}
\\ \nonumber
&& \frac{1}{2} C^{(1)} = -\mathcal{B}_{pp} - \mathcal{A}_{qq} + \mathcal{B}_{p}\mathcal{A}_{x} - \mathcal{B}_{y}\mathcal{A}_{q}
\\ \nonumber
&& C_{12\dot{1}\dot{2}} = \frac{1}{4} (\mathcal{A}_{xx} - \mathcal{B}_{yy}), \ C_{22\dot{1}\dot{2}} = \frac{1}{2} (\mathcal{A}_{xq} - \mathcal{B}_{yp})
\end{eqnarray}
Coefficient $C_{12\dot{1}\dot{2}}$ has a transparent geometrical meaning. If $C_{12\dot{1}\dot{2}} \ne 0$ the traceless Ricci tensor has four eigenvectors and two double eigenvalues. If $C_{12\dot{1}\dot{2}} = 0$ the traceless Ricci tensor has two eigenvectors and one quadruple eigenvalue (see \cite{Chudecki_struny}).

The metric (\ref{twierdzenie_metryka_Walker_2}) admits coordinate gauge freedom
\begin{eqnarray}
\label{ograniczone_transformacje}
&& q'=q'(q), \ p'=p'(p) , \ \frac{d q'}{dq} =: \mu(q) , \ \frac{dp'}{dp} =: \nu(p)
\\ \nonumber
&& x'=\frac{1}{\nu} \, x + \frac{1}{\nu} \frac{\partial \sigma}{\partial p} , \ y'=\frac{1}{\mu} \, y + \frac{1}{\mu} \frac{\partial \sigma}{\partial q}
\end{eqnarray}
where $\sigma=\sigma(q,p)$, $\mu=\mu(q)$ and $\nu=\nu(p)$ are arbitrary functions. Under (\ref{ograniczone_transformacje}) the functions $\mathcal{A}$ and $\mathcal{B}$ transform as follows
\begin{equation}
\label{transformacje_na_funkcje_AB}
 \mathcal{A}' = \frac{1}{\nu^{2}} \, \mathcal{A} - \frac{\nu_{p}}{\nu^{3}} \, x + \frac{1}{\nu} \frac{\partial}{\partial p} \left( \frac{1}{\nu}  \frac{\partial \sigma}{\partial p} \right) , \
  \mathcal{B}' = \frac{1}{\mu^{2}} \, \mathcal{B} + \frac{\mu_{q}}{\mu^{3}} \, y - \frac{1}{\mu} \frac{\partial}{\partial q} \left( \frac{1}{\mu}  \frac{\partial \sigma}{\partial q} \right) 
  \end{equation}

Let us discuss briefly possible Petrov-Penrose types of the metric (\ref{twierdzenie_metryka_Walker_2}).

\textbf{Case $\mathcal{A}_{xx} + \mathcal{B}_{yy} \ne 0$}. If $\mathcal{A}_{xx} + \mathcal{B}_{yy} \ne 0$ the metric (\ref{twierdzenie_metryka_Walker_2}) is of the type $[\textrm{II}]^{n} \otimes [\textrm{D}]^{nn}$ or $[\textrm{D}]^{n} \otimes [\textrm{D}]^{nn}$. The type $[\textrm{D}]^{n} \otimes [\textrm{D}]^{nn}$ is given by the condition
\begin{equation}
\label{warunek_na_typ_D}
2 C^{(2)} C^{(2)} - 3 C^{(3)} C^{(1)}=0 
\end{equation}
Eq. (\ref{warunek_na_typ_D}) written explicitly yields
\begin{equation}
\label{warunek_na_typ_D_2_jawnie}
( \mathcal{A}_{qx} + \mathcal{B}_{py})^{2} - (\mathcal{A}_{xx} + \mathcal{B}_{yy}) (\mathcal{B}_{pp} + \mathcal{A}_{qq} - \mathcal{B}_{p}\mathcal{A}_{x} + \mathcal{B}_{y}\mathcal{A}_{q}) = 0
\end{equation}

Interesting task is to find a solution of the condition (\ref{warunek_na_typ_D_2_jawnie}) such that a space is still of the type $[\textrm{D}]^{n} \otimes [\textrm{D}]^{nn}$, i.e., there is no additional SD $\mathcal{C}$ (such an additional structure is considered in Section \ref{sekcja_czwarta_struna}). Namely, to find a solution of the condition (\ref{warunek_na_typ_D_2_jawnie}) such that Eqs. (\ref{equations_for_the_second_SD_congruence}) are not satisfied. This problem is quite advanced and we do not consider it here. 

A special solution of Eq. (\ref{warunek_na_typ_D_2_jawnie}) is given by $\mathcal{A}=0$, $\mathcal{B}_{p} = g (\mathcal{B}_{y}, q)$ where $g$ is an arbitrary function. In this case Eq. (\ref{equations_for_the_second_SD_congruence_2}) is identically satisfied, but (\ref{equations_for_the_second_SD_congruence_1}) is not. Consequently, some $g$s lead to the type $[\textrm{D}]^{n} \otimes [\textrm{D}]^{nn}$ solutions but we were not able to find any explicit example.

\textbf{Case $\mathcal{A}_{xx} + \mathcal{B}_{yy} = 0$}. If $\mathcal{A}_{xx} + \mathcal{B}_{yy} = 0$ then $C^{(3)} = 0$ but at the same time $C_{\dot{1}\dot{1}\dot{2}\dot{2}} =0$. Then a space becomes SD space of the type $[\textrm{III}]^{n} \otimes [\textrm{O}]^{n}$ or $[\textrm{N}]^{n} \otimes [\textrm{O}]^{n}$. Such spaces are considered in Section \ref{subsekcja_III_NXnic}.

\subsection{Spaces of the types $[\textrm{II,D}]^{ne} \otimes [\textrm{D}]^{nn}$ and $[\textrm{D}]^{nn} \otimes [\textrm{D}]^{nn}$}
\label{sekcja_czwarta_struna}

\subsubsection{The second congruence of SD null strings}

In this Section the structure of a space will be specialized deeper. The metric (\ref{twierdzenie_metryka_Walker_2}) is equipped with three $\mathcal{C}s$:
\begin{eqnarray}
\label{trzy_struny_przepisy}
 \textrm{SD} && \mathcal{C}^{n}_{m^{A}} \textrm{ generated by the spinor } m_{A}, \ m_{2} \ne 0
\\ \nonumber
 \textrm{ASD} && \mathcal{C}^{n}_{m^{\dot{A}}} \textrm{ generated by the spinor } m_{\dot{A}}, \ m_{\dot{2}} \ne 0
\\ \nonumber
\textrm{ASD} && \mathcal{C}^{n}_{n^{\dot{A}}} \textrm{ generated by the spinor } n_{\dot{A}}, \ n_{\dot{1}} \ne 0
\end{eqnarray}
Now we assume the existence of the second SD $\mathcal{C}$. Let $\mathcal{C}_{n^{A}}$ be generated by the spinor $n_{A}$ such that $m^{A} n_{A} \ne 0 \ \Longrightarrow \ n_{1} \ne 0$. Hence, the spinor $n_{A}$ can be re-scaled to the form $n_{A} = [1,n]$. The transformation law for $n$ reads (compare (\ref{transformacja_spinorowa}))
\begin{equation}
\label{transformacja_na_h_tilda}
n'=\Delta^{-\frac{1}{2}} (n-\sigma_{pq})
\end{equation}
SD null string equations $n^{B} \nabla_{A \dot{M}} n_{B} = n_{A} N_{\dot{M}}$ written explicitly yield
\begin{subequations}
\label{equations_for_the_second_SD_congruence}
\begin{eqnarray}
\label{equations_for_the_second_SD_congruence_1}
&& n_{q}   - n_{y} \mathcal{B} - \mathcal{B}_{p} + n \mathcal{B}_{y} = n n_{x}
\\ 
\label{equations_for_the_second_SD_congruence_2}
&& n_{p}  + n_{x} \mathcal{A} + \mathcal{A}_{q} - n \mathcal{A}_{x} = n n_{y}
\end{eqnarray}
\end{subequations}
In general, $\mathcal{C}_{n^{A}}$ is expanding and the expansion is given by the formula
\begin{equation}
\label{ekspansion_of_the_second_SD_congruenceee}
N_{\dot{M}} = \sqrt{2} \frac{\partial n}{\partial p^{\dot{M}}}
\end{equation}
It it easy to check, that the properties of the intersection of $\mathcal{C}_{n^{A}}$ with ASD $\mathcal{C}s$ read
\begin{eqnarray}
\textrm{properties of } \mathcal{I}(\mathcal{C}_{n^{A}}, \mathcal{C}_{m^{\dot{A}}}): && \theta, \varrho \sim m^{\dot{A}} N_{\dot{A}} \sim \frac{\partial n}{\partial x}
\\ \nonumber
\textrm{properties of } \mathcal{I}(\mathcal{C}_{n^{A}}, \mathcal{C}_{n^{\dot{A}}}): && \theta, \varrho \sim n^{\dot{A}} N_{\dot{A}} \sim \frac{\partial n}{\partial y}
\end{eqnarray}
Hence, the space is equipped with four $\mathcal{I}s$ with the properties (listed in order given by (\ref{symbol_na_CC_CC})):
\begin{eqnarray}
\label{warrrrunki_na_dodatkowe_wlasnosci_czterystruny}
[--,--,++,++] & \textrm{if} & n_{x} \ne 0, n_{y} \ne 0
\\ \nonumber
[--,--,--,++] & \textrm{if} & n_{x} = 0, n_{y} \ne 0 
\\ \nonumber
[--,--,++,--] & \textrm{if} & n_{x} \ne 0, n_{y} = 0 
\\ \nonumber
[--,--,--,--] & \textrm{if} & n_{x} = n_{y} = 0
\\ \nonumber
\end{eqnarray}
The cases $[--,--,--,++]$ and $[--,--,++,--]$ are equivalent. In the case $[--,--,--,--]$, $\mathcal{C}_{n^{A}}$ is nonexpanding.

The integrability condition of the system (\ref{equations_for_the_second_SD_congruence})
 (see, e.g., \cite{Chudecki_struny}) yields
\begin{equation}
\label{warunek_calkowalnosci_rownan_czwartej_struny}
C^{(1)} - 4C^{(2)} \, n + 6 C^{(3)} \, n^{2} = 0
\end{equation}
Eq. (\ref{warunek_calkowalnosci_rownan_czwartej_struny}) is a quadratic equation for $n$. If $C^{(3)} \ne 0$ one finds a general solution in the form
\begin{equation}
\label{rozwiazanie_n_pm}
n_{\pm} = \frac{C^{(2)}}{3 C^{(3)}} \pm \frac{\sqrt{2 \delta}}{6 C^{(3)}}, \ \delta := 2 C^{(2)} C^{(2)} - 3 C^{(3)} C^{(1)}
\end{equation}
Note, that vanishing of a discriminant of the quadratic equation (\ref{warunek_calkowalnosci_rownan_czwartej_struny}) is equivalent to the type-D condition (\ref{warunek_na_typ_D}). Hence, if $\delta \ne 0$ and both $n_{-}$ and $n_{+}$ solve the system (\ref{equations_for_the_second_SD_congruence}) one arrives at a space of the type $[\textrm{II}]^{nee} \otimes [\textrm{D}]^{nn}$. This is an extremely interesting case with three SD $\mathcal{C}s$ and two ASD $\mathcal{C}s$, but - as we mentioned earlier - we do not consider it in this paper. If $n_{-}$ or $n_{+}$ solve the system (\ref{equations_for_the_second_SD_congruence}) then a space is of the type $[\textrm{II}]^{ne} \otimes [\textrm{D}]^{nn}$. If $\delta=0$ then there is only one double solution of Eq. (\ref{warunek_calkowalnosci_rownan_czwartej_struny}) and a space is of the type $[\textrm{D}]^{ne} \otimes [\textrm{D}]^{nn}$. Unfortunately, after putting $n_{\pm}$ into the system (\ref{equations_for_the_second_SD_congruence}) one gets very complicated equations for $\mathcal{A}$ and $\mathcal{B}$. These equations will be solved in a few special cases in the next Sections.

If $C^{(3)} =0$, we obtain
\begin{equation}
\label{n_na_typ_III}
n = \frac{C^{(1)}}{4 C^{(2)}}
\end{equation}
Solutions of the system (\ref{equations_for_the_second_SD_congruence}) with $n$ given by (\ref{n_na_typ_III}) together with $C^{(3)} =0$ imply that a space is of the type $[\textrm{III}]^{ne} \otimes [\textrm{O}]^{n}$. This case is considered in Section \ref{subsekcja_III_ne_nanic}.

\subsubsection{Spaces of the types $\{ [\textrm{II,D}]^{ne} \otimes [\textrm{D}]^{nn}, [--,--,++,++] \}$}

In this case $n_{x} \ne 0$ and $n_{y} \ne 0$. Such a case will be considered elsewhere.

\subsubsection{Spaces of the types $\{ [\textrm{II}]^{ne} \otimes [\textrm{D}]^{nn}, [--,--,--,++] \}$}
\label{subsection_II_D_mmmmmmpp}

\begin{Twierdzenie} 
\label{theorem_II_D_mmmmmmpp}
Let $(\mathcal{M}, ds^{2})$ be a complex (neutral) space of the type $\{ [\textrm{II}]^{ne} \otimes [\textrm{D}]^{nn}, [--,--,--,++] \}$ ($\{ [\textrm{II}_{r}]^{ne} \otimes [\textrm{D}_{r}]^{nn}, [--,--,--,++] \}$). Then there exists a local coordinate system $(q,p,x,n)$ such that the metric takes the form
\begin{equation}
\label{IIxD_mmmmmmpp}
\frac{1}{2} ds^{2} = -dpdx - n \, dqdp - p \, dq dn + B \, dq dn + A(p-B) \, dq^{2}
\end{equation}
where $A=A (q,n)$ and $B=B  (q,n)$ are arbitrary holomorphic (real smooth) functions such that $A B_{n} + B_{q} \ne 0$ and $A_{nn} (B-p)^{2} +(B-p) \, \frac{\partial}{\partial n} (A B_{n}+B_{q} ) - B_{n} (A B_{n} +B_{q})  \ne 0$.
\end{Twierdzenie}
\begin{proof}
From (\ref{warrrrunki_na_dodatkowe_wlasnosci_czterystruny}) it follows that $n=n(q,p,y)$ holds. Differentiating (\ref{equations_for_the_second_SD_congruence_2}) with respect to $x$ one gets $\mathcal{A}_{qx}=n\mathcal{A}_{xx}$. Because $n_{y} \ne 0$ then $\mathcal{A}_{xx}=\mathcal{A}_{qx}=0$. Hence, $\mathcal{A}$ can be gauged away without any loss of generality (compare (\ref{transformacje_na_funkcje_AB})). With $\mathcal{A}=0$, Eq. (\ref{equations_for_the_second_SD_congruence_2}) reduces to $n_{p}=nn_{y}$ with an implicit solution
\begin{equation}
\label{warunek_na_uwiklanie_rozwiazania}
y = -pn+ S (q,n)
\end{equation}
where $S=S(q,n)$ is an arbitrary function. We treat $n$ as an independent variable and $y$ as a function, $y=y(q,p,n)$. Note that
\begin{equation}
n_{y} = \frac{1}{S_{n}-p}, \ n_{p} = \frac{n}{S_{n}-p}, \ n_{q} =- \frac{S_{q}}{S_{n}-p}
\end{equation}
Let us denote
\begin{equation}
\widetilde{\mathcal{B}} = \widetilde{\mathcal{B}} (q,p,n) = \widetilde{\mathcal{B}} (q,p,y (q,p,n)) = \mathcal{B} (q,p,y)
\end{equation}
Hence
\begin{equation}
 \mathcal{B}_{y} = \widetilde{\mathcal{B}}_{n} n_{y} , \ \mathcal{B}_{p} = \widetilde{\mathcal{B}}_{p} + \widetilde{\mathcal{B}}_{n} n_{p}
\end{equation}
Eq. (\ref{equations_for_the_second_SD_congruence_1}) written in the coordinate system $(q,p,x,n)$ yields
\begin{equation}
S_{q} + \widetilde{\mathcal{B}} + (S_{n}-p) \, \widetilde{\mathcal{B}}_{p}=0
\end{equation}
with a general solution
\begin{equation}
\label{rozwiazanie_na_B_case_IInexDnn}
\widetilde{\mathcal{B}} = A(q,n) \, ( p  -  S_{n} ) - S_{q}
\end{equation}
where $A=A(q,n)$ is an arbitrary function. Inserting (\ref{warunek_na_uwiklanie_rozwiazania}), $\mathcal{A}=0$ and (\ref{rozwiazanie_na_B_case_IInexDnn}) into (\ref{twierdzenie_metryka_Walker_2}), denoting $B(q,n) := S_{n}$ one arrives at (\ref{IIxD_mmmmmmpp}).

The SD conformal curvature coefficients and the traceless Ricci tensor of the metric (\ref{IIxD_mmmmmmpp}) read
\begin{eqnarray}
\nonumber
&& C^{(3)} = \frac{1}{3(B-p)^{3}} \left( A_{nn} (B-p)^{2} +(B-p) \, \frac{\partial}{\partial n} (A B_{n}+B_{q} ) - B_{n} (A B_{n} +B_{q})  \right) 
\\ \label{C3_w_typie_mmmmmmpp}
&& C^{(2)} = 3n C^{(3)} + \frac{A  B_{n} + B_{q}}{(B-p)^{2}}, \ C^{(1)} = 4n C^{(2)} - 6n^{2} C^{(3)}
\\ \nonumber
&& C_{12 \dot{1}\dot{2}} = \frac{3}{4} C^{(3)}, \ C_{22 \dot{1}\dot{2}} = \frac{1}{2} C^{(2)}
\end{eqnarray}
Function $\delta$ defined by (\ref{rozwiazanie_n_pm}) takes the form
\begin{equation}
\delta = 2 \left( \frac{A B_{n} + B_{q}}{(B-p)^{2}}  \right)^{2}
\end{equation}
The SD Weyl spinor is of the type [II] iff $C^{(3)} \ne 0$ and $\delta \ne 0$. 
\end{proof}

\subsubsection{Spaces of the types $\{ [\textrm{D}]^{ne} \otimes [\textrm{D}]^{nn}, [--,--,--,++] \}$}
\label{subsection_D_D_mmmmmmpp}

\begin{Twierdzenie} 
\label{theorem_D_D_mmmmmmpp}
Let $(\mathcal{M}, ds^{2})$ be a complex (neutral) space of the type $\{ [\textrm{D}]^{ne} \otimes [\textrm{D}]^{nn}, [--,--,--,++] \}$ ($\{ [\textrm{D}_{r}]^{ne} \otimes [\textrm{D}_{r}]^{nn}, [--,--,--,++] \}$). Then there exists a local coordinate system $(q,p,x,z)$ such that the metric takes the form
\begin{equation}
\label{metryka_DxD}
\frac{1}{2} ds^{2} = -dp dx - F \, dq dp + (z-p) F_{z} \, dqdz
\end{equation}
where $F=F(q,z)$ is an arbitrary holomorphic (real smooth) function such that $\partial_{z} \partial_{q} \ln F_{z} \ne 0$.
\end{Twierdzenie}
\begin{proof}
The metric (\ref{IIxD_mmmmmmpp}) reduces to type $\{ [\textrm{D}]^{ne} \otimes [\textrm{D}]^{nn}, [--,--,--,++] \}$ iff $\delta=0$. Hence 
\begin{equation}
\label{constraint_na_typ_D_mmmmmmpp}
A B_{n} + B_{q}=0
\end{equation}
Multiplying (\ref{constraint_na_typ_D_mmmmmmpp}) by $dn \wedge dq$, treating $B$ as a new variable and $n$ as a function, $n=n(q,B)$, one arrives at the solution $A=n_{q}$. If we denote $B \rightarrow z$ and $n \rightarrow F$, we arrive at the metric (\ref{metryka_DxD}).

The SD conformal curvature coefficients read
\begin{equation}
C^{(3)} = \frac{\partial_{z} \partial_{q} \ln F_{z}}{3(z-p)F_{z}}, \ C^{(2)} = 3F C^{(3)}, \ C^{(1)} =  6F^{2} C^{(3)}
\end{equation}
The SD Weyl spinor is of the type [D] iff $\partial_{z} \partial_{q} \ln F_{z} \ne 0$.
\end{proof}

\subsubsection{Spaces of the types $ [\textrm{D}]^{nn} \otimes [\textrm{D}]^{nn} $}
\label{subsection_D_D_mmmmmmmm}

\begin{Twierdzenie} 
\label{theorem_D_D_mmmmmmmm}
Let $(\mathcal{M}, ds^{2})$ be a complex (neutral) space of the type $ [\textrm{D}]^{nn} \otimes [\textrm{D}]^{nn}$ ($ [\textrm{D}_{r}]^{nn} \otimes [\textrm{D}_{r}]^{nn}$). Then there exists a local coordinate systems $(q,p,x,y)$ such that the metric takes the form
\begin{equation}
\label{DxD_two_sided_conformally_recurrent}
\frac{1}{2} ds^{2} =   dqdy-dpdx + \mathcal{A}  \, dp^{2}  + \mathcal{B}  \, dq^{2}
\end{equation}
where $\mathcal{A} = \mathcal{A} (x,p)$ and $\mathcal{B} =\mathcal{B} (y,q)$ are arbitrary holomorphic (real smooth) functions such that $\mathcal{A}_{xx} + \mathcal{B}_{yy} \ne 0$.
\end{Twierdzenie}
\begin{proof}
Vanishing of the expansion of $\mathcal{C}^{e}_{n^{A}}$ implies $n=n(q,p)$ (compare (\ref{ekspansion_of_the_second_SD_congruenceee})). Consequently, from (\ref{transformacja_na_h_tilda}) it follows that $n$ can be gauged away. It leaves us with the conditions $\mathcal{B}_{p}=\mathcal{A}_{q}=0$ (compare (\ref{equations_for_the_second_SD_congruence})). Hence, $\mathcal{B}=\mathcal{B}(q,y)$ and $\mathcal{A}=\mathcal{A}(p,x)$. The SD conformal curvature coefficients read
\begin{equation}
C^{(3)} = -\frac{1}{3} (\mathcal{A}_{xx} + \mathcal{B}_{yy}), \ C^{(2)}=C^{(1)}=0
\end{equation}
so the SD Weyl spinor is of the type [D] iff $\mathcal{A}_{xx} + \mathcal{B}_{yy} \ne 0$.
\end{proof}
After suitable transformation of the variables the metric (\ref{DxD_two_sided_conformally_recurrent}) can be brought to the form
\begin{equation}
\label{DxD_two_sided_conformally_recurrent_second_form}
\frac{1}{2} ds^{2} = \mathcal{\hat{A}} (p,\hat{x}) \, dp d \hat{x}  + \mathcal{\hat{B}} (q,\hat{y}) \, dq d\hat{y}
\end{equation}
The metric (\ref{DxD_two_sided_conformally_recurrent}) ((\ref{DxD_two_sided_conformally_recurrent_second_form}), alternatively) is a well known metric with an interesting property: it is  two-sided conformally recurrent (see, e.g., \cite{Plebanski_Przanowski_rec}).

\color{black}

\subsection{Spaces of the types $[\textrm{III,N}]^{n} \otimes [\textrm{O}]^{n}$}
\label{subsekcja_III_NXnic}

\subsubsection{Types $[\textrm{III,N}]^{n} \otimes [\textrm{O}]^{n}$}
\label{subsekcja_III_NXnic_alesubsekcja}

The standard approach to SD solutions via weak nonexpanding $\mathcal{HH}$-spaces uses the formula (\ref{krzywizna}) \cite{Finley_Plebanski_All,Plebanski_Przanowski_rec,Chudecki_conformally_recurrent}. By demanding that  $C_{\dot{A}\dot{B}\dot{C}\dot{D}}=0$ a solution for $Q^{\dot{A}\dot{B}}$ can be obtained, but there are 15 arbitrary functions of two variables in this solution. Obviously, such an approach generates plenty of arbitrary functions. A number of them is gauge-dependent but it is not so straightforward to prove it. 

There is an interesting geometrical reason why the "$C_{\dot{A}\dot{B}\dot{C}\dot{D}}=0$" - approach is not the optimal one for SD spaces. Weak nonexpanding $\mathcal{HH}$-spaces are equipped with a single SD $\mathcal{C}^{n}$ and there are no ASD $\mathcal{C}s$ in general. However, self-duality condition $C_{\dot{A}\dot{B}\dot{C}\dot{D}}=0$ is equivalent to the existence of infinitely many distinct ASD $\mathcal{C}s$. The problem is that with $C_{\dot{A}\dot{B}\dot{C}\dot{D}}=0$, Plebański tetrad is still adapted to the SD $\mathcal{C}^{n}$ but it is not adapted to the ASD $\mathcal{C}s$ at all.

Our construction presented in Section \ref{subsekcja_deg_nx_D_nn} has an important advantage. We assumed the existence of ASD $\mathcal{C}^{nn}$ and we adapted Plebański tetrad to these two ASD $\mathcal{C}s$. Thus, the metric (\ref{twierdzenie_metryka_Walker_2}) is a better starting point for finding SD metrics then the metric (\ref{metryka_slaba_HH}) because it has been already adapted to the pair of ASD $\mathcal{C}s$. If we additionally demand $C_{\dot{A}\dot{B}\dot{C}\dot{D}}=0$ a solution for $Q^{\dot{A}\dot{B}}$ is no longer full of functions which are arbitrary but gauge-dependent. 

If $C_{\dot{1}\dot{1}\dot{2}\dot{2}}=0$ in (\ref{deg_n_x_D_nn_allformulas}) one arrives at spaces of the types $[\textrm{III,N}]^{n} \otimes [\textrm{O}]^{n}$. These are the only possible SD spaces equipped with SD $\mathcal{C}^{n}$ (SD spaces equipped with SD $\mathcal{C}^{e}$ will be considered elsewhere). Because in spaces of the types $[\textrm{III,N}]^{n} \otimes [\textrm{O}]^{n}$ there are infinitely many $\mathcal{I}s$ it does not make any sense to list their properties. 
\begin{Twierdzenie} 
\label{theorem_vacuum_IIIxO}
Let $(\mathcal{M}, ds^{2})$ be a complex (neutral) space of the type $[\textrm{III}]^{n} \otimes [\textrm{O}]^{n}$ ($[\textrm{III}_{r}]^{n} \otimes [\textrm{O}_{r}]^{n}$). Then there exists a local coordinate system $(q,p,x,y)$ such that the metric takes the form
\begin{equation}
\label{ogolna_metryka_vacuum_IIIxO}
\frac{1}{2} ds^{2} =  dqdy-dpdx + (M x^{2} + P x +\Omega) \, dp^{2}  + (- M y^{2} + N y) \,  dq^{2}
\end{equation}
where $M=M(q,p)$, $P=P(q,p)$, $\Omega = \Omega (q,p)$ and $N=N(q,p)$ are arbitrary holomorphic (real smooth) functions such that $2M_{p} \, y -2M_{q} \, x -N_{p} - P_{q} \ne 0$.
\end{Twierdzenie}
\begin{proof}
The ASD Weyl spinor of the metric (\ref{twierdzenie_metryka_Walker_2}) vanishes iff $C_{\dot{1}\dot{1}\dot{2}\dot{2}}=0$ what implies 
\begin{equation}
\label{solution_of_A_and_B_for_SD}
\mathcal{A} = M x^{2} + P x +\Omega , \ \mathcal{B} = - M y^{2} + N y + S
\end{equation}
where $M$, $P$, $\Omega$, $N$ and $S$ are arbitrary functions of variables $(q,p)$. Function $S$ can be gauged away without any loss of generality (compare (\ref{transformacje_na_funkcje_AB})). Let us collect all the formulas for the SD curvature and traceless Ricci tensor. Inserting (\ref{solution_of_A_and_B_for_SD}) into (\ref{deg_n_x_D_nn_allformulas}) one obtains
\begin{eqnarray}
\label{krzywizna_i_Ricci_foor_SD}
 C^{(2)} &=& 2M_{p} \, y -2M_{q} \, x -N_{p} - P_{q}
\\ \nonumber
 \frac{1}{2} C^{(1)} &=& xy( 2MM_{q}x-2MM_{p}y +2MN_{p}+2MP_{q}) + y^{2} (M_{pp} - PM_{p}) 
\\ \nonumber
&&- x^{2} (M_{qq}+NM_{q})+ y(- N_{pp} +P N_{p}+ 2M \Omega_{q} )
\\ \nonumber
&&   - x(P_{qq} + NP_{q}) -N\Omega_{q} -\Omega_{qq}
\\ \nonumber
 C_{12\dot{1}\dot{2}} &=& M
\\ \nonumber
 C_{22\dot{1}\dot{2}} &=& M_{p} \, y +M_{q} \, x - \frac{1}{2}N_{p} + \frac{1}{2}P_{q}
\end{eqnarray}
Condition $C^{(2)} \ne 0$ implies $2M_{p} \, y -2M_{q} \, x -N_{p} - P_{q} \ne 0$.
\end{proof}

\begin{Twierdzenie} 
\label{theorem_vacuum_NxO}
Let $(\mathcal{M}, ds^{2})$ be a complex (neutral) space of the type $[\textrm{N}]^{n} \otimes [\textrm{O}]^{n}$ ($[\textrm{N}_{r}]^{n} \otimes [\textrm{O}_{r}]^{n}$). Then there exists a local coordinate system $(q,p,x,y)$ such that the metric takes the form
\begin{equation}
\label{ogolna_metryka_vacuum_NxO}
\frac{1}{2} ds^{2} =  dqdy-dpdx + (M_{0} x^{2} + \Sigma_{p} \, x +\Omega) \, dp^{2}  - ( M_{0} y^{2} +\Sigma_{q} \, y) \,  dq^{2}
\end{equation}
where $M_{0}$ is a constant, $\Omega = \Omega (q,p)$ and $\Sigma = \Sigma (q,p)$ are arbitrary holomorphic (real smooth) functions such that $y(\Sigma_{qpp} - \Sigma_{p} \Sigma_{qp}+ 2M_{0} \Omega_{q} )
 - x(\Sigma_{pqq} -\Sigma_{q}\Sigma_{pq}) + \Sigma_{q} \Omega_{q} -\Omega_{qq} \ne 0$.
\end{Twierdzenie}
\begin{proof}
Condition $C^{(2)}=0$ implies $M=M_{0}=\textrm{const}$ and $N_{p} + P_{q} = 0$. Thus, there is a function $\Sigma$ such that $N=-\Sigma_{q}$ and $P=\Sigma_{p}$. Formulas (\ref{krzywizna_i_Ricci_foor_SD}) simplify to the form
\begin{eqnarray}
\label{krzywizna_i_Ricci_foor_SD_typ_N}
 \frac{1}{2} C^{(1)} &=&  y(\Sigma_{qpp} - \Sigma_{p} \Sigma_{qp}+ 2M_{0} \Omega_{q} )
 - x(\Sigma_{pqq} -\Sigma_{q}\Sigma_{pq}) + \Sigma_{q} \Omega_{q} -\Omega_{qq}
\\ \nonumber
 C_{12\dot{1}\dot{2}} &=& M_{0}
\\ \nonumber
 C_{22\dot{1}\dot{2}} &=&   \Sigma_{qp} 
\end{eqnarray}
Condition $C^{(1)} \ne 0$ implies $y(\Sigma_{qpp} - \Sigma_{p} \Sigma_{qp}+ 2M_{0} \Omega_{q} )
 - x(\Sigma_{pqq} -\Sigma_{q}\Sigma_{pq}) + \Sigma_{q} \Omega_{q} -\Omega_{qq} \ne 0$.
\end{proof}

\textbf{Remark}. Spaces of the type $ [\textrm{N}]^{n} \otimes [\textrm{O}]^{n}$ are two-sided conformally recurrent. Such spaces have been considered by Plebański and Przanowski in \cite{Plebanski_Przanowski_rec}. Plebański and Przanowski listed three classes of such solutions which depend on five arbitrary functions of two variables each. Our result shows that all these metrics can be reduced to the metric (\ref{ogolna_metryka_vacuum_NxO}) with two functions of two variables and one constant. Consequently, our approach is a serious improvement of results published in \cite{Plebanski_Przanowski_rec}.

\subsubsection{Type $[\textrm{III}]^{ne} \otimes [\textrm{O}]^{n}$}
\label{subsekcja_III_ne_nanic}

The only additional structure with which a space of the type $[\textrm{III}]^{n} \otimes [\textrm{O}]^{n}$ can be equipped is one more SD $\mathcal{C}$. This second congruence must be $\mathcal{C}^{e}$. With such a congruence an algebraic reduction to the type $[\textrm{N}]^{n} \otimes [\textrm{O}]^{n}$ is not possible anymore. Thus, we assume the existence of the function $n$ such that it has the form (\ref{n_na_typ_III}) with $C^{(1)}$ and $C^{(2)}$ given by (\ref{krzywizna_i_Ricci_foor_SD}). Inserting $n$ into Eqs. (\ref{equations_for_the_second_SD_congruence}) one finds that $M=M_{0}=\textrm{const}$ and
\begin{eqnarray}
\label{uklad_rownan_na_typ_IIIxnic_zdrugastruna}
&& 2(ba_{q} - ab_{q}) - 4a^{2} N_{p} - 2M_{0} af + cb=0
\\ \nonumber
&& 2(ac_{q}-ca_{q}) + 2Nac - c^{2} =0
\\ \nonumber
&& 2(af_{q} - fa_{q}) + 2Naf-fc=0
\\ \nonumber
&& 2(ba_{p} - ab_{p}) - 4M_{0} a^{2} \Omega + 2P ba - b^{2} = 0
\\ \nonumber
&& 2(ac_{p} - ca_{p}) + 4a^{2} P_{q} - 2 M_{0} af + cb =0
\\ \nonumber
&& 2(a f_{p} - f a_{p}) + 4a^{2} \Omega_{q} + 2 \Omega ac - 2P fa +fb = 0
\end{eqnarray}
where we denoted
\begin{equation}
a:= N_{p} + P_{q}, \ b := P N_{p} - N_{pp} + 2M_{0} \Omega_{q}, \ c:= NP_{q} + P_{qq}, \ f:= \Omega_{qq} + N \Omega_{q}
\end{equation}
(\ref{uklad_rownan_na_typ_IIIxnic_zdrugastruna}) is a system of six equations for three functions $N$, $P$ and $\Omega$ of two variables $(p,q)$. This system is surprisingly complicated. We were able to find only three different special solutions
\begin{eqnarray}
\label{jawne_przyklady_rozwiazan}
(i) && N=0, \ M_{0} \textrm{ is arbitrary}, \ P = \frac{4}{4p-q}, \ \Omega = \xi P - \xi_{p} - M_{0} \xi^{2}, \ \xi=\xi (p) \ \ \ \ 
\\ \nonumber
(ii) && N=M_{0} = \Omega=0, \ P=\frac{4}{4p-q} + \xi, \ \xi=\xi(p)
\\ \nonumber
(iii) && P=\Omega=0, \ M_{0} \textrm{ is arbitrary}, \ N = \frac{\xi_{q}}{p- \xi}, \ \xi=\xi(q), \ \xi_{q} \ne 0
\end{eqnarray}
where $\xi$ is an arbitrary function of its variable. Hence, the metric (\ref{ogolna_metryka_vacuum_IIIxO}) with (\ref{jawne_przyklady_rozwiazan}) is an example of the type $[\textrm{III}]^{ne} \otimes [\textrm{O}]^{n}$.

% &&&&&&&&&&&&&&&&&&&&&&&&&&&&&&&&&&&&&&&&&&&&&&&&&&&&&&&&&&&&&&&&&&&&&&&&&&&&&&&&&&&&&&&&

\section{Para-Kähler Einstein spaces}
\setcounter{equation}{0}
\label{sekcja_rozwiazania_Einsteinowskie}

\subsection{General case}

From Theorem \ref{main_theorem_degn_x_degnn} and from the formulas (\ref{deg_n_x_D_nn_allformulas}) one easily obtains Einstein spaces of the types $[\textrm{deg}]^{n} \otimes [\textrm{D}]^{nn}$. With $C_{12\dot{1}\dot{2}}=C_{22\dot{1}\dot{2}}=0$ and $R=-4 \Lambda$ one arrives at the solution
\begin{equation}
\label{general_vacuum_solution}
\mathcal{A} = \frac{\Lambda}{2} x^{2} + \Phi_{p} \, x + \Omega, \ 
\mathcal{B} = \frac{\Lambda}{2} y^{2} + \Phi_{q} \, y + \Sigma
\end{equation} 
where $\Phi (q,p)$, $\Sigma (q,p)$ and $\Omega (q,p)$ are arbitrary functions of their variables. Hence, the general metric of the type $[\textrm{deg}]^{n} \otimes [\textrm{D}]^{nn}$ Einstein spaces can be brought to the form 
\begin{equation}
\label{twierdzenie_metryka_typ_IIIn_x_nic_Einstein}
\frac{1}{2} ds^{2} =  dqdy-dpdx + \left( \frac{\Lambda}{2} x^{2} + \Phi_{p} \, x + \Omega \right) dp^{2}  + \left( \frac{\Lambda}{2} y^{2} + \Phi_{q} \, y + \Sigma \right)  dq^{2}
\end{equation}
The SD conformal curvature coefficients read
\begin{eqnarray}
\label{EEEinstein_spaces_ogolnakrzywizna}
&& C^{(3)} = 2 C_{\dot{1}\dot{1}\dot{2}\dot{2}} = -\frac{2}{3} \Lambda, \ C^{(2)} = -2 \Phi_{pq}
\\ \nonumber
&& \frac{1}{2} C^{(1)} = (\Phi_{p}\Phi_{pq} - \Phi_{qpp} - \Lambda \Omega_{q} ) y - (\Phi_{q}\Phi_{pq} + \Phi_{pqq} - \Lambda \Sigma_{p} ) x - \Phi_{q} \Omega_{q} + \Phi_{p} \Sigma_{p} - \Omega_{qq} - \Sigma_{pp}
\end{eqnarray}

[\textbf{Remark}. The solution (\ref{general_vacuum_solution}) corresponds to the $\mathcal{HH}$-space generated by the key function
\begin{equation}
\label{funkcja_kluczowa}
\Theta = - \frac{\Lambda}{12} x^{2} y^{2} - \frac{1}{6} \Phi_{p} \, xy^{2} - \frac{1}{6} \Phi_{q} \, y x^{2} - \frac{1}{2} \Sigma x^{2} - \frac{1}{2} \Omega y^{2} + \alpha_{\dot{A}} p^{\dot{A}} + \beta
\end{equation}
where $\alpha_{\dot{A}}$ and $\beta$ are arbitrary functions of $(q,p)$. Also,
\begin{equation}
\label{funkcje_strukturalne_nieekspandujace_hiperniebo}
F^{\dot{1}} = \Phi_{p}, \ F^{\dot{2}} = \Phi_{q}, \ N_{\dot{1}} = \Sigma_{p}, \ N_{\dot{2}} = - \Omega_{q}, \ \gamma = -\Lambda \beta + \Sigma \Omega - \frac{\partial \alpha^{\dot{M}}}{\partial q^{\dot{M}}} + \alpha_{\dot{M}} F^{\dot{M}}
\end{equation}
The \textsl{structural functions} $F^{\dot{A}}$, $N^{\dot{A}}$ and $\gamma$ are well-known in $\mathcal{HH}$-spaces theory \cite{Chudecki_Killingi_2}. The form of the key function (\ref{funkcja_kluczowa}) is used in Sections \ref{subsekcja_typyD_II_symetrie} and \ref{subsekcja_typyIII_N_symetrie} in which para-Kähler Einstein spaces with symmetries are considered.]

Under (\ref{gauge}) functions $\Phi$, $\Sigma$ and $\Omega$ transform as follows
\begin{eqnarray}
\label{transformacje_na_przypadekEinsteinowski}
\Phi' &=& \Phi - \Lambda \sigma + \ln \frac{\mu}{\nu} + \Phi_{0}
\\ \nonumber
\mu^{2} \, \Sigma' &=& \Sigma -\Phi'_{q} \sigma_{q} - \frac{\Lambda}{2} \sigma_{q}^{2} - \mu \frac{\partial}{\partial{q}} \left( \frac{\sigma_{q}}{\mu}   \right)
\\ \nonumber
\nu^{2} \, \Omega' &=& \Omega -\Phi'_{p} \sigma_{p} - \frac{\Lambda}{2} \sigma_{p}^{2} + \nu \frac{\partial}{\partial{p}} \left( \frac{\sigma_{p}}{\nu}   \right)
\end{eqnarray}
where $\sigma=\sigma(q,p)$, $\mu=\mu (q)$ and $\nu = \nu (p)$ are arbitrary gauge functions and $\Phi_{0}$ is an arbitrary gauge constant. Hence, $\Phi$, $\Sigma$ or $\Omega$ can be gauged away but for different Petrov-Penrose types different choices are optimal.

\subsection{Types $ [\textrm{II}]^{n} \otimes [\textrm{D}]^{nn}$ and $ [\textrm{D}]^{nn} \otimes [\textrm{D}]^{nn}$}

\subsubsection{General case}
\label{typy_IInxDnn_alesubsekcja}

\begin{Twierdzenie} 
\label{theorem_vacuum_IIxD}
Let $(\mathcal{M}, ds^{2})$ be complex (neutral) Einstein space of the type $ [\textrm{II}]^{n} \otimes [\textrm{D}]^{nn}$ or $ [\textrm{D}]^{nn} \otimes [\textrm{D}]^{nn}$ ($ [\textrm{II}_{r}]^{n} \otimes [\textrm{D}_{r}]^{nn}$, $ [\textrm{II}_{rc}]^{n} \otimes [\textrm{D}_{r}]^{nn}$ or $ [\textrm{D}_{r}]^{nn} \otimes [\textrm{D}_{r}]^{nn}$). Then there exists a local coordinate system $(q,p,x,y)$ such that the metric takes the form
\begin{equation}
\label{ogolna_metryka_vacuum_IIxD}
\frac{1}{2} ds^{2} =  dqdy-dpdx + \left( \frac{\Lambda}{2} x^{2}  + \Omega \right) dp^{2}  + \left( \frac{\Lambda}{2} y^{2}  + \Sigma \right)  dq^{2}
\end{equation}
where $\Omega = \Omega (q,p)$ and $\Sigma = \Sigma (q,p)$ are arbitrary holomorphic (real smooth) functions such that
\begin{eqnarray}
\nonumber
&& \textrm{for the type }  [\textrm{II}]^{n} \otimes [\textrm{D}]^{nn} \ ([\textrm{II}_{r}]^{n} \otimes [\textrm{D}_{r}]^{nn}, \ [\textrm{II}_{rc}]^{n} \otimes [\textrm{D}_{r}]^{nn}): |\Sigma_{p}| + |\Omega_{q}| \ne 0;
\\ \nonumber
&& \textrm{for the type } [\textrm{D}]^{nn} \otimes [\textrm{D}]^{nn} \ ([\textrm{D}_{r}]^{nn} \otimes [\textrm{D}_{r}]^{nn}): \Sigma=\Omega=0. 
\end{eqnarray}
\end{Twierdzenie}
\begin{proof}
The cosmological constant $\Lambda$ is necessarily nonzero so $\Phi$ can be gauged away (compare (\ref{transformacje_na_przypadekEinsteinowski})). Additionally, type-[D] condition (\ref{warunek_na_typ_D}) implies $C^{(1)}=0 \ \Longleftrightarrow \ \Sigma_{p}=\Omega_{q}=0$. Hence, $\Sigma=\Sigma (q)$ and $\Omega = \Omega (p)$. Thus, using gauge functions $\mu(q)$ and $\nu (p)$ one can always put $\Sigma=\Omega=0$. 
\end{proof}

\textbf{Remark}. Consider the metric (\ref{ogolna_metryka_vacuum_IIxD}) of the type $ [\textrm{D}]^{nn} \otimes [\textrm{D}]^{nn}$. Performing the coordinate transformation
\begin{equation}
\frac{1}{y} \ \rightarrow \ \frac{1}{y} + \frac{\Lambda}{2} \, q , \ \frac{1}{x} \ \rightarrow \ - \frac{1}{x} - \frac{\Lambda}{2} \, p
\end{equation}
one arrives at the well-known form 
\begin{equation}
\label{metryka_DnnxDnn_Einstein}
ds^{2} = \frac{2dxdp}{\left( 1 + \dfrac{\Lambda}{2} xp \right)^{2}} + \frac{2dydq}{\left( 1 + \dfrac{\Lambda}{2} yq \right)^{2}} 
\end{equation}
There are two different classes of homogeneous pKE spaces. A neutral slice of the metric (\ref{metryka_DnnxDnn_Einstein}) corresponds to one of them \cite{Bor_Makhmali_Nurowski}. The second class is so-called \textsl{dancing metric} but the dancing metric is of the type $ [\textrm{O}_{r}]^{e} \otimes [\textrm{D}_{r}]^{nn}$ and it cannot be obtained from the metric (\ref{ogolna_metryka_vacuum_IIxD}).

\subsubsection{Types $ [\textrm{II}]^{n} \otimes [\textrm{D}]^{nn}$ and $ [\textrm{D}]^{nn} \otimes [\textrm{D}]^{nn}$ with symmetries}
\label{subsekcja_typyD_II_symetrie}

Symmetries\footnote{We use the following terminology: if a vector $K$ satisfies Eqs. $\nabla_{(a} K_{b)} = \chi_{0} g_{ab}$ then $K$ is called \textsl{a homothetic vector}; if $\chi_{0} \ne 0$ then $K$ is called \textsl{a proper homothetic vector}; if $\chi_{0}=0$ the $K$ is called \textsl{a Killing vector}.} in nonexpanding $\mathcal{HH}$-spaces were analyzed in \cite{Plebanski_Finley_Killingi,Chudecki_Killingi_2}. It was proven that Eqs. $\nabla_{(a} K_{b)} = \chi_{0} g_{ab}$ in nonexpanding $\mathcal{HH}$-spaces can be reduced to a single, first order, partial, linear differential equation called \textsl{the master equation} for the key function $\Theta$. Feeding the master equation with (\ref{funkcja_kluczowa}) and (\ref{funkcje_strukturalne_nieekspandujace_hiperniebo}) with $\Phi=0$ we find that any Killing vector admitted by the metric (\ref{ogolna_metryka_vacuum_IIxD}) has the form
\begin{equation}
K = \delta^{\dot{1}} \frac{\partial}{\partial q} + \delta^{\dot{2}} \frac{\partial}{\partial p} - \left( \frac{d \delta^{\dot{2}}}{dp} \, x + \frac{1}{\Lambda} \frac{d^{2} \delta^{\dot{2}}}{dp^{2}} \right) \frac{\partial}{\partial x} - \left(  \frac{d \delta^{\dot{1}}}{dq} \, y - \frac{1}{\Lambda} \frac{d^{2} \delta^{\dot{1}}}{dq^{2}} \right) \frac{\partial}{\partial y}
\end{equation}
where $\delta^{\dot{1}} = \delta^{\dot{1}} (q)$, $\delta^{\dot{2}} = \delta^{\dot{2}} (p)$. The vector $K$ cannot be null and it cannot be proper homothetic. The master equation reduces to the system of two differential equations
\begin{eqnarray}
\label{rownanie_master_wnnioski}
&& \delta^{\dot{1}} \Sigma_{q} + \delta^{\dot{2}} \Sigma_{p} + 2 \Sigma \, \frac{d \delta^{\dot{1}}}{dq} + \frac{1}{\Lambda} \frac{d^{3} \delta^{\dot{1}}}{dq^{3}} = 0
\\ \nonumber
&& \delta^{\dot{1}} \Omega_{q} + \delta^{\dot{2}} \Omega_{p} + 2 \Omega \, \frac{d \delta^{\dot{2}}}{dp} + \frac{1}{\Lambda} \frac{d^{3} \delta^{\dot{2}}}{dp^{3}} = 0
\end{eqnarray}
Detailed analysis of (\ref{rownanie_master_wnnioski}) is long and tedious. We skip all the details and present only final results. The metric (\ref{ogolna_metryka_vacuum_IIxD}) does not admit any Killing vector in general. If an existence of a single Killing vector is assumed it can be brought to the form $K_{1} = \partial_{q}$ or $K_{1} = \partial_{q} + \partial_{p}$ without any loss of generality. The maximal number of Killing vectors admitted by the metric (\ref{ogolna_metryka_vacuum_IIxD}) of the type $ [\textrm{II}]^{n} \otimes [\textrm{D}]^{nn}$ is 2. Type $ [\textrm{D}]^{nn} \otimes [\textrm{D}]^{nn}$ is equipped with 6 Killing vectors. The results are presented in the Table \ref{Killingi_1} (all quantities in the Table \ref{Killingi_1} with a subscript $0$ are constants).

\renewcommand{\arraystretch}{1.5}
\begin{table}[ht!]
%\begin{center}
%\begin{tabular}{|c|c|}   \hline
%\afterpage{
\begin{longtable}{|c|c|}   \hline
\textrm{Killing vectors}      & Functions in the metric  \\  \hline 
\multicolumn{2}{|c|}{Type $ [\textrm{II}]^{n} \otimes [\textrm{D}]^{nn}$} \\  \hline 
$K_{1} = \partial_{q}$ & $\Sigma=\Sigma(p)$, $\Sigma_{p} \ne 0$, $\Omega=0$ \\  \hline
$K_{1} = \partial_{q}$,   &  $\Omega = 0$, $\Sigma (p) = \exp \left(  \displaystyle \int \dfrac{-2 \, dp}{\gamma_{0} p^{2} + \xi_{0} p + \zeta_{0}}  \right)$  \\ 
$K_{2} = q \partial_{q} - y \partial_{y} + \zeta_{0} \partial_{p}+ \xi_{0} (p\partial_{p} -x \partial_{x}) $ & $|\gamma_{0}| + |\xi_{0}| + |\zeta_{0}| \ne 0$ \\
$+ \gamma_{0} (p^{2} \partial_{p} - 2(px+ \Lambda^{-1}) \partial{x})$  &  if $\gamma_{0} \ne 0$  then $\gamma_{0}=1$, $\xi_{0}=0$, $\zeta_{0}$ is arbitrary    \\ 
  &   if $\gamma_{0} = 0$  and  $\xi_{0} \ne 0$ then $\zeta_{0}=0$  \\ 
    &   if $\gamma_{0} = \xi_{0} = 0$ then $\zeta_{0}=1$  \\ \hline
  $K_{1} = \partial_{q} + \partial_{p}$ & $\Sigma=\Sigma(z)$, $\Omega=\Omega (z)$, $z:=q-p$, \\ 
    &   $|\Sigma_{z}| + |\Omega_{z}| \ne 0$ \\  \hline
   $K_{1} = \partial_{q} + \partial_{p}$, & $\Sigma(z) = \Sigma_{0} z^{-2}$, $\Omega (z) = \Omega_{0} z^{-2}$, $z:=q-p$,  \\
   $K_{2} = q \partial_{q} + p \partial_{p} -x \partial_{x} -y \partial_{y}$   &  $|\Sigma_{0}| + |\Omega_{0}| \ne 0$ \\
\hline
$K_{1} = \partial_{q} + \partial_{p}$,  &  $\Sigma(z) = \dfrac{\Sigma_{0}}{(1-e^{a_{0}z})^{2}} - \dfrac{a_{0}^{2}}{2 \Lambda}$,   \\
$K_{2} = e^{a_{0} q} \left(  \partial_{q} + (-a_{0}y + a_{0}^{2} \Lambda^{-1}) \partial_{y} \right)$   & $\Omega (z) = \dfrac{\Omega_{0}}{(1-e^{a_{0}z})^{2}} - \dfrac{a_{0}^{2}}{2 \Lambda}$,   \\
\ \ \ $+e^{a_{0} p}   \left(  \partial_{p} - (a_{0} x + a_{0}^{2} \Lambda^{-1}) \partial_{x}  \right)$  & $z:=q-p$, $a_{0} \ne 0$, $|\Sigma_{0}| + |\Omega_{0}| \ne 0$   \\ \hline
$K_{1} = \partial_{q} + \partial_{p}$,  &  $\Sigma(z) = \Sigma_{0} e^{-2a_{0}z}- \dfrac{a_{0}^{2}}{2 \Lambda}$, $z:=q-p$,   \\
$K_{2} = e^{a_{0} q} \left(  \partial_{q} + (-a_{0}y + a_{0}^{2} \Lambda^{-1}) \partial_{y} \right)$   & $\Omega  = \Omega_{0}$,  $a_{0} \ne 0$, $\Sigma_{0} \ne 0$ \\ \hline 
\multicolumn{2}{|c|}{Type $ [\textrm{D}]^{nn} \otimes [\textrm{D}]^{nn}$} \\  \hline
$K_{1} = \partial_{p}$, $K_{2} = \partial_{q}$, $K_{3} = q \partial_{q} - y \partial_{y}$, & $\Sigma = \Omega = 0$ \\
$K_{4} = p \partial_{p} - x \partial_{x}$,  & \\
 $K_{5} = q^{2} \partial_{q} - 2 (qy-\Lambda^{-1} ) \partial_{y}$, &  \\
$K_{6} = p^{2} \partial_{p} - 2 (px + \Lambda^{-1} ) \partial_{x}$ & \\ \hline
%\end{tabular}
%\end{center}
\caption{Killing vectors in spaces of the types $ [\textrm{II}]^{n} \otimes [\textrm{D}]^{nn}$ or $ [\textrm{D}]^{nn} \otimes [\textrm{D}]^{nn}$.}
\label{Killingi_1}
\end{longtable}
\end{table}
%}

\subsection{Types $ [\textrm{III,N}]^{n} \otimes [\textrm{O}]^{n}$}
\label{subsekcja_III_NXnic_Einstein}

\subsubsection{General case}
\label{subsekcja_III_NXnic_Einstein_alesubsekcja}

\begin{Twierdzenie} 
\label{ttwierdzenie_metryka_typ_III_Nn_x_nic_Einstein}
Let $(\mathcal{M}, ds^{2})$ be a complex (neutral) Einstein space of the type $ [\textrm{III,N}]^{n} \otimes [\textrm{O}]^{n}$ ($ [\textrm{III}_{r},\textrm{N}_{r}]^{n} \otimes [\textrm{O}_{r}]^{n}$). Then there exists a local coordinate system $(q,p,x,y)$ such that the metric takes the form
\begin{equation}
\label{twierdzenie_metryka_typ_III_Nn_x_nic_Einstein}
\frac{1}{2} ds^{2} =  dqdy-dpdx + \left(  \Phi_{p} \, x + \Omega \right) dp^{2}  + \Phi_{q} \, y \,   dq^{2}
\end{equation}
where $\Phi = \Phi (q,p)$ and $\Omega = \Omega (q,p)$ are arbitrary holomorphic (real smooth) functions such that
\begin{eqnarray}
\nonumber
&& \textrm{for the type }  [\textrm{III}]^{n} \otimes [\textrm{O}]^{n} \ ([\textrm{III}_{r}]^{n} \otimes [\textrm{O}_{r}]^{n}): \Phi_{pq} \ne 0, \Omega \textrm{ is arbitrary};
\\ \nonumber
&& \textrm{for the type } [\textrm{N}]^{n} \otimes [\textrm{O}]^{n} \ ([\textrm{N}_{r}]^{n} \otimes [\textrm{O}_{r}]^{n}): \Phi=0, \Omega_{qq} \ne 0. 
\end{eqnarray}
\end{Twierdzenie}
\begin{proof}
If cosmological constant $\Lambda =0$ the ASD Weyl spinor vanishes and a space is of types $[\textrm{III}]^{n} \otimes [\textrm{O}]^{n}$ or $[\textrm{N}]^{n} \otimes [\textrm{O}]^{n}$ (compare (\ref{EEEinstein_spaces_ogolnakrzywizna})). For the type $[\textrm{III}]^{n} \otimes [\textrm{O}]^{n}$, $\Phi$ cannot be gauged away anymore but $\Sigma$ or $\Omega$ can. Let $\Sigma = 0$ what remains valid for both types $[\textrm{III}]^{n} \otimes [\textrm{O}]^{n}$ and $[\textrm{N}]^{n} \otimes [\textrm{O}]^{n}$. Coefficient $C^{(2)}$ is nonzero iff $\Phi_{pq} \ne 0$.

For the type $[\textrm{N}]^{n} \otimes [\textrm{O}]^{n}$ coefficient $C^{(2)}$ is zero which implies $\Phi_{pq}=0$. Consequently, $\Phi = \Phi_{1} (q) + \Phi_{2} (p)$ and using $\mu$ and $\nu$ both $\Phi_{i}$ can be gauged away. The function $\Omega$ must be such that $\Omega_{qq} \ne 0$, otherwise $C^{(1)}=0$ and the space becomes flat.
\end{proof}
\textbf{Remark}. A short historical remark about exact solutions of  algebraically degenerate heavenly spaces is needed. All metrics of SD Einstein spaces of the type $[\textrm{N}]^{n} \otimes [\textrm{O}]^{n}$ and $[\textrm{N}]^{e} \otimes [\textrm{O}]^{n}$ were found in \cite{Plebanski_further} (the metrics (5.20) and (5.24) in \cite{Plebanski_further}). A great progress in the subject was done by Fette, Janis and Newman. They found all algebraically degenerate heavenly metrics \cite{Fette_1,Fette_2}. Finally, an original approach towards algebraically degenerate heavenly spaces was used by Finley and Plebański in \cite{Finley_Plebanski_All}. Consequently, they found more compact forms of algebraically degenerate heavenly metrics.

\subsubsection{Types $[\textrm{III,N}]^{n} \otimes [\textrm{O}]^{n}$  with symmetries}
\label{subsekcja_typyIII_N_symetrie}

If the metric (\ref{twierdzenie_metryka_typ_III_Nn_x_nic_Einstein}) admits homothetic vector then it can be always brought to the form
\begin{eqnarray}
\label{Killing_vector_w_niebie}
K &=& \delta^{\dot{1}} \frac{\partial}{\partial q} + \delta^{\dot{2}} \frac{\partial}{\partial p} + \left( 2\chi_{0} x - \frac{\partial \delta^{\dot{2}}}{\partial p} x + \frac{\partial \delta^{\dot{1}}}{\partial p}y + \epsilon^{\dot{1}} \right)  \frac{\partial}{\partial x}
\\ \nonumber
 && \ \ \ \ \ \ \ \ \ \ \ \ \ \ \ \ + \left( 2\chi_{0} y + \frac{\partial \delta^{\dot{2}}}{\partial q} x - \frac{\partial \delta^{\dot{1}}}{\partial q}y + \epsilon^{\dot{2}} \right)  \frac{\partial}{\partial y}
\end{eqnarray}
where $\delta^{\dot{A}} = \delta^{\dot{A}} (p,q)$, $\epsilon^{\dot{A}} = \epsilon^{\dot{A}} (p,q)$. The master equation splits into system of equations
\begin{subequations}
\label{rownanie_master_dla_niebianskich}
\begin{eqnarray}
\label{rownanie_master_dla_niebianskich_1}
&&  \frac{\partial \delta^{\dot{1}}}{\partial p} = a_{0} \, e^{\Phi}, \ 
    \frac{\partial \delta^{\dot{2}}}{\partial q} = b_{0} \, e^{- \Phi}; \ a_{0}, b_{0} \textrm{ are constants}
   \\ \label{rownanie_master_dla_niebianskich_2}
   && \delta^{\dot{1}} \Phi_{q} + \delta^{\dot{2}} \Phi_{p} + \frac{\partial \delta^{\dot{2}}}{\partial p} - \frac{\partial \delta^{\dot{1}}}{\partial q} = \textrm{const},
   \\ \label{rownanie_master_dla_niebianskich_3}
   &&   \epsilon^{\dot{2}} \Phi_{q} = - \frac{\partial \epsilon^{\dot{2}}}{\partial q}, 
   \\ \label{rownanie_master_dla_niebianskich_4}
    && \delta^{\dot{1}} \Omega_{q} + \delta^{\dot{2}} \Omega_{p} - 2 \Omega \left( \chi_{0} - \frac{\partial \delta^{\dot{2}}}{\partial p} \right) +  \epsilon^{\dot{1}} \Phi_{p} = \frac{\partial \epsilon^{\dot{1}}}{\partial p}, 
   \\ \label{rownanie_master_dla_niebianskich_5}
   &&  2 \Omega \frac{\partial \delta^{\dot{2}}}{\partial q} = 
   \frac{\partial \epsilon^{\dot{1}}}{\partial q} - \frac{\partial \epsilon^{\dot{2}}}{\partial p}
\end{eqnarray}
\end{subequations}
with the transformation rules
\begin{eqnarray}
&&  \delta'^{\dot{1}} = \mu (q) \, \delta^{\dot{1}} , \  \delta'^{\dot{2}} = \nu (p) \, \delta^{\dot{2}} ,
\\ \nonumber
&& \nu \epsilon'^{\dot{1}} = \epsilon^{\dot{1}} - \partial_{p} (\delta^{\dot{1}} \sigma_{q} - \delta^{\dot{2}} \sigma_{p} + 2 \chi_{0} \sigma ) + 2 \sigma_{qp} \delta^{\dot{1}}
\\ \nonumber
&& \mu \epsilon'^{\dot{2}} = \epsilon^{\dot{2}} + \partial_{q} (\delta^{\dot{1}} \sigma_{q} - \delta^{\dot{2}} \sigma_{p} - 2 \chi_{0} \sigma ) + 2 \sigma_{qp} \delta^{\dot{2}}
\end{eqnarray}

\begin{Wniosek}
The metric (\ref{twierdzenie_metryka_typ_III_Nn_x_nic_Einstein}) does not admit any non-null homothetic vector in general.
\end{Wniosek}
\begin{proof}
Killing vector $K$ (\ref{Killing_vector_w_niebie}) is non-null if and only if $\delta^{\dot{A}} \ne 0$ (see Theorem 2.2 of \cite{Chudecki_null}). If $\Phi \ne 0$ then from (\ref{rownanie_master_dla_niebianskich_1}-\ref{rownanie_master_dla_niebianskich_2}) it follows that there is an algebraic condition for $\Phi$. Hence, $\Phi$ is not arbitrary anymore. Consider now the case with $\Phi=0$. Then $\delta^{\dot{A}}$ becomes linear in $p$ and $q$ and (\ref{rownanie_master_dla_niebianskich_4})-(\ref{rownanie_master_dla_niebianskich_5}) constitute algebraic conditions for $\Omega$, as a result of which $\Omega$ cannot be arbitrary anymore. Consequently, if $\delta^{\dot{A}} \ne 0$ then Eqs. (\ref{rownanie_master_dla_niebianskich}) are not satisfied for arbitrary $\Phi$ and $\Omega$. Thus, the metric (\ref{twierdzenie_metryka_typ_III_Nn_x_nic_Einstein}) does not admit any non-null homothetic vector in general.
\end{proof}

Comprehensive analysis of a symmetry algebra of the metric (\ref{twierdzenie_metryka_typ_III_Nn_x_nic_Einstein}) is outside a scope of this text. Instead, we focus on null homothetic vectors. The paper \cite{Chudecki_null} was devoted to the issue of the existence of null homothetic vectors in Einstein spaces. Here we would like to mention a few additional remarks which have not been noticed in \cite{Chudecki_null}.

The homothetic vector (\ref{Killing_vector_w_niebie}) is null if and only if $\delta^{\dot{A}}=0$. Thus, any null homothetic vector admitted by the metric (\ref{twierdzenie_metryka_typ_III_Nn_x_nic_Einstein}) takes the form
\begin{equation}
K =  \left( 2\chi_{0} x + \epsilon_{p}   \right)  \frac{\partial}{\partial x}
 + \left( 2\chi_{0} y + \epsilon_{q} \right)  \frac{\partial}{\partial y}
\end{equation}
where $\epsilon=\epsilon (p,q)$ is an arbitrary function which transforms as follows
\begin{equation}
\label{transformacja_na_epsilon_zerowedelta}
\epsilon' = \epsilon - 2\chi_{0} \sigma + \epsilon_{0}
\end{equation}
Consequently, the master equation reduces to the system of two equations 
\begin{subequations}
\label{rownanie_master_dla_niebianskich_1_null_ogolne}
\begin{eqnarray}
\label{rownanie_master_dla_niebianskich_1_null}
\epsilon_{q} \Phi_{q} = - \epsilon_{qq}
\\
\label{rownanie_master_dla_niebianskich_2_null}
-2 \chi_{0} \Omega + \epsilon_{p} \Phi_{p} = \epsilon_{pp}
\end{eqnarray}
\end{subequations}

\subsubsection{Type $[\textrm{III}]^{n} \otimes [\textrm{O}]^{n}$ with null symmetries}

For type $[\textrm{III}]^{n} \otimes [\textrm{O}]^{n}$ function $\Phi$ is nonzero. Thus, from (\ref{rownanie_master_dla_niebianskich_1_null_ogolne}) it follows that the metric (\ref{twierdzenie_metryka_typ_III_Nn_x_nic_Einstein}) does not admit any null homothetic vector in general. Existence of such a symmetry puts additional constraints on the functions $\Phi$ and $\Omega$.

\begin{Twierdzenie} 
\label{twierdzenie_metryka_typ_III_Nn_x_nic_Einstein_symetriazerowahomotetyczna}
Let $(\mathcal{M}, ds^{2})$ be a complex (neutral) Einstein space of the type $ [\textrm{III}]^{n} \otimes [\textrm{O}]^{n}$ ($ [\textrm{III}_{r}]^{n} \otimes [\textrm{O}_{r}]^{n}$) admitting null and proper homothetic vector field $K$. Then there exists a local coordinate system $(q,p,x,y)$ such that the metric takes the form
\begin{equation}
\label{twierdzenie_metryka_typ_III_Nn_x_nic_Einstein_null_proper_homothetic}
\frac{1}{2} ds^{2} =  dqdy-dpdx +   \Phi_{p} \, x \,  dp^{2}  + \Phi_{q} \, y  \,  dq^{2}
\end{equation}
where $\Phi = \Phi (q,p)$ is an arbitrary holomorphic (real smooth) function such that $\Phi_{pq} \ne 0$. Vector $K$ has the form $K = 2 \chi_{0} (x \partial_{x} + y \partial_{y})$.
\end{Twierdzenie}
\begin{proof}
Let $K$ be a null and proper homothetic vector field. With $\chi_{0} \ne 0$ one can always gauge away $\epsilon$ (compare (\ref{transformacja_na_epsilon_zerowedelta})) which implies $\Omega = 0$ (compare (\ref{rownanie_master_dla_niebianskich_2_null})). It implies that $K = 2 \chi_{0} (x \partial_{x} + y \partial_{y})$ holds and the metric (\ref{twierdzenie_metryka_typ_III_Nn_x_nic_Einstein}) reduces to the form (\ref{twierdzenie_metryka_typ_III_Nn_x_nic_Einstein_null_proper_homothetic}).
\end{proof}

The general form of a metric of a space of type $[\textrm{III}]^{n} \otimes [\textrm{O}]^{n}$ equipped with null and proper homothetic vector field was found in \cite{Chudecki_null} (it is the metric (4.14) of \cite{Chudecki_null}). However, there is a dependence on two functions of two variables in (4.14) of \cite{Chudecki_null}. Here we have shown that one of these functions is, in fact, gauge-dependent which is a significant improvement of the results of \cite{Chudecki_null}. 

\begin{Twierdzenie} 
\label{twierdzenie_metryka_typ_III_Nn_x_nic_Einstein_symetriazerowa}
Let $(\mathcal{M}, ds^{2})$ be a complex (neutral) Einstein space of the type $ [\textrm{III}]^{n} \otimes [\textrm{O}]^{n}$ ($ [\textrm{III}_{r}]^{n} \otimes [\textrm{O}_{r}]^{n}$) admitting null Killing vector field $K$. Then there exists a local coordinate system $(q,p,x,y)$ such that the metric takes the form
\begin{equation}
\label{twierdzenie_metryka_typ_III_Nn_x_nic_Einstein_symetriazerowa_1}
\frac{1}{2} ds^{2} =  -dpdx + (2p +H) \, dqdy + \left( \Omega - \frac{x}{2p+H} \right) dp^{2}
\end{equation}
where $\Omega = \Omega (q,p)$ and $H=H(q)$ are arbitrary holomorphic (real smooth) functions such that $H_{q} \ne 0$. Vector $K$ has the form $K =  \partial_{y}$.
\end{Twierdzenie}
\begin{proof}
With $\chi_{0}=0$ from (\ref{rownanie_master_dla_niebianskich_1_null_ogolne}) it follows that $\epsilon_{p} = e^{\Phi} f(q)$ and $\epsilon_{q} = e^{-\Phi} h (p)$ where $f$ and $h$ are arbitrary nonzero functions. Both $f$ and $h$ can be  brought to $1$ with help of the gauge functions $\mu$ and $\nu$ (compare (\ref{transformacje_na_przypadekEinsteinowski})). Substitution $\Phi = \ln Q$ brings the null Killing vector to the form $K = Q \partial_{x} + Q^{-1} \partial_{y}$. Eqs. (\ref{rownanie_master_dla_niebianskich_1_null_ogolne}) yield $Q^{2} Q_{q} + Q_{p}=0$. A general solution of this equation reads $q = p \, Q^{2} + G (Q)$ where $G$ is an arbitrary function such that $QG_{QQ}-G_{Q} \ne 0$ (otherwise $C^{(2)}=0$ and the SD Weyl spinor is not of the type [III] anymore). Treating $Q$ as a new variable and $q$ as a function, $q=q(p,Q)$, one finds  
\begin{equation}
Q_{q} = \frac{1}{2pQ + G_{Q}}, \ Q_{p} = - \frac{Q^{2}}{2pQ + G_{Q}}
\end{equation}
and the metric reads
\begin{equation}
\frac{1}{2} ds^{2} = dpd(-x+yQ^{2}) + \frac{2pQ + G_{Q}}{Q} (Q \, dy dQ + y \, dQ^{2}) + \left( \Omega + \frac{Q(-x +yQ^{2})}{2pQ + G_{Q}} \right) dp^{2}
\end{equation}
where $\Omega = \Omega (Q,p)$. Performing the coordinate transformation 
\begin{equation}
p \rightarrow \widetilde{p}, \ Q \rightarrow \widetilde{q}, \ x \rightarrow \widetilde{x} + \widetilde{y} \widetilde{q}, \ y \rightarrow \frac{\widetilde{y}}{\widetilde{q}},
\end{equation}
denoting $H := G_{\widetilde{q}} / \widetilde{q}$ and dropping tildes one arrives at the metric (\ref{twierdzenie_metryka_typ_III_Nn_x_nic_Einstein_symetriazerowa_1}). Also, $QG_{QQ}-G_{Q} \ne 0$ implies that $H_{\widetilde{q}}\ne 0$. The null Killing vector takes the form $K =  \partial_{\tilde{y}}$. 
\end{proof}

The metric of a space of type $ [\textrm{III}]^{n} \otimes [\textrm{O}]^{n}$ equipped with a null Killing vector was found in \cite{Chudecki_null} (the metric (5.28) in \cite{Chudecki_null}), but form (\ref{twierdzenie_metryka_typ_III_Nn_x_nic_Einstein_symetriazerowa_1}) is much more concise.

\subsubsection{Type $[\textrm{N}]^{n} \otimes [\textrm{O}]^{n}$ with null symmetries}

In this case $\Phi$ is zero. If $\chi_{0} \ne 0$ then (\ref{rownanie_master_dla_niebianskich_1_null_ogolne}) implies $\Omega_{qq}=0$ and the metric becomes flat. Thus, one has $\chi_{0}=0$. Hence, any null Killing vector field of the metric 
\begin{equation}
\label{twierdzenie_metryka_typ_III_Nn_x_nic_Einstein_null_Killing}
\frac{1}{2} ds^{2} =  dqdy-dpdx +    \Omega \,  dp^{2}  
\end{equation}
has the form $K = \epsilon_{0} ( q \partial_{x} + p \partial_{y} ) + a_{0} \partial_{x}  + b_{0} \partial_{y}$. Therefore, the metric (\ref{twierdzenie_metryka_typ_III_Nn_x_nic_Einstein_null_Killing}) is automatically equipped with three different null Killing vectors, $K_{1}=\partial_{x}$, $K_{2} = \partial_{y}$ and $K_{3}=q \partial_{x} + p \partial_{y}$.

We point out that there is an interesting geometrical difference between vectors $K_{1}$, $K_{2}$ and $K_{3}$. Any null vector has the form $K_{A\dot{A}}=k_{A} k_{\dot{A}}$. If $K_{A\dot{A}}$ is a null Killing vector field and the space is Einstein space then both $k_{A}$ and $k_{\dot{A}}$ generate a SD (ASD, respectively) $\mathcal{C}$ (see \cite{Chudecki_null}). The metric (\ref{twierdzenie_metryka_typ_III_Nn_x_nic_Einstein_null_Killing}) is equipped with a single SD $\mathcal{C}$ and infinitely many ASD $\mathcal{C}s$. The SD $\mathcal{C}$ is generated by the spinor $m^{A}$ (compare (\ref{trzy_struny_przepisy})), hence $k^{A} \sim m^{A}$ for all the vectors $K_{1}$, $K_{2}$ and $K_{3}$. Dotted spinor $k_{\dot{A}}$ generates the ASD $\mathcal{C}_{k^{\dot{A}}}$ and this particular $\mathcal{C}$ is somehow distinguished between infinitely many ASD $\mathcal{C}s$. However, for $K_{1}$ and $K_{2}$ we find that $\mathcal{C}_{k^{\dot{A}}}=\mathcal{C}^{n}$ while for $K_{3}$ the congruence $\mathcal{C}_{k^{\dot{A}}}$ is expanding, $\mathcal{C}_{k^{\dot{A}}}=\mathcal{C}^{e}$. 

Metrics of type $[\textrm{N}]^{n} \otimes [\textrm{O}]^{n}$ equipped with $K_{1}$ or $K_{3}$ were listed as two different metrics in \cite{Chudecki_null} ((5.28) and (5.33) of \cite{Chudecki_null}). Here we proved that these are, in fact, the same metrics equipped with three different null Killing vectors (although it seems difficult to show that (5.28) of \cite{Chudecki_null} can be brought to the form (5.33) of \cite{Chudecki_null}).

%#####################################################################################

\section{Concluding remarks.}

This paper is the first part of more extensive work devoted to pK and pKE-spaces. In this part we have found all pK and pKE metrics equipped with a single nonexpanding congruence of SD null strings. In our formalism these are metrics of types $[\textrm{deg}]^{n} \otimes [\textrm{D}]^{nn}$. Also, a few other interesting metrics with nonzero traceless Ricci tensor equipped with congruences of SD and ASD null strings have been constructed. The results are gathered in the Table \ref{summary} (the metrics marked by $^{*}$ have been already known, the rest of the metrics are new results).

The second part of our work will be devoted to the pK and pKE-spaces equipped with expanding congruence of SD null strings and algebraically degenerate SD Weyl spinor. These are spaces of the types $[\textrm{deg}]^{e} \otimes [\textrm{D}]^{nn}$.

The third family of pKE-spaces are those with algebraically general SD Weyl spinor which is equivalent to the lack of existence of congruences of SD null strings. These are spaces of the types $[\textrm{I}] \otimes [\textrm{D}]^{nn}$ or $[\textrm{I}] \otimes [\textrm{O}]^{n}$. Special examples of type $[\textrm{I}] \otimes [\textrm{O}]^{n}$ have been found in \cite{Boyer} but we are still far from the full solution of such a problem. Explicit examples of pKE-spaces of the type $[\textrm{I}] \otimes [\textrm{D}]^{nn}$ are even greater challenge. According to our best knowledge such examples are not known yet. It is only known that a general solution depends on two holomorphic functions of three variables each. However, a thorough analysis of the problem gives hope that an explicit example of the type $[\textrm{I}] \otimes [\textrm{D}]^{nn}$ will be constructed. This problem is now intensively studied but the results will be presented elsewhere.
 
\begin{table}[H]
 \footnotesize
\begin{center}
\begin{tabular}{|c|c|c|}   \hline
 Type  &   Metric   & Functions in the metric        \\ \hline \hline
   \multicolumn{3}{|c|}{Spaces with $C_{ab} \ne 0$}  \\ \cline{1-3}
 $ [\textrm{deg}]^{n} \otimes [\textrm{any}] $ & (\ref{metryka_slaba_HH_alternative_form})$^{*}$ & 3 functions of 4 variables     \\ \hline
 $\{ [\textrm{deg}]^{n} \otimes [\textrm{any}]^{e},[++]  \}$ & (\ref{twierdzenie_metryka_degn_x_anye_plus_plus})  & 2 functions of 4 variables, 2 functions of 3 variables     \\ \hline
 $\{ [\textrm{deg}]^{n} \otimes [\textrm{any}]^{e},[--]  \}$ & (\ref{twierdzenie_metryka_degn_x_anye_minus_minus})$^{*}$ & 2 functions of 4 variables, 1 function of 3 variables     \\ \hline
 $ [\textrm{deg}]^{n} \otimes [\textrm{deg}]^{n}$ & (\ref{twierdzenie_metryka_Walker_1})$^{*}$ & 1 function of 4 variables, 2 functions of 3 variables     \\ \hline
 $\{ [\textrm{deg}]^{n} \otimes [\textrm{deg}]^{ne},[--,++]  \}$ & (\ref{twierdzenie_metrykadegn_x_degne_minisminus_plusplus}) & 4 functions of 3 variables     \\ \hline
 $\{ [\textrm{deg}]^{n} \otimes [\textrm{deg}]^{ne},[--,--]  \}$ & (\ref{twierdzenie_metrykadegn_x_degne_minisminus_minisminus}) & 3 functions of 3 variables     \\ \hline
 $ [\textrm{deg}]^{n} \otimes [\textrm{D}]^{nn}$ & (\ref{twierdzenie_metryka_Walker_2}) & 2 functions of 3 variables     \\ \hline
 $\{ [\textrm{II}]^{ne} \otimes [\textrm{D}]^{nn}, [--,--,--,++] \}$ & (\ref{IIxD_mmmmmmpp}) & 2 functions of 2 variables     \\ \hline
 $\{ [\textrm{D}]^{ne} \otimes [\textrm{D}]^{nn}, [--,--,--,++] \}$ & (\ref{metryka_DxD}) & 1 function of 2 variables     \\ \hline
 $ [\textrm{D}]^{nn} \otimes [\textrm{D}]^{nn}$ & (\ref{DxD_two_sided_conformally_recurrent})$^{*}$ & 2 functions of 2 variables     \\ \hline
 $ [\textrm{III}]^{n} \otimes [\textrm{O}]^{n}$ & (\ref{ogolna_metryka_vacuum_IIIxO}) & 4 functions of 2 variables     \\ \hline
 $ [\textrm{N}]^{n} \otimes [\textrm{O}]^{n}$ & (\ref{ogolna_metryka_vacuum_NxO}) & 2 functions of 2 variables, 1 constant     \\ \hline
  \multicolumn{3}{|c|}{Einstein spaces}  \\ \cline{1-3}
  $ [\textrm{II}]^{n} \otimes [\textrm{D}]^{nn}$ & (\ref{ogolna_metryka_vacuum_IIxD}) & 2 functions of 2 variables, 1 constant    \\ \hline
  $ [\textrm{D}]^{nn} \otimes [\textrm{D}]^{nn}$ & (\ref{metryka_DnnxDnn_Einstein})$^{*}$ &  1 constant    \\ \hline
  $ [\textrm{III}]^{n} \otimes [\textrm{O}]^{n}$ & (\ref{twierdzenie_metryka_typ_III_Nn_x_nic_Einstein})$^{*}$ &  2 functions of 2 variables    \\ \hline
  $ [\textrm{N}]^{n} \otimes [\textrm{O}]^{n}$ & (\ref{twierdzenie_metryka_typ_III_Nn_x_nic_Einstein})$^{*}$ &  1 function of 2 variables    \\ \hline
\end{tabular}
\caption{Summary of main results.}
\label{summary}
\end{center}
\end{table}

%#####################################################################################

\appendix
\renewcommand{\theequation}{\Alph{section}.\arabic{equation}}
\setcounter{equation}{0}
%\section{\normalsize Appendix}
\section{Appendix. Classification of spaces equipped with congruences of SD and ASD null strings.}
\label{Dodatek_klasyfikacja}

In this Appendix we present a detailed classification of spaces equipped with at most two SD and two ASD $\mathcal{C}s$. Such a classification can be given in terms of the properties of $\mathcal{C}s$ and $\mathcal{I}s$  without specification of a spinorial basis (compare the Tables \ref{cztery_kongruencje} and \ref{cztery_kongruencje_przeciecia}). In the Tables \ref{Tabela_typy_one_congruence}, \ref{Tabela_typy_two_congruence} and \ref{Tabela_typy_four_congruence} we call such a classification "spinorial".

A special choice of a spinorial basis allows to adapt the null tetrad to the structure of $\mathcal{C}s$. Let the spinorial basis be chosen in such a manner that $m_{A}=[0,m]$, $n_{A} = [n,0]$, $m_{\dot{A}} = [0, \dot{m}]$ and $n_{\dot{A}} = [\dot{n},0]$. With such a choice of the spinorial basis the null tetrad is \textsl{adapted} to the $\mathcal{C}s$ in a sense that $\mathcal{C}_{m^{A}}$ is spanned by $(\partial_{2}, \partial_{4})$, $\mathcal{C}_{n^{A}}$ is spanned by $(\partial_{1}, \partial_{3})$, etc. Also, $\mathcal{I} (\mathcal{C}_{m^{A}}, \mathcal{C}_{m^{\dot{A}}}) \sim \partial_{4}$, $\mathcal{I} (\mathcal{C}_{m^{A}}, \mathcal{C}_{n^{\dot{A}}}) \sim \partial_{2}$, etc. For details see the Tables \ref{cztery_kongruencje_adapted}, \ref{cztery_kongruencje_przeciecia_adapted} and the Scheme \ref{Congruences}. In the Tables \ref{Tabela_typy_one_congruence}, \ref{Tabela_typy_two_congruence} and \ref{Tabela_typy_four_congruence} we call such a classification "tetradial". 

Note that not all the types listed in the Tables \ref{Tabela_typy_one_congruence}, \ref{Tabela_typy_two_congruence} and \ref{Tabela_typy_four_congruence} are admitted by an arbitrary space. We enumerate  only two restrictions:
\begin{eqnarray}
\nonumber
 (i) && \mathcal{C}^{en} \textrm{ are admitted only by the spaces with nonzero traceless Ricci tensor}
\\ \nonumber
 (ii) && \textrm{if the curvature scalar } R \ne 0 \textrm{ then } \mathcal{C}^{n} \textrm{ are admitted only by the spaces}
\\ \nonumber
&& \textrm{of the types [II] and [D]}
\end{eqnarray} 
An interesting question arises: are there any subtypes which cannot exist at all? We are going to dig this problem deeper.

\begin{figure}[H]
\begin{center}
\includegraphics[scale=0.8]{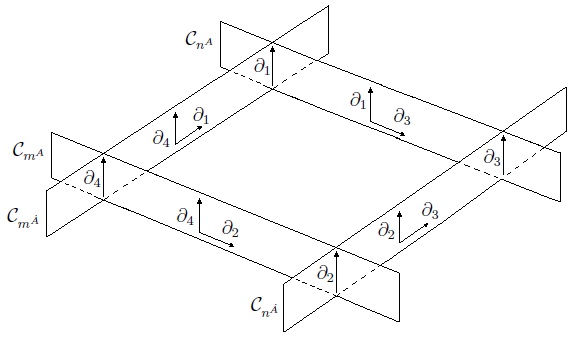}
\caption {Congruences of null strings and congruences of null geodesics in adapted null tetrad.}
\label{Congruences}
\end{center}
\end{figure}

\begin{table}[H]
\begin{center}
\begin{tabular}{|c|c|c|c|c|}   \hline
 Congruence  &   Spanned   & Spinor  & Expansion & Conditions    \\ \hline \hline
 $\mathcal{C}_{m^{A}}$ & $(\partial_{2}, \partial_{4})$ & $m_{A} = [0,m]$, $m \ne 0$  & $M_{\dot{A}} = \sqrt{2} m \, [\Gamma_{421}, -\Gamma_{423}]$  & $\Gamma_{422} = \Gamma_{424} = 0 $ \\ \hline
 $\mathcal{C}_{n^{A}}$ & $(\partial_{1}, \partial_{3})$  & $n_{A} = [n,0]$, $n \ne 0$  & $N_{\dot{A}} = \sqrt{2} n \, [\Gamma_{314}, \Gamma_{312}]$  & $\Gamma_{311} = \Gamma_{313}=0$  \\ \hline
  $\mathcal{C}_{m^{\dot{A}}}$ & $(\partial_{1}, \partial_{4})$  & $m_{\dot{A}} = [0, \dot{m}]$, $\dot{m} \ne 0$  & $M_{A} = \sqrt{2} \dot{m} \,  [\Gamma_{412}, - \Gamma_{413}]$  & $\Gamma_{411} = \Gamma_{414}=0$ \\ \hline
  $\mathcal{C}_{n^{\dot{A}}}$ & $(\partial_{2}, \partial_{3})$  & $n_{\dot{A}} = [\dot{n},0]$, $\dot{n} \ne 0$  & $N_{A} = \sqrt{2} \dot{n} \, [\Gamma_{324}, \Gamma_{321}]$  & $\Gamma_{322} = \Gamma_{323}=0$ \\ \hline
\end{tabular}
\caption{Four distinct congruences of null strings in adapted null tetrad.}
\label{cztery_kongruencje_adapted}
\end{center}
\end{table}

\begin{table}[H]
\begin{center}
\begin{tabular}{|c|c|c|c|}   \hline
 Intersection  &   Tangent vector   & Expansion   & Twist     \\ \hline \hline
 $\mathcal{I} (\mathcal{C}_{m^{A}}, \mathcal{C}_{m^{\dot{A}}})$ & $\sim \partial_{4}$ & $\theta \sim \Gamma_{412}+\Gamma_{421}$  & $\varrho \sim \Gamma_{412}-\Gamma_{421}$   \\ \hline
 $\mathcal{I} (\mathcal{C}_{m^{A}}, \mathcal{C}_{n^{\dot{A}}})$ & $\sim \partial_{2}$ & $\theta \sim \Gamma_{324}+\Gamma_{423}$  & $\varrho \sim \Gamma_{324}-\Gamma_{423}$    \\ \hline
$\mathcal{I} (\mathcal{C}_{n^{A}}, \mathcal{C}_{m^{\dot{A}}})$  & $\sim \partial_{1}$ & $\theta \sim \Gamma_{413}+\Gamma_{314}$  & $\varrho \sim \Gamma_{413}-\Gamma_{314}$    \\ \hline
 $\mathcal{I} (\mathcal{C}_{n^{A}}, \mathcal{C}_{n^{\dot{A}}})$ & $\sim \partial_{3}$ & $\theta \sim \Gamma_{321}+\Gamma_{312}$  & $\varrho \sim \Gamma_{321}-\Gamma_{312}$   \\ \hline
\end{tabular}
\caption{Four intersections of congruences of null strings in adapted null tetrad.}
\label{cztery_kongruencje_przeciecia_adapted}
\end{center}
\end{table}

\newpage
\begin{longtable}{|c|c|c|}   \hline
  Type / &  \multicolumn{2}{|c|}{Conditions}  \\ \cline{2-3}
 Subtype  &   Spinorial   & Tetradial    \\ \hline 
 \rowcolor{light_gray_1}
 Type $ [ \: \cdot \: ]^{n} \otimes [ \: \cdot \: ]^{n}$ & $M^{\dot{A}}=M^{A}=0$ & $\Gamma_{42}=\Gamma_{41}=0$  \\ \hline
$[--]$ & - & -  \\ \hline 
\rowcolor{light_gray_1}
 Type $ [  \: \cdot \: ]^{n} \otimes [ \: \cdot \: ]^{e}$ & $M^{\dot{A}}=0$, $M^{A} \ne 0$ & $\Gamma_{42}=0$, $\Gamma_{41} \ne 0$  \\ \hline
 $[--]$ & $m_{A}M^{A}=0$ & $\Gamma_{412}=0$  \\ \hline
 $[++]$ & $m_{A}M^{A} \ne 0$ & $\Gamma_{412} \ne 0$  \\ \hline 
 \rowcolor{light_gray_1}
 Type $ [ \: \cdot \: ]^{e} \otimes [ \: \cdot \: ]^{e}$ & $M^{\dot{A}}\ne 0$, $M^{A} \ne 0$ & $\Gamma_{42} \ne 0$, $\Gamma_{41} \ne 0$  \\ \hline
 $[--]$ & $m_{A}M^{A}=m_{\dot{A}}M^{\dot{A}}=0$ & $\Gamma_{412}=\Gamma_{421}=0$ \\ \hline
 $[+-]$ & $m_{A}M^{A}=m_{\dot{A}}M^{\dot{A}} \ne 0$ & $\Gamma_{412}=\Gamma_{421} \ne 0$ \\ \hline
 $[-+]$ & $m_{A}M^{A}=-m_{\dot{A}}M^{\dot{A}} \ne 0$ & $\Gamma_{412}=-\Gamma_{421} \ne 0$ \\ \hline
 $[++]$ & $m_{A}M^{A} \pm m_{\dot{A}}M^{\dot{A}} \ne 0$ & $\Gamma_{412} \pm \Gamma_{421} \ne 0$ \\ \hline 
\caption{Types of spaces equipped with one SD and one ASD congruence of null strings.}
\label{Tabela_typy_one_congruence}
\end{longtable}

%\begin{table}[ht]
%\begin{center}
%\setlength\LTleft{0pt}
%\setlength\LTright{0pt}
{  \footnotesize  
\begin{longtable}{|c|c|c|}   \hline
%\begin{longtabu} to \textwidth {|c|c|c|}
 Type / &  \multicolumn{2}{|c|}{Conditions}  \\ \cline{2-3}
 Subtype  &   Spinorial   & Tetradial    \\ \hline 
 \rowcolor{light_gray_1}
 Type $ [ \: \cdot \: ]^{n} \otimes [ \: \cdot \: ]^{nn}$ & $M^{\dot{A}}=N^{A}=M^{A}=0$ & $\Gamma_{42}=\Gamma_{32}=\Gamma_{41}=0$  \\ \hline
$[--,--]$ & - & -  \\ \hline 
\rowcolor{light_gray_1}
 Type $ [ \: \cdot \: ]^{n} \otimes [ \: \cdot \: ]^{ne}$ & $M^{\dot{A}}=M^{A} =0$, $N^{A} \ne 0$ & $\Gamma_{42}=\Gamma_{41} = 0$,        $\Gamma_{32} \ne 0$  \\ \hline
 $[--,--]$ & $m_{A}N^{A}=0$ & $\Gamma_{324}=0$  \\ \hline
 $[--,++]$ & $m_{A}N^{A} \ne 0$ & $\Gamma_{324} \ne 0$  \\ \hline 
 \rowcolor{light_gray_1}
 Type $ [ \: \cdot \: ]^{n} \otimes [ \: \cdot \: ]^{ee}$ & $M^{\dot{A}}= 0$, $M^{A} \ne 0$, $N^{A} \ne 0$ & $\Gamma_{42} = 0$, $\Gamma_{41} \ne 0$, $\Gamma_{32} \ne 0$  \\ \hline
 $[--,--]$ & $m_{A}M^{A}=m_{A}N^{A}=0$ & $\Gamma_{412}=\Gamma_{324}=0$ \\ \hline
 $[--,++]$ & $m_{A}M^{A}=0$, $m_{A}N^{A} \ne 0$ & $\Gamma_{412}=0$, $\Gamma_{324} \ne 0$ \\ \hline
 $[++,++]$ & $m_{A}M^{A} \ne 0$, $m_{A}N^{A} \ne 0$ & $\Gamma_{412} \ne 0$, $\Gamma_{324} \ne 0$ \\ \hline 
 \rowcolor{light_gray_1}
 Type $ [ \: \cdot \:]^{e} \otimes [ \: \cdot \: ]^{nn}$ & $M^{\dot{A}}\ne 0$, $M^{A} =N^{A} = 0$ & $\Gamma_{42} \ne 0$, $\Gamma_{41} = \Gamma_{32} = 0$  \\ \hline
 $[--,++]$ & $m_{\dot{A}}M^{\dot{A}}=0$, $n_{\dot{A}}M^{\dot{A}} \ne 0$ & $\Gamma_{421}=0$, $\Gamma_{423} \ne 0$  \\ \hline
 $[++,++]$ & $m_{\dot{A}}M^{\dot{A}} \ne 0$,  $n_{\dot{A}}M^{\dot{A}} \ne 0$ & $\Gamma_{421} \ne 0$, $\Gamma_{423} \ne 0$  \\ \hline 
 \rowcolor{light_gray_1}
 Type $ [ \: \cdot \: ]^{e} \otimes [ \: \cdot \: ]^{ne}$ & $M^{\dot{A}}\ne 0$, $M^{A} = 0$, $N^{A} \ne 0$ & $\Gamma_{42} \ne 0$, $\Gamma_{41} = 0$, $\Gamma_{32} \ne 0$  \\ \hline
 $[--,+-]$ & $m_{\dot{A}}M^{\dot{A}}=0$, $m_{A} N^{A} = n_{\dot{A}}M^{\dot{A}} \ne 0$ & $\Gamma_{421}=0$, $\Gamma_{324} = \Gamma_{423} \ne 0$  \\ \hline
 $[--,-+]$ & $m_{\dot{A}}M^{\dot{A}}=0$, $m_{A} N^{A} = -n_{\dot{A}}M^{\dot{A}} \ne 0$ & $\Gamma_{421}=0$, $\Gamma_{324} =- \Gamma_{423} \ne 0$  \\ \hline
 $[--,++]$ & $m_{\dot{A}}M^{\dot{A}}=0$, $m_{A} N^{A} \pm n_{\dot{A}}M^{\dot{A}} \ne 0$ & $\Gamma_{421}=0$, $\Gamma_{324} \pm \Gamma_{423} \ne 0$  \\ \hline
 $[++,--]$ & $m_{\dot{A}}M^{\dot{A}} \ne 0$, $m_{A} N^{A} = n_{\dot{A}}M^{\dot{A}} = 0$ & $\Gamma_{421} \ne 0$, $\Gamma_{324} = \Gamma_{423} = 0$  \\ \hline
 $[++,+-]$ & $m_{\dot{A}}M^{\dot{A}} \ne 0$, $m_{A} N^{A} = n_{\dot{A}}M^{\dot{A}} \ne 0$ & $\Gamma_{421} \ne 0$, $\Gamma_{324} = \Gamma_{423} \ne 0$  \\ \hline
 $[++,-+]$ & $m_{\dot{A}}M^{\dot{A}} \ne 0$, $m_{A} N^{A} =- n_{\dot{A}}M^{\dot{A}} \ne 0$ & $\Gamma_{421} \ne 0$, $\Gamma_{324} = -\Gamma_{423} \ne 0$  \\ \hline
 $[++,++]$ & $m_{\dot{A}}M^{\dot{A}} \ne 0$, $m_{A} N^{A} \pm n_{\dot{A}}M^{\dot{A}} \ne 0$ & $\Gamma_{421} \ne 0$, $\Gamma_{324} \pm \Gamma_{423} \ne 0$  \\ \hline 
 \rowcolor{light_gray_1}
 Type $ [ \: \cdot \: ]^{e} \otimes [ \: \cdot \: ]^{ee}$ & $M^{\dot{A}}\ne 0$, $M^{A} \ne 0$, $N^{A} \ne 0$ & $\Gamma_{42} \ne 0$, $\Gamma_{41} \ne 0$, $\Gamma_{32} \ne 0$  \\ \hline
 $[--,+-]$ & $m_{A}M^{A}=m_{\dot{A}}M^{\dot{A}} = 0$, $m_{A} N^{A} = n_{\dot{A}}M^{\dot{A}} \ne 0$ & $\Gamma_{412}=\Gamma_{421}=0$, $\Gamma_{324} = \Gamma_{423} \ne 0$  \\ \hline
 $[--,-+]$ & $m_{A}M^{A}=m_{\dot{A}}M^{\dot{A}} = 0$, $m_{A} N^{A} =- n_{\dot{A}}M^{\dot{A}} \ne 0$ & $\Gamma_{412}=\Gamma_{421}=0$, $\Gamma_{324} = -\Gamma_{423} \ne 0$  \\ \hline
 $[--,++]$ & $m_{A}M^{A}=m_{\dot{A}}M^{\dot{A}} = 0$, $m_{A} N^{A} \pm n_{\dot{A}}M^{\dot{A}} \ne 0$ & $\Gamma_{412}=\Gamma_{421}=0$, $\Gamma_{324} \pm \Gamma_{423} \ne 0$  \\ \hline
 $[+-,+-]$ & $m_{A}M^{A}=m_{\dot{A}}M^{\dot{A}} \ne 0$, $m_{A} N^{A} = n_{\dot{A}}M^{\dot{A}} \ne 0$ & $\Gamma_{412}=\Gamma_{421} \ne 0$, $\Gamma_{324} = \Gamma_{423} \ne 0$  \\ \hline
 $[+-,-+]$ & $m_{A}M^{A}=m_{\dot{A}}M^{\dot{A}} \ne 0$, $m_{A} N^{A} =- n_{\dot{A}}M^{\dot{A}} \ne 0$ & $\Gamma_{412}=\Gamma_{421} \ne 0$, $\Gamma_{324} = -\Gamma_{423} \ne 0$  \\ \hline
 $[-+,-+]$ & $m_{A}M^{A}=-m_{\dot{A}}M^{\dot{A}} \ne 0$, $m_{A} N^{A} = -n_{\dot{A}}M^{\dot{A}} \ne 0$ & $\Gamma_{412}=-\Gamma_{421} \ne 0$, $\Gamma_{324} = -\Gamma_{423} \ne 0$  \\ \hline
 $[++,+-]$ & $m_{A}M^{A} \pm m_{\dot{A}}M^{\dot{A}} \ne 0$, $m_{A} N^{A} = n_{\dot{A}}M^{\dot{A}} \ne 0$ & $\Gamma_{412} \pm \Gamma_{421} \ne 0$, $\Gamma_{324} = \Gamma_{423} \ne 0$  \\ \hline
 $[++,-+]$ & $m_{A}M^{A} \pm m_{\dot{A}}M^{\dot{A}} \ne 0$, $m_{A} N^{A} = -n_{\dot{A}}M^{\dot{A}} \ne 0$ & $\Gamma_{412} \pm \Gamma_{421} \ne 0$, $\Gamma_{324} = -\Gamma_{423} \ne 0$  \\ \hline
 $[++,++]$ & $m_{A}M^{A} \pm m_{\dot{A}}M^{\dot{A}} \ne 0$, $m_{A} N^{A} \pm n_{\dot{A}}M^{\dot{A}} \ne 0$ & $\Gamma_{412} \pm \Gamma_{421} \ne 0$, $\Gamma_{324} \pm \Gamma_{423} \ne 0$  \\ \hline 
\caption{Types of spaces equipped with one SD and two ASD congruences of null strings.}
\label{Tabela_typy_two_congruence}
\end{longtable}
}
%\end{center}
%\end{table}

%\begin{table}[ht]
%\begin{center}
%\setlength\LTleft{0pt}
%\setlength\LTright{0pt}
{ \footnotesize
\begin{longtable}{|c|c|c|}   \hline
%\begin{longtabu} to \textwidth {|c|c|c|}
 Type / &  \multicolumn{2}{|c|}{Conditions}  \\ \cline{2-3}
 Subtype  &   Spinorial   & Tetradial    \\ \hline 
  \rowcolor{light_gray_1}
 Type $ [ \: \cdot \: ]^{nn} \otimes [ \: \cdot \: ]^{nn}$ & $M^{\dot{A}}=N^{\dot{A}}=N^{A}=M^{A}=0$ & $\Gamma_{42}=\Gamma_{31}=\Gamma_{32}=\Gamma_{41}=0$  \\ \hline
$[--,--,--,--]$ & - & -  \\ \hline 
 \rowcolor{light_gray_1}
 Type $ [ \: \cdot \: ]^{nn} \otimes [ \: \cdot \: ]^{ne}$ & $M^{\dot{A}}=N^{\dot{A}}=M^{A} =0$, $N^{A} \ne 0$ & $\Gamma_{42}=\Gamma_{31}=\Gamma_{41} = 0$,        $\Gamma_{32} \ne 0$  \\ \hline
 $[--,++,--,++]$ & $m_{A}N^{A} \ne 0$, $n_{A} N^{A} \ne 0$ & $\Gamma_{324} \ne 0$, $\Gamma_{321} \ne 0$  \\ \hline
 $[--,--,--,++]$ & $m_{A}N^{A} = 0$, $n_{A} N^{A} \ne 0$ & $\Gamma_{324} = 0$, $\Gamma_{321} \ne 0$  \\ \hline
  \rowcolor{light_gray_1}
Type $ [ \: \cdot \: ]^{nn} \otimes [ \: \cdot \: ]^{ee}$ & $M^{\dot{A}}=N^{\dot{A}}=0$, $M^{A} \ne 0$, $N^{A} \ne 0$ & $\Gamma_{42}=\Gamma_{31}=0$, $\Gamma_{41} \ne 0$,        $\Gamma_{32} \ne 0$  \\ \hline
 $[--,--,++,++]$ & $m_{A} M^{A}=0$, $m_{A}N^{A} = 0$, $n_{A} M^{A} \ne 0$, $n_{A} N^{A} \ne 0$ & $\Gamma_{412}=0$, $\Gamma_{324}=0$, $\Gamma_{413} \ne 0$, $\Gamma_{321} \ne 0$  \\ \hline
 $[--,++,++,--]$ & $m_{A} M^{A}=0$, $m_{A}N^{A} \ne 0$, $n_{A} M^{A} \ne 0$, $n_{A} N^{A}=0$ & $\Gamma_{412}=0$, $\Gamma_{324} \ne 0$, $\Gamma_{413} \ne 0$, $\Gamma_{321} = 0$  \\ \hline
 $[--,++,++,++]$ & $m_{A} M^{A}=0$, $m_{A}N^{A} \ne 0$, $n_{A} M^{A} \ne 0$, $n_{A} N^{A} \ne 0$ & $\Gamma_{412}=0$, $\Gamma_{324} \ne 0$, $\Gamma_{413} \ne 0$, $\Gamma_{321} \ne 0$  \\ \hline
 $[++,++,++,++]$ & $m_{A} M^{A} \ne 0$, $m_{A}N^{A} \ne 0$, $n_{A} M^{A} \ne 0$, $n_{A} N^{A} \ne 0$ & $\Gamma_{412} \ne 0$, $\Gamma_{324} \ne 0$, $\Gamma_{413} \ne 0$, $\Gamma_{321} \ne 0$  \\ \hline 
 \rowcolor{light_gray_1}
 Type $ [ \: \cdot \: ]^{ne} \otimes [ \: \cdot \: ]^{ne}$ & $M^{\dot{A}}=0$, $N^{\dot{A}} \ne 0$, $M^{A} = 0$, $N^{A} \ne 0$ & $\Gamma_{42}=\Gamma_{41}=0$, $\Gamma_{31} \ne 0$,        $\Gamma_{32} \ne 0$  \\ \hline
 $[--,--,--,-+]$ & $m_{A} N^{A}=0$, $m_{\dot{A}}N^{\dot{A}} = 0$, $n_{A} N^{A}= -n_{\dot{A}} N^{\dot{A}} \ne 0$ & $\Gamma_{324}=0$, $\Gamma_{314}=0$, $\Gamma_{321} = -\Gamma_{312} \ne 0$  \\ \hline
 $[--,--,--,+-]$ & $m_{A} N^{A}=0$, $m_{\dot{A}}N^{\dot{A}} = 0$, $n_{A} N^{A}= n_{\dot{A}} N^{\dot{A}} \ne 0$ & $\Gamma_{324}=0$, $\Gamma_{314}=0$, $\Gamma_{321} = \Gamma_{312} \ne 0$  \\ \hline
 $[--,--,--,++]$ & $m_{A} N^{A}=0$, $m_{\dot{A}}N^{\dot{A}} = 0$, $n_{A} N^{A} \pm n_{\dot{A}} N^{\dot{A}} \ne 0$ & $\Gamma_{324}=0$, $\Gamma_{314}=0$, $\Gamma_{321} \pm \Gamma_{312} \ne 0$  \\ \hline
 $[--,--,++,-+]$ & $m_{A} N^{A}=0$, $m_{\dot{A}}N^{\dot{A}} \ne 0$, $n_{A} N^{A}= -n_{\dot{A}} N^{\dot{A}} \ne 0$ & $\Gamma_{324}=0$, $\Gamma_{314} \ne 0$, $\Gamma_{321} = -\Gamma_{312} \ne 0$  \\ \hline
 $[--,--,++,+-]$ & $m_{A} N^{A}=0$, $m_{\dot{A}}N^{\dot{A}} \ne 0$, $n_{A} N^{A}= n_{\dot{A}} N^{\dot{A}} \ne 0$ & $\Gamma_{324}=0$, $\Gamma_{314} \ne 0$, $\Gamma_{321} = \Gamma_{312} \ne 0$  \\ \hline
 $[--,--,++,++]$ & $m_{A} N^{A}=0$, $m_{\dot{A}}N^{\dot{A}} \ne 0$, $n_{A} N^{A} \pm n_{\dot{A}} N^{\dot{A}} \ne 0$ & $\Gamma_{324}=0$, $\Gamma_{314} \ne 0$, $\Gamma_{321} \pm \Gamma_{312} \ne 0$  \\ \hline
 $[--,++,++,--]$ & $m_{A} N^{A} \ne 0$, $m_{\dot{A}}N^{\dot{A}} \ne 0$, $n_{A} N^{A} =n_{\dot{A}} N^{\dot{A}} = 0$ & $\Gamma_{324} \ne 0$, $\Gamma_{314} \ne 0$, $\Gamma_{321} = \Gamma_{312} \ne 0$  \\ \hline
  $[--,++,++,-+]$ & $m_{A} N^{A} \ne 0$, $m_{\dot{A}}N^{\dot{A}} \ne 0$, $n_{A} N^{A}= -n_{\dot{A}} N^{\dot{A}} \ne 0$ & $\Gamma_{324} \ne 0$, $\Gamma_{314} \ne 0$, $\Gamma_{321} = -\Gamma_{312} \ne 0$  \\ \hline
 $[--,++,++,+-]$ & $m_{A} N^{A} \ne 0$, $m_{\dot{A}}N^{\dot{A}} \ne 0$, $n_{A} N^{A}= n_{\dot{A}} N^{\dot{A}} \ne 0$ & $\Gamma_{324} \ne 0$, $\Gamma_{314} \ne 0$, $\Gamma_{321} = \Gamma_{312} \ne 0$  \\ \hline
 $[--,++,++,++]$ & $m_{A} N^{A} \ne 0$, $m_{\dot{A}}N^{\dot{A}} \ne 0$, $n_{A} N^{A} \pm n_{\dot{A}} N^{\dot{A}} \ne 0$ & $\Gamma_{324} \ne 0$, $\Gamma_{314} \ne 0$, $\Gamma_{321} \pm \Gamma_{312} \ne 0$  \\ \hline 
 \rowcolor{light_gray_1}
 Type $ [ \: \cdot \: ]^{ne} \otimes [ \: \cdot \: ]^{ee}$ & $M^{\dot{A}}=0$, $N^{\dot{A}} \ne 0$, $M^{A} \ne 0$, $N^{A} \ne 0$ & $\Gamma_{42}=0$, $\Gamma_{41} \ne 0$, $\Gamma_{31} \ne 0$,        $\Gamma_{32} \ne 0$  \\ \hline
 \rowcolor{light_gray_2}
for all types below  & $m_{A} M^{A} = m_{A}N^{A} = 0$, & $\Gamma_{412} = \Gamma_{324}= 0$, \\ \hline
$[--,--,-+,-+]$  & $n_{A}M^{A}=-m_{\dot{A}} N^{\dot{A}} \ne 0$, $n_{A} N^{A} = - n_{\dot{A}} N^{\dot{A}} \ne 0$                 & $\Gamma_{413} = -\Gamma_{314} \ne 0$, $\Gamma_{321} = -\Gamma_{312} \ne 0$   \\ \hline             
$[--,--,-+,+-]$  & $n_{A}M^{A}=-m_{\dot{A}} N^{\dot{A}} \ne 0$, $n_{A} N^{A} =  n_{\dot{A}} N^{\dot{A}} \ne 0$ 
                 & $\Gamma_{413} = -\Gamma_{314} \ne 0$, $\Gamma_{321} = \Gamma_{312} \ne 0$     \\ \hline
$[--,--,-+,++]$  & $n_{A}M^{A}=-m_{\dot{A}} N^{\dot{A}} \ne 0$, $n_{A} N^{A} \pm n_{\dot{A}} N^{\dot{A}} \ne 0$      
                 & $\Gamma_{413} = -\Gamma_{314} \ne 0$, $\Gamma_{321} \pm \Gamma_{312} \ne 0$  \\ \hline     
$[--,--,+-,+-]$  & $n_{A}M^{A}=m_{\dot{A}} N^{\dot{A}} \ne 0$, $n_{A} N^{A} = n_{\dot{A}} N^{\dot{A}} \ne 0$ 
                 & $\Gamma_{413} = \Gamma_{314} \ne 0$, $\Gamma_{321} = \Gamma_{312} \ne 0$  \\ \hline 
$[--,--,+-,++]$  & $n_{A}M^{A}=m_{\dot{A}} N^{\dot{A}} \ne 0$, $n_{A} N^{A} \pm n_{\dot{A}} N^{\dot{A}} \ne 0$ 
                 &  $\Gamma_{413} = \Gamma_{314} \ne 0$, $\Gamma_{321} \pm \Gamma_{312} \ne 0$  \\ \hline $[--,--,++,++]$  & $n_{A}M^{A} \pm m_{\dot{A}} N^{\dot{A}} \ne 0$, $n_{A} N^{A} \pm n_{\dot{A}} N^{\dot{A}} \ne 0$ 
                 & $\Gamma_{413} \pm \Gamma_{314} \ne 0$, $\Gamma_{321} \pm \Gamma_{312} \ne 0$  \\ \hline    
 \rowcolor{light_gray_2}
for all types below  &  $m_{A} M^{A} = 0$, $m_{A}N^{A} \ne 0$,  & $\Gamma_{412} = 0$, $\Gamma_{324} \ne 0$, \\ \hline
$[--,++,-+,--]$ & $n_{A}M^{A} =- m_{\dot{A}} N^{\dot{A}} \ne 0$, $n_{A} N^{A} = n_{\dot{A}} N^{\dot{A}} = 0$ 
                & $\Gamma_{413} =- \Gamma_{314} \ne 0$, $\Gamma_{321} = \Gamma_{312} = 0$  \\ \hline  
$[--,++,-+,-+]$ & $n_{A}M^{A} =- m_{\dot{A}} N^{\dot{A}} \ne 0$, $n_{A} N^{A} = -n_{\dot{A}} N^{\dot{A}} \ne 0$ 
                & $\Gamma_{413} =- \Gamma_{314} \ne 0$, $\Gamma_{321} = -\Gamma_{312} \ne 0$  \\ \hline     
$[--,++,-+,+-]$ & $n_{A}M^{A} =- m_{\dot{A}} N^{\dot{A}} \ne 0$, $n_{A} N^{A} = n_{\dot{A}} N^{\dot{A}} \ne 0$
                & $\Gamma_{413} =- \Gamma_{314} \ne 0$, $\Gamma_{321} = \Gamma_{312} \ne 0$  \\ \hline  
$[--,++,-+,++]$ & $n_{A}M^{A} =- m_{\dot{A}} N^{\dot{A}} \ne 0$, $n_{A} N^{A} \pm n_{\dot{A}} N^{\dot{A}} \ne 0$
                & $\Gamma_{413} =- \Gamma_{314} \ne 0$, $\Gamma_{321} \pm \Gamma_{312} \ne 0$  \\ \hline     
$[--,++,+-,--]$ & $n_{A}M^{A} = m_{\dot{A}} N^{\dot{A}} \ne 0$, $n_{A} N^{A} = n_{\dot{A}} N^{\dot{A}} = 0$ 
                & $\Gamma_{413} = \Gamma_{314} \ne 0$, $\Gamma_{321} = \Gamma_{312} = 0$  \\ \hline  
$[--,++,+-,-+]$ & $n_{A}M^{A} = m_{\dot{A}} N^{\dot{A}} \ne 0$, $n_{A} N^{A} = -n_{\dot{A}} N^{\dot{A}} \ne 0$ 
                & $\Gamma_{413} = \Gamma_{314} \ne 0$, $\Gamma_{321} = -\Gamma_{312} \ne 0$  \\ \hline     
$[--,++,+-,+-]$ & $n_{A}M^{A} = m_{\dot{A}} N^{\dot{A}} \ne 0$, $n_{A} N^{A} = n_{\dot{A}} N^{\dot{A}} \ne 0$
                & $\Gamma_{413} = \Gamma_{314} \ne 0$, $\Gamma_{321} = \Gamma_{312} \ne 0$  \\ \hline  
$[--,++,+-,++]$ & $n_{A}M^{A} = m_{\dot{A}} N^{\dot{A}} \ne 0$, $n_{A} N^{A} \pm n_{\dot{A}} N^{\dot{A}} \ne 0$
                & $\Gamma_{413} = \Gamma_{314} \ne 0$, $\Gamma_{321} \pm \Gamma_{312} \ne 0$  \\ \hline   
$[--,++,++,--]$ & $n_{A}M^{A} \pm m_{\dot{A}} N^{\dot{A}} \ne 0$, $n_{A} N^{A} = n_{\dot{A}} N^{\dot{A}} = 0$ 
                & $\Gamma_{413} \pm \Gamma_{314} \ne 0$, $\Gamma_{321} = \Gamma_{312} = 0$  \\ \hline  
$[--,++,++,-+]$ & $n_{A}M^{A} \pm m_{\dot{A}} N^{\dot{A}} \ne 0$, $n_{A} N^{A} = -n_{\dot{A}} N^{\dot{A}} \ne 0$
                & $\Gamma_{413} \pm \Gamma_{314} \ne 0$, $\Gamma_{321} = -\Gamma_{312} \ne 0$  \\ \hline     
$[--,++,++,+-]$ & $n_{A}M^{A} \pm m_{\dot{A}} N^{\dot{A}} \ne 0$, $n_{A} N^{A} = n_{\dot{A}} N^{\dot{A}} \ne 0$ 
                & $\Gamma_{413} \pm \Gamma_{314} \ne 0$, $\Gamma_{321} = \Gamma_{312} \ne 0$  \\ \hline  
$[--,++,++,++]$ & $n_{A}M^{A} \pm m_{\dot{A}} N^{\dot{A}} \ne 0$, $n_{A} N^{A} \pm n_{\dot{A}} N^{\dot{A}} \ne 0$
                & $\Gamma_{413} \pm \Gamma_{314} \ne 0$, $\Gamma_{321} \pm \Gamma_{312} \ne 0$  \\ \hline  
\rowcolor{light_gray_2}
for all types below  & $m_{A} M^{A} \ne 0$, $m_{A}N^{A} \ne 0$, & $\Gamma_{412} \ne 0$, $\Gamma_{324} \ne 0$, \\ \hline
$[++,++,--,-+]$ &  $n_{A}M^{A} = m_{\dot{A}} N^{\dot{A}} = 0$, $n_{A} N^{A} = - n_{\dot{A}} N^{\dot{A}} \ne 0$                &  $\Gamma_{413} = \Gamma_{314} = 0$, $\Gamma_{321} = -\Gamma_{312} \ne 0$  \\ \hline      
$[++,++,--,+-]$ &  $n_{A}M^{A} = m_{\dot{A}} N^{\dot{A}} = 0$, $n_{A} N^{A} =  n_{\dot{A}} N^{\dot{A}} \ne 0$                &  $\Gamma_{413} = \Gamma_{314} = 0$, $\Gamma_{321} = \Gamma_{312} \ne 0$  \\ \hline          
$[++,++,--,++]$ &  $n_{A}M^{A} = m_{\dot{A}} N^{\dot{A}} = 0$, $n_{A} N^{A} \pm n_{\dot{A}} N^{\dot{A}} \ne 0$                &  $\Gamma_{413} = \Gamma_{314} = 0$, $\Gamma_{321} \pm \Gamma_{312} \ne 0$  \\ \hline
$[++,++,-+,-+]$ &  $n_{A}M^{A} = -m_{\dot{A}} N^{\dot{A}} \ne 0$, $n_{A} N^{A} = - n_{\dot{A}} N^{\dot{A}} \ne 0$                &  $\Gamma_{413} = -\Gamma_{314} \ne 0$, $\Gamma_{321} = -\Gamma_{312} \ne 0$  \\ \hline      
$[++,++,-+,+-]$ &  $n_{A}M^{A} = -m_{\dot{A}} N^{\dot{A}} \ne 0$, $n_{A} N^{A} =  n_{\dot{A}} N^{\dot{A}} \ne 0$                &  $\Gamma_{413} = -\Gamma_{314} \ne 0$, $\Gamma_{321} = \Gamma_{312} \ne 0$  \\ \hline          
$[++,++,-+,++]$ &  $n_{A}M^{A} = -m_{\dot{A}} N^{\dot{A}} \ne 0$, $n_{A} N^{A} \pm n_{\dot{A}} N^{\dot{A}} \ne 0$                &  $\Gamma_{413} = -\Gamma_{314} \ne 0$, $\Gamma_{321} \pm \Gamma_{312} \ne 0$  \\ \hline             
$[++,++,+-,+-]$ &  $n_{A}M^{A} = m_{\dot{A}} N^{\dot{A}} \ne 0$, $n_{A} N^{A} =  n_{\dot{A}} N^{\dot{A}} \ne 0$                &  $\Gamma_{413} = \Gamma_{314} \ne 0$, $\Gamma_{321} = \Gamma_{312} \ne 0$  \\ \hline      
$[++,++,+-,++]$ &  $n_{A}M^{A} = m_{\dot{A}} N^{\dot{A}} \ne 0$, $n_{A} N^{A} \pm  n_{\dot{A}} N^{\dot{A}} \ne 0$                &  $\Gamma_{413} = \Gamma_{314} \ne 0$, $\Gamma_{321} \pm \Gamma_{312} \ne 0$  \\ \hline          
$[++,++,++,++]$ &  $n_{A}M^{A} \pm m_{\dot{A}} N^{\dot{A}} \ne 0$, $n_{A} N^{A} \pm n_{\dot{A}} N^{\dot{A}} \ne 0$                &  $\Gamma_{413} \pm \Gamma_{314} \ne 0$, $\Gamma_{321} \pm \Gamma_{312} \ne 0$  \\ \hline
 \rowcolor{light_gray_1}
 Type $ [ \: \cdot \: ]^{ee} \otimes [ \: \cdot \: ]^{ee}$ & $M^{\dot{A}}\ne0$, $N^{\dot{A}} \ne 0$, $M^{A} \ne 0$, $N^{A} \ne 0$ & $\Gamma_{42}\ne0$, $\Gamma_{41} \ne 0$, $\Gamma_{31} \ne 0$,        $\Gamma_{32} \ne 0$  \\ \hline 
\rowcolor{light_gray_2}
for all types below  & $m_{A} M^{A} = m_{\dot{A}}M^{\dot{A}} =n_{A} N^{A} = n_{\dot{A}} N^{\dot{A}}= 0$, &
 $\Gamma_{412} = \Gamma_{421}=\Gamma_{321}=\Gamma_{312}= 0$, \\ \hline
$[--,-+,-+,--]$ 
            & $m_{A}N^{A} =- n_{\dot{A}} M^{\dot{A}} \ne 0$, $n_{A} M^{A} = - m_{\dot{A}} N^{\dot{A}} \ne 0$        
            & $\Gamma_{324} = -\Gamma_{423} \ne 0$, $\Gamma_{413} = -\Gamma_{314} \ne 0$  \\ \hline    
$[--,-+,+-,--]$ 
            & $m_{A}N^{A} =- n_{\dot{A}} M^{\dot{A}} \ne 0$, $n_{A} M^{A} =  m_{\dot{A}} N^{\dot{A}} \ne 0$        
            & $\Gamma_{324} = -\Gamma_{423} \ne 0$, $\Gamma_{413} = \Gamma_{314} \ne 0$  \\ \hline    
$[--,-+,++,--]$ 
            & $m_{A}N^{A} =- n_{\dot{A}} M^{\dot{A}} \ne 0$, $n_{A} M^{A} \pm m_{\dot{A}} N^{\dot{A}} \ne 0$        
            & $\Gamma_{324} = -\Gamma_{423} \ne 0$, $\Gamma_{413} \pm \Gamma_{314} \ne 0$  \\ \hline    
$[--,+-,+-,--]$ 
            & $m_{A}N^{A} = n_{\dot{A}} M^{\dot{A}} \ne 0$, $n_{A} M^{A} =  m_{\dot{A}} N^{\dot{A}} \ne 0$        
            & $\Gamma_{324} = \Gamma_{423} \ne 0$, $\Gamma_{413} = \Gamma_{314} \ne 0$  \\ \hline 
$[--,+-,++,--]$ 
            & $m_{A}N^{A} = n_{\dot{A}} M^{\dot{A}} \ne 0$, $n_{A} M^{A} \pm  m_{\dot{A}} N^{\dot{A}} \ne 0$        
            & $\Gamma_{324} = \Gamma_{423} \ne 0$, $\Gamma_{413} \pm \Gamma_{314} \ne 0$  \\ \hline 
$[--,++,++,--]$ 
            & $m_{A}N^{A} \pm n_{\dot{A}} M^{\dot{A}} \ne 0$, $n_{A} M^{A} \pm  m_{\dot{A}} N^{\dot{A}} \ne 0$        
            & $\Gamma_{324} \pm \Gamma_{423} \ne 0$, $\Gamma_{413} \pm \Gamma_{314} \ne 0$  \\ \hline   
\rowcolor{light_gray_2}
for all types below  & $m_{A} M^{A} = m_{\dot{A}}M^{\dot{A}} =0$, $n_{A} N^{A} = -n_{\dot{A}} N^{\dot{A}} \ne 0$, &
 $\Gamma_{412} = \Gamma_{421}=0$, $\Gamma_{321}=-\Gamma_{312} \ne 0$, \\ \hline
$[--,-+,-+,-+]$ 
            & $m_{A}N^{A} =- n_{\dot{A}} M^{\dot{A}} \ne 0$, $n_{A} M^{A} = - m_{\dot{A}} N^{\dot{A}} \ne 0$        
            & $\Gamma_{324} = -\Gamma_{423} \ne 0$, $\Gamma_{413} = -\Gamma_{314} \ne 0$  \\ \hline    
$[--,-+,+-,-+]$ 
            & $m_{A}N^{A} =- n_{\dot{A}} M^{\dot{A}} \ne 0$, $n_{A} M^{A} =  m_{\dot{A}} N^{\dot{A}} \ne 0$        
            & $\Gamma_{324} = -\Gamma_{423} \ne 0$, $\Gamma_{413} = \Gamma_{314} \ne 0$  \\ \hline    
$[--,-+,++,-+]$ 
            & $m_{A}N^{A} =- n_{\dot{A}} M^{\dot{A}} \ne 0$, $n_{A} M^{A} \pm m_{\dot{A}} N^{\dot{A}} \ne 0$        
            & $\Gamma_{324} = -\Gamma_{423} \ne 0$, $\Gamma_{413} \pm \Gamma_{314} \ne 0$  \\ \hline    
$[--,+-,+-,-+]$ 
            & $m_{A}N^{A} = n_{\dot{A}} M^{\dot{A}} \ne 0$, $n_{A} M^{A} =  m_{\dot{A}} N^{\dot{A}} \ne 0$        
            & $\Gamma_{324} = \Gamma_{423} \ne 0$, $\Gamma_{413} = \Gamma_{314} \ne 0$  \\ \hline 
$[--,+-,++,-+]$ 
            & $m_{A}N^{A} = n_{\dot{A}} M^{\dot{A}} \ne 0$, $n_{A} M^{A} \pm  m_{\dot{A}} N^{\dot{A}} \ne 0$        
            & $\Gamma_{324} = \Gamma_{423} \ne 0$, $\Gamma_{413} \pm \Gamma_{314} \ne 0$  \\ \hline 
$[--,++,++,-+]$ 
            & $m_{A}N^{A} \pm n_{\dot{A}} M^{\dot{A}} \ne 0$, $n_{A} M^{A} \pm  m_{\dot{A}} N^{\dot{A}} \ne 0$        
            & $\Gamma_{324} \pm \Gamma_{423} \ne 0$, $\Gamma_{413} \pm \Gamma_{314} \ne 0$  \\ \hline    
\rowcolor{light_gray_2}
for all types below  & $m_{A} M^{A} = m_{\dot{A}}M^{\dot{A}} =0$, $n_{A} N^{A} = n_{\dot{A}} N^{\dot{A}} \ne 0$, &
 $\Gamma_{412} = \Gamma_{421}=0$, $\Gamma_{321}=\Gamma_{312} \ne 0$, \\ \hline
$[--,-+,-+,+-]$ 
            & $m_{A}N^{A} =- n_{\dot{A}} M^{\dot{A}} \ne 0$, $n_{A} M^{A} = - m_{\dot{A}} N^{\dot{A}} \ne 0$        
            & $\Gamma_{324} = -\Gamma_{423} \ne 0$, $\Gamma_{413} = -\Gamma_{314} \ne 0$  \\ \hline    
$[--,-+,+-,+-]$ 
            & $m_{A}N^{A} =- n_{\dot{A}} M^{\dot{A}} \ne 0$, $n_{A} M^{A} =  m_{\dot{A}} N^{\dot{A}} \ne 0$        
            & $\Gamma_{324} = -\Gamma_{423} \ne 0$, $\Gamma_{413} = \Gamma_{314} \ne 0$  \\ \hline    
$[--,-+,++,+-]$ 
            & $m_{A}N^{A} =- n_{\dot{A}} M^{\dot{A}} \ne 0$, $n_{A} M^{A} \pm m_{\dot{A}} N^{\dot{A}} \ne 0$        
            & $\Gamma_{324} = -\Gamma_{423} \ne 0$, $\Gamma_{413} \pm \Gamma_{314} \ne 0$  \\ \hline    
$[--,+-,+-,+-]$ 
            & $m_{A}N^{A} = n_{\dot{A}} M^{\dot{A}} \ne 0$, $n_{A} M^{A} =  m_{\dot{A}} N^{\dot{A}} \ne 0$        
            & $\Gamma_{324} = \Gamma_{423} \ne 0$, $\Gamma_{413} = \Gamma_{314} \ne 0$  \\ \hline 
$[--,+-,++,+-]$ 
            & $m_{A}N^{A} = n_{\dot{A}} M^{\dot{A}} \ne 0$, $n_{A} M^{A} \pm  m_{\dot{A}} N^{\dot{A}} \ne 0$        
            & $\Gamma_{324} = \Gamma_{423} \ne 0$, $\Gamma_{413} \pm \Gamma_{314} \ne 0$  \\ \hline 
$[--,++,++,+-]$ 
            & $m_{A}N^{A} \pm n_{\dot{A}} M^{\dot{A}} \ne 0$, $n_{A} M^{A} \pm  m_{\dot{A}} N^{\dot{A}} \ne 0$        
            & $\Gamma_{324} \pm \Gamma_{423} \ne 0$, $\Gamma_{413} \pm \Gamma_{314} \ne 0$  \\ \hline    
\rowcolor{light_gray_2}
for all types below  & $m_{A} M^{A} = m_{\dot{A}}M^{\dot{A}} =0$, $n_{A} N^{A} \pm n_{\dot{A}} N^{\dot{A}} \ne 0$, &
 $\Gamma_{412} = \Gamma_{421}=0$, $\Gamma_{321} \pm \Gamma_{312} \ne 0$, \\ \hline
$[--,-+,-+,++]$ 
            & $m_{A}N^{A} =- n_{\dot{A}} M^{\dot{A}} \ne 0$, $n_{A} M^{A} = - m_{\dot{A}} N^{\dot{A}} \ne 0$        
            & $\Gamma_{324} = -\Gamma_{423} \ne 0$, $\Gamma_{413} = -\Gamma_{314} \ne 0$  \\ \hline    
$[--,-+,+-,++]$ 
            & $m_{A}N^{A} =- n_{\dot{A}} M^{\dot{A}} \ne 0$, $n_{A} M^{A} =  m_{\dot{A}} N^{\dot{A}} \ne 0$        
            & $\Gamma_{324} = -\Gamma_{423} \ne 0$, $\Gamma_{413} = \Gamma_{314} \ne 0$  \\ \hline    
$[--,-+,++,++]$ 
            & $m_{A}N^{A} =- n_{\dot{A}} M^{\dot{A}} \ne 0$, $n_{A} M^{A} \pm m_{\dot{A}} N^{\dot{A}} \ne 0$        
            & $\Gamma_{324} = -\Gamma_{423} \ne 0$, $\Gamma_{413} \pm \Gamma_{314} \ne 0$  \\ \hline    
$[--,+-,+-,++]$ 
            & $m_{A}N^{A} = n_{\dot{A}} M^{\dot{A}} \ne 0$, $n_{A} M^{A} =  m_{\dot{A}} N^{\dot{A}} \ne 0$        
            & $\Gamma_{324} = \Gamma_{423} \ne 0$, $\Gamma_{413} = \Gamma_{314} \ne 0$  \\ \hline 
$[--,+-,++,++]$ 
            & $m_{A}N^{A} = n_{\dot{A}} M^{\dot{A}} \ne 0$, $n_{A} M^{A} \pm  m_{\dot{A}} N^{\dot{A}} \ne 0$        
            & $\Gamma_{324} = \Gamma_{423} \ne 0$, $\Gamma_{413} \pm \Gamma_{314} \ne 0$  \\ \hline 
$[--,++,++,++]$ 
            & $m_{A}N^{A} \pm n_{\dot{A}} M^{\dot{A}} \ne 0$, $n_{A} M^{A} \pm  m_{\dot{A}} N^{\dot{A}} \ne 0$        
            & $\Gamma_{324} \pm \Gamma_{423} \ne 0$, $\Gamma_{413} \pm \Gamma_{314} \ne 0$  \\ \hline  
\rowcolor{light_gray_2}
for all types below  & $m_{A} M^{A} = -m_{\dot{A}}M^{\dot{A}} \ne 0$, $n_{A} N^{A} = -n_{\dot{A}} N^{\dot{A}} \ne 0$, &
 $\Gamma_{412} = -\Gamma_{421} \ne 0$, $\Gamma_{321}=-\Gamma_{312} \ne 0$, \\ \hline
$[-+,-+,-+,-+]$ 
            & $m_{A}N^{A} =- n_{\dot{A}} M^{\dot{A}} \ne 0$, $n_{A} M^{A} = - m_{\dot{A}} N^{\dot{A}} \ne 0$        
            & $\Gamma_{324} = -\Gamma_{423} \ne 0$, $\Gamma_{413} = -\Gamma_{314} \ne 0$  \\ \hline    
$[-+,-+,+-,-+]$ 
            & $m_{A}N^{A} =- n_{\dot{A}} M^{\dot{A}} \ne 0$, $n_{A} M^{A} =  m_{\dot{A}} N^{\dot{A}} \ne 0$        
            & $\Gamma_{324} = -\Gamma_{423} \ne 0$, $\Gamma_{413} = \Gamma_{314} \ne 0$  \\ \hline    
$[-+,-+,++,-+]$ 
            & $m_{A}N^{A} =- n_{\dot{A}} M^{\dot{A}} \ne 0$, $n_{A} M^{A} \pm m_{\dot{A}} N^{\dot{A}} \ne 0$        
            & $\Gamma_{324} = -\Gamma_{423} \ne 0$, $\Gamma_{413} \pm \Gamma_{314} \ne 0$  \\ \hline    
$[-+,+-,+-,-+]$ 
            & $m_{A}N^{A} = n_{\dot{A}} M^{\dot{A}} \ne 0$, $n_{A} M^{A} =  m_{\dot{A}} N^{\dot{A}} \ne 0$        
            & $\Gamma_{324} = \Gamma_{423} \ne 0$, $\Gamma_{413} = \Gamma_{314} \ne 0$  \\ \hline 
$[-+,+-,++,-+]$ 
            & $m_{A}N^{A} = n_{\dot{A}} M^{\dot{A}} \ne 0$, $n_{A} M^{A} \pm  m_{\dot{A}} N^{\dot{A}} \ne 0$        
            & $\Gamma_{324} = \Gamma_{423} \ne 0$, $\Gamma_{413} \pm \Gamma_{314} \ne 0$  \\ \hline 
$[-+,++,++,-+]$ 
            & $m_{A}N^{A} \pm n_{\dot{A}} M^{\dot{A}} \ne 0$, $n_{A} M^{A} \pm  m_{\dot{A}} N^{\dot{A}} \ne 0$        
            & $\Gamma_{324} \pm \Gamma_{423} \ne 0$, $\Gamma_{413} \pm \Gamma_{314} \ne 0$  \\ \hline  
\rowcolor{light_gray_2}
for all types below  & $m_{A} M^{A} = -m_{\dot{A}}M^{\dot{A}} \ne 0$, $n_{A} N^{A} = n_{\dot{A}} N^{\dot{A}} \ne 0$, &
 $\Gamma_{412} = -\Gamma_{421} \ne 0$, $\Gamma_{321}=\Gamma_{312} \ne 0$, \\ \hline
$[-+,-+,+-,+-]$ 
            & $m_{A}N^{A} =- n_{\dot{A}} M^{\dot{A}} \ne 0$, $n_{A} M^{A} =  m_{\dot{A}} N^{\dot{A}} \ne 0$        
            & $\Gamma_{324} = -\Gamma_{423} \ne 0$, $\Gamma_{413} = \Gamma_{314} \ne 0$  \\ \hline    
$[-+,-+,++,+-]$ 
            & $m_{A}N^{A} =- n_{\dot{A}} M^{\dot{A}} \ne 0$, $n_{A} M^{A} \pm m_{\dot{A}} N^{\dot{A}} \ne 0$        
            & $\Gamma_{324} = -\Gamma_{423} \ne 0$, $\Gamma_{413} \pm \Gamma_{314} \ne 0$  \\ \hline    
$[-+,+-,+-,+-]$ 
            & $m_{A}N^{A} = n_{\dot{A}} M^{\dot{A}} \ne 0$, $n_{A} M^{A} =  m_{\dot{A}} N^{\dot{A}} \ne 0$        
            & $\Gamma_{324} = \Gamma_{423} \ne 0$, $\Gamma_{413} = \Gamma_{314} \ne 0$  \\ \hline 
$[-+,+-,++,+-]$ 
            & $m_{A}N^{A} = n_{\dot{A}} M^{\dot{A}} \ne 0$, $n_{A} M^{A} \pm  m_{\dot{A}} N^{\dot{A}} \ne 0$        
            & $\Gamma_{324} = \Gamma_{423} \ne 0$, $\Gamma_{413} \pm \Gamma_{314} \ne 0$  \\ \hline 
$[-+,++,++,+-]$ 
            & $m_{A}N^{A} \pm n_{\dot{A}} M^{\dot{A}} \ne 0$, $n_{A} M^{A} \pm  m_{\dot{A}} N^{\dot{A}} \ne 0$        
            & $\Gamma_{324} \pm \Gamma_{423} \ne 0$, $\Gamma_{413} \pm \Gamma_{314} \ne 0$  \\ \hline  
\rowcolor{light_gray_2}
for all types below  & $m_{A} M^{A} = -m_{\dot{A}}M^{\dot{A}} \ne 0$, $n_{A} N^{A}  \pm n_{\dot{A}} N^{\dot{A}} \ne 0$, &$\Gamma_{412} = -\Gamma_{421} \ne 0$, $\Gamma_{321} \pm \Gamma_{312} \ne 0$, \\ \hline
$[-+,-+,++,++]$ 
            & $m_{A}N^{A} =- n_{\dot{A}} M^{\dot{A}} \ne 0$, $n_{A} M^{A} \pm m_{\dot{A}} N^{\dot{A}} \ne 0$        
            & $\Gamma_{324} = -\Gamma_{423} \ne 0$, $\Gamma_{413} \pm \Gamma_{314} \ne 0$  \\ \hline    
$[-+,+-,+-,++]$ 
            & $m_{A}N^{A} = n_{\dot{A}} M^{\dot{A}} \ne 0$, $n_{A} M^{A} =  m_{\dot{A}} N^{\dot{A}} \ne 0$        
            & $\Gamma_{324} = \Gamma_{423} \ne 0$, $\Gamma_{413} = \Gamma_{314} \ne 0$  \\ \hline 
$[-+,+-,++,++]$ 
            & $m_{A}N^{A} = n_{\dot{A}} M^{\dot{A}} \ne 0$, $n_{A} M^{A} \pm  m_{\dot{A}} N^{\dot{A}} \ne 0$        
            & $\Gamma_{324} = \Gamma_{423} \ne 0$, $\Gamma_{413} \pm \Gamma_{314} \ne 0$  \\ \hline 
$[-+,++,++,++]$ 
            & $m_{A}N^{A} \pm n_{\dot{A}} M^{\dot{A}} \ne 0$, $n_{A} M^{A} \pm  m_{\dot{A}} N^{\dot{A}} \ne 0$        
            & $\Gamma_{324} \pm \Gamma_{423} \ne 0$, $\Gamma_{413} \pm \Gamma_{314} \ne 0$  \\ \hline  
\rowcolor{light_gray_2}
for all types below  & $m_{A} M^{A} = m_{\dot{A}}M^{\dot{A}} \ne 0$, $n_{A} N^{A}  = n_{\dot{A}} N^{\dot{A}} \ne 0$, &$\Gamma_{412} = \Gamma_{421} \ne 0$, $\Gamma_{321} = \Gamma_{312} \ne 0$, \\ \hline
$[+-,+-,+-,+-]$ 
            & $m_{A}N^{A} = n_{\dot{A}} M^{\dot{A}} \ne 0$, $n_{A} M^{A} =  m_{\dot{A}} N^{\dot{A}} \ne 0$        
            & $\Gamma_{324} = \Gamma_{423} \ne 0$, $\Gamma_{413} = \Gamma_{314} \ne 0$  \\ \hline 
$[+-,+-,++,+-]$ 
            & $m_{A}N^{A} = n_{\dot{A}} M^{\dot{A}} \ne 0$, $n_{A} M^{A} \pm  m_{\dot{A}} N^{\dot{A}} \ne 0$        
            & $\Gamma_{324} = \Gamma_{423} \ne 0$, $\Gamma_{413} \pm \Gamma_{314} \ne 0$  \\ \hline 
$[+-,++,++,+-]$ 
            & $m_{A}N^{A} \pm n_{\dot{A}} M^{\dot{A}} \ne 0$, $n_{A} M^{A} \pm  m_{\dot{A}} N^{\dot{A}} \ne 0$        
            & $\Gamma_{324} \pm \Gamma_{423} \ne 0$, $\Gamma_{413} \pm \Gamma_{314} \ne 0$  \\ \hline   
\rowcolor{light_gray_2}              
for all types below  & $m_{A} M^{A} = m_{\dot{A}}M^{\dot{A}} \ne 0$, $n_{A} N^{A}  \pm n_{\dot{A}} N^{\dot{A}} \ne 0$, &$\Gamma_{412} = \Gamma_{421} \ne 0$, $\Gamma_{321} \pm \Gamma_{312} \ne 0$, \\ \hline
$[+-,+-,++,++]$ 
            & $m_{A}N^{A} = n_{\dot{A}} M^{\dot{A}} \ne 0$, $n_{A} M^{A} \pm  m_{\dot{A}} N^{\dot{A}} \ne 0$        
            & $\Gamma_{324} = \Gamma_{423} \ne 0$, $\Gamma_{413} \pm \Gamma_{314} \ne 0$  \\ \hline 
$[+-,++,++,++]$ 
            & $m_{A}N^{A} \pm n_{\dot{A}} M^{\dot{A}} \ne 0$, $n_{A} M^{A} \pm  m_{\dot{A}} N^{\dot{A}} \ne 0$        
            & $\Gamma_{324} \pm \Gamma_{423} \ne 0$, $\Gamma_{413} \pm \Gamma_{314} \ne 0$  \\ \hline   
\rowcolor{light_gray_2}              
for all types below  & $m_{A} M^{A} \pm m_{\dot{A}}M^{\dot{A}} \ne 0$, $n_{A} N^{A}  \pm n_{\dot{A}} N^{\dot{A}} \ne 0$, &$\Gamma_{412} \pm \Gamma_{421} \ne 0$, $\Gamma_{321} \pm \Gamma_{312} \ne 0$, \\ \hline
$[++,++,++,++]$ 
            & $m_{A}N^{A} \pm n_{\dot{A}} M^{\dot{A}} \ne 0$, $n_{A} M^{A} \pm  m_{\dot{A}} N^{\dot{A}} \ne 0$        
            & $\Gamma_{324} \pm \Gamma_{423} \ne 0$, $\Gamma_{413} \pm \Gamma_{314} \ne 0$  \\ \hline           
\caption{Types of spaces equipped with two SD and two ASD congruences of null strings.}
\label{Tabela_typy_four_congruence}
\end{longtable}
}
%\end{center}
%\end{table}

%#####################################################################################

\end{document}